\documentclass{article}
\usepackage[margin=1in]{geometry}
\usepackage{graphicx} \usepackage{amsmath, amssymb, amsthm}
\usepackage{physics, complexity, mathtools, booktabs}
\usepackage{tikz-cd}
\usepackage{stmaryrd}
\usepackage{mleftright}
\usepackage{subcaption}
\usepackage{appendix}
\usepackage{hyperref}
\usepackage{authblk}
\usepackage{aliascnt}
\usepackage[nameinlink,capitalize,noabbrev]{cleveref}
\usepackage{xspace}
\usepackage{enumerate}

\usepackage{xcolor} \usepackage{ninecolors}\NineColors{saturation=high}

\usepackage[
    backend=biber,
    style=alphabetic, 
    sorting=none
]{biblatex}
\hypersetup{
    colorlinks=true,
    linkcolor=blue3,
    linkbordercolor=blue3,
    filecolor=blue3,      
    urlcolor=black,
    urlbordercolor=cyan,
    citecolor=blue3,
    citebordercolor=cyan,
    pdfpagemode=FullScreen,
}
\usepackage[useregional]{datetime2}
\usepackage{datetime}
\newdateformat{monthdayyeardate}{\monthname[\THEMONTH]~\THEDAY, \THEYEAR}

\title{Natural proofs for quantum state preparation lower bounds}

\date{\monthdayyeardate\today}
\date{}

\newtheorem{theorem}{Theorem}[section]

\newtheorem{corollary}{Corollary}[theorem]

\newaliascnt{lemma}{theorem}
\newtheorem{lemma}[lemma]{Lemma}
\aliascntresetthe{lemma}

\crefname{lemma}{lemma}{lemmas}
\Crefname{lemma}{Lemma}{Lemmas}

\newtheorem*{remark}{Remark}

\theoremstyle{definition}
\newtheorem{definition}{Definition}[section]

\newtheorem{claim}{Claim}[section]
\crefname{claim}{Claim}{Claims}

\newtheorem{proposition}{Proposition}[section]

\newcommand{\mc}{\mathcal}

\newcommand{\mbb}{\mathbb}
\newcommand{\mr}{\mathrm}

\newcommand{\Cliff}{\textsf{Cliff}}
\newcommand{\AND}{\textsf{AND}}

\newcommand{\MH}{\textsf{MH}}

\newcommand{\AxCQ}[1]{\textsf{A}_{#1}\textsf{CQ}}

\newcommand{\PRF}{\textsf{PRF}}

\newcommand{\Samp}{\mathsf{Samp}}

\newcommand{\tnorm}[1]{\norm{#1}_1}
\newcommand{\enc}{\operatorname{enc}}
\newcommand{\Unif}{\operatorname{Unif}}
\newcommand{\Pois}{\operatorname{Pois}}

\newcommand{\eps}{\varepsilon}
\renewcommand{\epsilon}{\eps}

\newcommand{\wh}[1]{\widehat{#1}}
\newcommand{\sse}{\subseteq}

\newcommand{\zo}{\{0,1\}}

\newcommand{\sel}{\mathsf{sel}}
\renewcommand{\succ}{\mathsf{succ}}

\newcommand{\Expect}{\mathbb{E}} \newcommand{\ZZ}{\mathbb{Z}}
\newcommand{\NN}{\mathbb{N}}
\newcommand{\CC}{\mathbb{C}}

\newcommand{\RR}{\mathbb{R}}

\newcommand{\calB}{\mathcal{B}}
\newcommand{\calN}{\mathcal{N}}

\newcommand{\calP}{\mathcal{P}}

\newcommand{\calG}{\mathcal{G}}
\newcommand{\calD}{\mathcal{D}}

\newcommand{\calC}{\mathcal{C}}

\renewcommand{\abs}[1]{{\left| #1 \right|}}

\newcommand{\abra}[1]{{\left\langle #1 \right\rangle}}
\newcommand{\pbra}[1]{{\left( #1 \right)}}
\newcommand{\cbra}[1]{{\left\{ #1 \right\}}}
\newcommand{\sbra}[1]{{\left[ #1 \right]}}
\DeclarePairedDelimiter\ceil{\lceil}{\rceil}
\DeclarePairedDelimiter\floor{\lfloor}{\rfloor}

\newif\ifcommentsfootnote
\commentsfootnotetrue

\newif\ifcommentsinline
\commentsinlinetrue
\definecolor{cb-salmon-pink}{RGB}{255, 182, 119}
\ifcommentsinline
    \newcommand{\nat}[1]{\textcolor{blue}{\textbf{Nat: }}\textcolor{gray}{#1}}
    \newcommand{\hnote}[1]{\textcolor{red}{[\textbf{HY: #1}]}}
\else
    \newcommand{\nat}[1]{}
    \newcommand{\hnote}[1]{}
\fi

\ifcommentsfootnote
    \newcommand{\natf}[1]{\footnote{{\bf \color{violet}Nat:} {#1}}}
\else
    \newcommand{\natf}[1]{}
\fi

\newcommand{\TCZ}{\mathsf{TC}^0}

\newcommand{\amplistwords}{amplitude list\xspace}
\newcommand{\amplistswords}{amplitude lists\xspace}
 \newcommand{\al}{\mathsf{al}}

\newcommand{\matdesc}{\mathsf{mat}}
\newcommand{\Haar}{\mathrm{Haar}}
\newcommand{\PRAL}{PRAL\xspace}

\newcommand{\pralwords}{pseudorandom amplitude-list\xspace}
\newcommand{\Pralwords}{Pseudorandom amplitude-list\xspace}  

\newcommand{\TVD}{\mathsf{TVD}}

\usepackage{tcolorbox}
\usepackage{tikz}

\tcbset{
  problembox/.style={
    colback=black!1,
    colframe=black!45,
    boxrule=0.45pt,
    arc=1pt,
    left=6pt,
    right=6pt,
    top=5pt,
    bottom=5pt,
    before skip=10pt,
    after skip=10pt
  }
}

\begin{document}

\author{Christine Li}
\author{Natalie Parham}
\affil{Columbia University}
\maketitle

\begin{abstract}
We identify a barrier that helps explain why proving stronger quantum state-preparation lower bounds has been so difficult. 
In particular, we establish a quantum analogue of the Razborov–Rudich natural proofs barrier for state-preparation lower bounds.
We call a property of quantum states \emph{natural} if it holds for a sufficiently large fraction of Haar-random states and can be efficiently tested when given all of the state’s amplitudes. 
Under a standard cryptographic assumption, we show that no natural property can prove superpolynomial state-preparation lower bounds even against a fixed level of the Magic Hierarchy. 
We show that several existing state-preparation lower-bound techniques are natural in our sense, including arguments based on approximate Pauli degree, obstructions to being the unique ground state of a local Hamiltonian, and mutual information. 
\end{abstract}

\hypersetup{linktocpage}
\setcounter{tocdepth}{2}
{
\small
\tableofcontents
}
\newpage

\section{Introduction}

Understanding the limitations of quantum circuits is a fundamental goal in quantum computation. For state preparation, a circuit lower bound shows that a particular quantum state cannot be prepared within specified resource constraints, like time or space. Such bounds help identify which features of quantum states genuinely require complex quantum dynamics, rather than merely reflecting limitations of currently known preparation algorithms. These questions also connect to quantum many-body physics, where circuit complexity provides one way to classify phases of matter and patterns of entanglement~\cite{chen2010local}.

Despite their importance, quantum circuit lower bounds are notoriously difficult to prove. One of the most basic and successful techniques is the \emph{lightcone argument}. In a shallow circuit of two-qubit gates, information cannot spread throughout the entire system, restricting the correlations that the circuit can create from an initially unentangled state. For example, a constant-depth circuit of two-qubit gates cannot prepare the CAT state $\frac{1}{\sqrt{2}}(\ket{0^n} + \ket{1^n})$, since no qubit can be correlated with the entire system. But this argument is no longer effective once circuits are allowed enough depth or global operations. Beyond this regime, we have far fewer general techniques to certify that a state is difficult to prepare.

\paragraph{Why has it been so difficult to prove stronger lower bounds?}
 One clue comes from \textit{classical} circuit complexity.
Rosenthal showed that state-preparation lower bounds against constant-depth circuits with many-qubit fanout gates would imply breakthrough \emph{classical} circuit lower bounds ~\cite{rosenthal2024efficient}.
In particular, this would prove that some explicit Boolean function cannot be implemented in the circuit model $\TCZ$, a classical circuit class for which no superpolynomial lower bounds for explicit Boolean functions are known, even for depth 4.

Not only has proving $\TCZ$ lower bounds been a longstanding open problem for decades, but also classical complexity theory gives evidence that a broad class of natural approaches will not suffice. 
In their seminal work, Razborov and Rudich established the \emph{natural proofs barrier}, a ``meta-theorem'' showing that, under a cryptographic assumption, so called \emph{natural} proof techniques cannot establish lower bounds against sufficiently powerful classical circuits \cite{razborov1997naturalproofs}.
Under suitable cryptographic assumptions, this obstruction can apply even to $\TCZ$ circuits. 
Their result is unusual in that it places a limitation not on the circuits being studied, but on the \emph{kinds of arguments} that can prove lower bounds against them.

In this work we ask: \emph{are there analogous limitations on how we can prove quantum lower bounds?}
Their connection to classical circuit complexity discussed above suggest that strong quantum circuit lower bounds should be difficult to prove, 
but this does not help us to identify proof approaches for circuit lower bounds that are themselves subject to cryptographic barriers.
In this paper, we establish a large class of general proof approaches that cannot be useful for proving stronger state preparation lower bounds without breaking some cryptographic assumption. 

A general proof strategy is to first identify a property of the target state and then show that no state prepared by the allowed circuits has that property.
Following the style of Razborov and Rudich, we informally call a property of quantum states \emph{natural} if it satisfies the following two conditions:
\begin{itemize}
\item \emph{Constructivity:} Given a list of all the amplitudes of a state up to sufficiently high precision, a classical algorithm can decide efficiently whether the state has the property. Here, efficient means polynomial in the length of the \amplistwords (and thus may be exponential in the number of qubits).
    \item \emph{Largeness:} The property holds for a sufficiently large fraction of Haar-random states.
\end{itemize}
Finally, we call a property \emph{useful} against a circuit model $\calC$ if no state approximately preparable by $\calC$ has the property. We say that a lower-bound proof against $\calC$ is \emph{natural} if it can be recast as identifying a natural property that is useful against $\calC$.

We emphasize that constructivity is meant to capture a classical mathematical test of a state property, rather than an efficient physical procedure for verifying the property. The amplitude list is analogous to the truth table of a Boolean function in the formulation of classical natural proofs. This is a deliberately permissive notion: the algorithm may inspect all $2^n$ amplitudes and run in time polynomial in the list length.

\begin{tcolorbox}[problembox]
\textbf{Example: Pauli expectations.} 
As an example, we describe a lower bound against Clifford circuits that is a natural proof. 
Consider states with the property that the magnitudes of all nonidentity Pauli expectations are bounded away from one: $\max_{P \neq I} \abs{\bra{\psi} P \ket{\psi}}< 0.9$. This condition can be tested from the \amplistwords in polynomial time (constructivity) and holds for Haar-random states with high probability (largeness). 
In contrast, every state prepared by a Clifford circuit has many nonidentity Pauli expectations with absolute value one (usefulness). 
This example shows how a natural property can certify hardness for a restricted state preparation model: Clifford circuits. 
\end{tcolorbox}

\paragraph{Our result} 
Our main result places limits on the extent to which these natural-proof approaches can establish lower bounds against more powerful quantum circuit models.
Our barrier already applies to a restricted circuit model:
the \emph{magic hierarchy}, whose circuits alternate between
arbitrary Clifford circuits and constant-depth circuits of
one- and two-qubit gates~\cite{parham2025mh}. Its levels bound
the number of alternations.

\begin{theorem}[Informal, natural proofs barrier against state preparation, see \Cref{thm:main:natpfsMH:body}]
    Assuming that classically subexponentially secure pseudorandom functions\footnote{We say a family of functions with security parameter $\lambda$ is subexponentially-secure if there exists a constant $\alpha$ such that any distinguishing algorithm of size $2^{\lambda^\alpha}$ cannot distinguish the function from truly random with advantage better than $2^{-\lambda^\alpha}$.} can be implemented by polynomial-size, depth-$d$ threshold circuits, then no natural property can be used to prove that a state cannot be $e^{-n}$-approximately prepared by polynomial-size circuits in the $(8d+20)$th level of the magic hierarchy.
\end{theorem}

We previously already knew that  state-preparation lower bounds
at fixed levels of the magic hierarchy imply lower bounds for
fixed-depth classical threshold circuits
\cite{rosenthal2024efficient,parham2025mh}, but this does not by itself imply our theorem. 
\cite{rosenthal2024efficient,parham2025mh} do not show any connection between any notion of a natural property of quantum states to the a natural property of Boolean functions. 
Our theorem avoids the need for such a correspondence by
placing the obstruction directly on properties of quantum states.

We emphasize that our cryptographic assumption is \emph{entirely classical}. It is also the same assumption used to get the \textit{classical} natural proofs barrier against $\TCZ$ circuits. 
A concrete  psuedorandom function (\PRF) construction makes the level of our barrier explicit. The construction of Naor and Reingold~\cite{naor2004number} admits polynomial-size, depth-four threshold circuits, as shown by Krause and Lucks~\cite{krause2002prfsintc0}. Assuming this construction is classically subexponentially secure\footnote{This security follows from a sufficiently strong subexponential version of the decisional Diffie-Hellman assumption, via the quantitative security reduction of Naor and Reingold~\cite{naor2004number}. This Diffie-Hellman assumption is actually secure against quantum adversaries, but we can still use it because we only need classical security.}
, we obtain the following corollary.
\begin{corollary}[Informal]
    If the Naor-Reingold PRF construction~\cite{naor2004number} is classically subexponentially secure, then no natural property is useful against polynomial-sized circuits in the $52$nd level of the Magic Hierarchy for state preparation. 
\end{corollary}
For reference, \cite{parham2025mh} already showed that a lower bound against level $64$ of the magic hierarchy implies a lower bound against depth-4 $\TCZ$ circuits, which themselves encounter the natural proofs barrier . 

\paragraph{Unitary implementation.}
We also formulate natural proofs for unitary implementation,
replacing the amplitude list by a finite-precision description
of the full unitary matrix. Pseudorandom matrix descriptions
supported on efficiently implementable unitaries would yield
an analogous barrier
(\Cref{prop:unitary-natural-barrier}), but we do not obtain such a cryptographic construction here.

There is, however, a concrete reason that directly extending
our state construction is nontrivial.
Our proof first uses
a random function to produce statistically Haar-like amplitude
lists, and then replaces it by a PRF. Using the argument of Brakerski and Yuen \cite{brakerski2026scalablePRU}, we show that reproducing
this statistical stage for full unitary matrices, with efficient
Boolean-oracle implementations, would already resolve the
Aaronson-Kuperberg unitary synthesis problem
(\Cref{prop:statistical-matrices-synthesis}).

\subsection{Technical overview}
To obtain a natural proofs barrier, the key is to construct ensembles of states that admit small approximate preparation circuits, but whose complete finite-precision \amplistswords look Haar-random to the allowed classical algorithms. A useful natural property would reject these states, while largeness would make it accept a sufficiently large fraction of Haar-random states. Its recognition algorithm from constructivity would therefore distinguish the two ensembles, contradicting pseudorandomness.

\begin{definition}[\Pralwords (\PRAL) family, informal, see \Cref{def:PRAL}]
    A \emph{\pralwords (\PRAL) family} consists of pure-state ensembles that look Haar random even to a classical algorithm given the entire finite-precision \amplistwords and time polynomial in its length. The family supports any polynomially bounded amplitude precision.

\end{definition}
We call such a family $\eps$-approximately $\calC$ preparable if each of its states can be prepared to error $\eps$ by a polynomial-size circuit in $\calC$. Note that this requires a small preparation circuit for each state, not an efficient procedure for sampling the ensemble.

\begin{theorem}[Informal, see \Cref{thm:quantum-natural-proofs-barrier}]
    If a \PRAL family can be $\eps$-approximately prepared by polynomial-size circuits in $\calC$, then no natural property is $\eps$-useful against $\calC$ for state preparation. 
\end{theorem}

The main technical task is therefore to construct a \PRAL family whose states can be approximately prepared within a fixed level of the magic hierarchy. 
Standard pseudorandom-state constructions do not immediately suffice for this purpose. 
For example, consider the binary phase state, $\frac{1}{\sqrt{2^n}} \sum_{x\in \zo^n} (-1)^{f_k(x)}\ket{x}$, for a pseudorandom function $f_k:\zo^n \to \zo$ with a randomly chosen key $k$. Assuming the function family is secure against quantum queries, this ensemble is indistinguishable from Haar random states to polynomial-time quantum adversaries given polynomially many copies~\cite{brakerski2019binary}.  However, the security we want here is against adversaries who can look at the entire \textit{amplitude list}. In this case, it is very easy to see that this state is not Haar random because all of its amplitudes have the same magnitude $2^{-n/2}$, whereas Haar random states have varying amplitude magnitudes.

We therefore follow the Gaussian-based framework of Brakerski and Shmueli~\cite{brakerski2020scalable}, generating approximate complex Gaussian entries and normalizing the resulting vector. Our task is to obtain the required pseudorandomness of the complete \amplistwords while keeping the preparation circuit
within a fixed level of the magic hierarchy.

\paragraph{1. Random amplitudes from a random function.} 
A Haar random state can be generated by choosing independent complex Gaussian amplitudes and normalizing. We mimic this process with a random function $H:\zo^n \to \zo^R$ by using the function value $H(x)$ as the random bits that determine the amplitude at index $x$. We set
\begin{align*}
    \ket{\psi_H} \propto \sum_x \Samp(H(x)) \ket{x},
\end{align*}
where $\Samp$ converts random bits into an approximate Gaussian sample. 
When $H$ is random, its values $H(x)$ are independent, so the amplitudes $v_x$ behave like independent Gaussian samples. After normalization, we show that the resulting truncated \amplistswords are statistically close to those of Haar-random states.

 This statistical random-function step is one point at which the state-preparation setting differs sharply from the unitary setting: as discussed in \Cref{ssec:unitary-synthesis}, an analogous statistically Haar-like construction for full unitary matrices with efficient oracle implementations would already imply unitary synthesis.

\paragraph{2. Preparing the target states in MH.}
The challenge is that knowing how $H(x)$ determines the amplitude $v_x$ does not immediately give a circuit that prepares $\ket{\psi_H} \propto \sum_x v_x \ket{x}$. 
Our state preparation has two conceptual components. A coherent sampler computes a numerical description of the amplitude $v_x$ associated with input $x$, and our amplitude loader turns those descriptions into the actual amplitudes of the state. We implement both within a fixed level of the magic hierarchy. 

\begin{enumerate}[i.]
    \item \textbf{Gaussian sampler} Given an index $x$, coherently compute the amplitude $\ket{x}\ket{0} \to \ket{x} \ket{v_x}$.
    We evaluate $H(x)$ and apply   $\Samp$ to obtain an encoding of $v_x = \Samp(H(x))$. 
    
    \item \textbf{Amplitude loader} Use coherent access to these computed values $v_x$ to prepare the normalized state $\ket{\psi_H} = \frac{1}{\norm{v}}\sum_{x} v_x \ket{x}$. 
\end{enumerate}

Our sampler generates Gaussian-like values using comparisons of random numbers. Our loader reconstructs each complex amplitude from its encoded magnitude and phase. We implement these operations so that, when $H$ can be evaluated in a fixed level of the magic hierarchy, the full preparation also remains in a fixed level. 

\paragraph{3. Replacing the random function with a PRF.}
Finally, we replace $H$ with a sufficiently secure PRF. We show how a test distinguishing the resulting \amplistswords could be used to distinguish the PRF from a random function: just compute the \amplistwords from the function, then run the test. The reduction constructs the classical \amplistwords directly from evaluations of the function; it does not simulate the quantum preparation circuit.

Constructing the entire \amplistwords takes $2^{O(n)}$ time, so ordinary polynomial-time PRF security is not enough for this reduction. We choose the PRF security parameter to be a sufficiently large polynomial in the state and precision parameters, so that the assumed subexponential security covers both this runtime and the required distinguishing advantage. When the PRF can additionally be evaluated by polynomial-size $\TCZ$ circuits, the preparation procedure from the previous step places the corresponding states in a fixed level of the magic hierarchy. This gives the required \PRAL family and completes the natural-proofs barrier.

\subsection{Related works}

\paragraph{Relation to scalable pseudorandom states.}
Our construction builds on the Gaussian-based framework of Brakerski
and Shmueli~\cite{brakerski2020scalable}, which generates approximate
Gaussian coordinates and prepares the corresponding normalized state
using quantum rejection sampling~\cite{ozols2013quantum}. Their goal is
polynomial-time state preparation, whereas ours is to realize the
sampling and amplitude-loading procedures within a fixed level of the
magic hierarchy. This requires an explicit finite Gaussian sampler and
an amplitude-loading implementation tailored to our restricted circuit
model.

The security notions also differ. A scalable pseudorandom state with
security parameter $\kappa$ is required to fool polynomial-time quantum
adversaries given polynomially many copies of the state, with $\kappa$
independent of the number of output qubits~\cite{brakerski2020scalable}.
A \PRAL family instead has to fool a classical algorithm that is given
the entire $T$-bit amplitude list. One way to convert the former access
model into the latter is tomography: reconstruct the state to sufficient
precision and then run the amplitude-list distinguisher.

Quantitatively, this requires a stronger guarantee than the standard
scalable-PRS definition. To reconstruct an $n$-qubit state accurately
enough to determine a $T$-bit amplitude list, with failure probability
small compared with a desired distinguishing advantage $2^{-\lambda}$,
a straightforward tomography reduction uses
$2^{O(n+T+\lambda)}$ copies and computation. Thus polynomial-time,
polynomial-copy security in $\kappa$ does not suffice when
$\kappa=\poly(n,T,\lambda)$. By contrast, a scalable PRS secure against,
say, $2^{\kappa^\alpha}$ quantum time and copies with advantage at most
$2^{-\kappa^\alpha}$ would suffice: choosing $\kappa$ to be a
sufficiently large polynomial in $n,T,\lambda$ covers the tomography
reduction while keeping the state-generation complexity polynomial.

The statistical random-function construction of Brakerski and Shmueli
already has enough quantitative strength for this route: the joint
distribution of $t$ copies is within $O(t/e^\kappa)$ trace distance of
$t$ copies of a Haar-random state~\cite{brakerski2020scalable}. Hence
even exponentially many copies can be supported by increasing
$\kappa$ only polynomially. Turning this into the corresponding
computational guarantee, however, would require the underlying
pseudorandom function to remain secure against the resulting
exponential-resource \emph{quantum} adversary.

Our proof avoids this stronger quantum assumption. Since the
distinguisher in our framework receives a classical amplitude list, we
construct that list directly from classical evaluations of the
underlying function. Consequently, the computational step of our
barrier requires only a classically subexponentially secure PRF.

\paragraph{Other mention of quantum natural proofs.}
An independent work by Chen et al.~\cite{chen2026finegrainedcomplexityquantumnatural} uses the term \emph{quantum natural proofs} in the QSETH setting, where natural properties are properties of Boolean functions whose constructivity may be quantum. Their framework is used to establish compression obliviousness from quantum-secure pseudorandom functions, whereas ours concerns properties of quantum states and state-preparation lower bounds.

\subsection{Discussion}

We know that various constant-time circuit models are equivalent and captured by the magic hierarchy. So our barrier also holds for other models at fixed depth. We chose to work within the magic hierarchy for convenience, and to compare with the levels in \cite{parham2025mh} where it is shown that the \emph{classical} natural proofs barrier is encountered.

\paragraph{Other preparation models.}
The magic hierarchy provides a convenient setting in which
to state a quantitative barrier, but the underlying obstruction
is not specific to this model. Any preparation model that can
simulate the circuits preparing the states in our \PRAL family inherits
the corresponding barrier. The simulations established
in~\cite{parham2025mh} therefore give consequences for
constant-depth circuits with fanout and for appropriate
bounded-round adaptive circuit models. We use the magic hierarchy
to track the required circuit resources and compare directly
with the earlier fixed-level connections to classical lower bounds.

\paragraph{Can the barrier extend to weaker circuit models?}
Our barrier applies whenever the states in our \PRAL family can be prepared within the circuit model of interest. Thus, extending the barrier to weaker models amounts to asking whether comparable \PRAL families can be prepared with fewer computational resources.

A particularly interesting target is $\QAC^0$. This class is contained in the more powerful constant-depth models captured by the magic hierarchy, but is conjectured to be strictly weaker. Proving lower bounds against $\QAC^0$ has attracted substantial recent attention, yet our understanding remains limited. With polynomially many unrestricted ancillas, the strongest known lower bounds only reach very small constant depths~\cite{padé2020depth2,rosenthal2021qac0,joshi2026qac0}; for arbitrary constant depth, known results require additional restrictions on the number of ancillas or circuit size~\cite{nadimpalli2024pauli,anshu2025computational}.

Recent work also suggests that existing techniques are approaching an interesting threshold. Anshu et al.~\cite{anshu2025computational} show that a depth-d $\QAC^0$ circuit computing PARITY needs at least $n^{1+ 3^{-d}}$ ancillas. They further observe that if one could strengthen this to be $n^{1+ \exp(-o(d))}$ that would already imply $\mathsf{PARITY}\notin\QAC^0$. This raises the possibility that further progress may require qualitatively new ideas. Is there a natural proofs barrier in this regime?

Recent work shows that pseudorandom unitaries, and hence pseudorandom states in the usual copy-access sense, can be implemented with high accuracy by polynomial-size $\QAC^0$ circuits~\cite{qac0-pru,grier-qac0-tc0} \footnote{The original construction incurred inverse polynomial implementation error from approximating polylogarithmic-width FANOUT gates. Subsequent work of Grier, Morris, and Wu gives exact polynomial-size $\QAC^0$ implementations of such fanout gates, removing this bottleneck from \cite{grier-qac0-tc0}}. 
What remains insufficient for our purposes is the security guarantee: the construction in \cite{qac0-pru} is secure only against quasipolynomial-time adversaries, whereas a distinguisher given the entire amplitude list may run in $2^{O(n)}$ time. Strengthening this security would therefore give a promising route to a natural-proofs barrier for $\QAC^0$.

\paragraph{Unitary natural proofs.}
It is natural to ask whether the same barrier can be established for
lower bounds on implementing unitaries. In
\Cref{sec:unitary-natural-proofs}, we formulate the corresponding
notion of a natural unitary property and identify the pseudorandom
matrix object that would imply such a barrier. Interestingly, the
unitary setting appears qualitatively more difficult than the state
setting: we show that reproducing the statistical random-function
stage of our state construction for full unitary matrices would
already imply a positive solution to the unitary synthesis problem,
via the relaxed-net argument of Brakerski and
Yuen~\cite{brakerski2026scalablePRU}.

\pdfbookmark[1]{Acknowledgments}{acknowledgments}
\paragraph*{Acknowledgments.}
We thank Henry Yuen, Miranda Christ, Tal Malkin, John Bostanci, and Ben Foxman for helpful discussions. We also thank Sabee Grewal for helpful feedback on this manuscript.

\pdfbookmark[1]{AI disclosure}{ai-disclosure}
\paragraph*{AI disclosure.}
The main ideas and structure of this work were entirely human-generated. Before using AI, we had a full version of the result, but our use of AI significantly improved it. Using ChatGPT, we were able to optimize our prepration of the states in the \PRAL family to a much lower circuit complexity, and generate some additional examples. We also used AI for assistance with improving the exposition, and checking arguments. We take full responsibility for the exposition and correctness of this manuscript.
 \section{Quantum natural proofs}
\subsection{Definitions}

\begin{definition}[State \amplistwords, $\al_T(\psi)$]\label{def:state-amp-list}
    Fix a signed binary encoding for real numbers in $[-1,1]$, and let $c_{\mathrm{enc}}$ be the number of sign and format bits used per complex amplitude. Let $\ket{\psi}=\sum_{x\in\zo^n}\alpha_x\ket{x}$ be an $n$-qubit state. For $T\in\NN$, let $\al_T(\psi)$ denote the binary string obtained by writing, for every $x\in\zo^n$, the real and imaginary parts of $\alpha_x$ truncated to $T$ fractional bits. Thus
    \[
        \al_T(\psi)\in \pbra{\zo^{2T+c_{\mathrm{enc}}}}^{2^n}\cong \zo^{2^n(2T+c_{\mathrm{enc}})}.
    \]
    We also identify $\al_T(\psi)$ with the vector in $\CC^{2^n}$ obtained by interpreting this binary string as truncated complex amplitudes.

\end{definition}

For a quantum circuit model $\calC$, write $\calC_{\poly}$ for its
polynomial-size restriction, namely, the class of circuit families
$\{C_n\}_{n\in\NN}$ such that $C_n\in\calC$ and there is a polynomial
$q$ satisfying $\abs{C_n}\leq q\pbra{n}$. Thus, $\calC$ specifies the structural
restrictions on the circuits, while $\calC_{\poly}$ additionally
restricts their size as a function of the number of qubits.

\begin{definition}[Natural promise property for quantum states]
\label{def:natural-promise-property}
    A promise property of $n$-qubit states is a pair of disjoint sets
    \begin{align*}
        P_n^{\mathrm{yes}},P_n^{\mathrm{no}}
        \sse S(\CC^{2^n})
    \end{align*}
    both invariant under global phase.
    The promise property
    $\{P_n^{\mathrm{yes}},P_n^{\mathrm{no}}\}_{n\in\NN}$
    is \emph{natural} if it satisfies:
    \begin{enumerate}
        \item \textbf{Constructivity}: There is a polynomial-time deterministic classical
        algorithm $F$ and a polynomial $\tau:\NN\to\NN$ such that, for
        every $T\geq\tau\pbra{n}$,
        \begin{align*}
            \ket{\psi}\in P_n^{\mathrm{yes}}
            &\implies F\pbra{1^n,\al_T\pbra{\psi}}=1,\\
            \ket{\psi}\in P_n^{\mathrm{no}}
            &\implies F\pbra{1^n,\al_T\pbra{\psi}}=0.
        \end{align*}
        Here, running time is measured in the total input length. We assume that $P_n^{\mathrm{yes}}$ is Haar-measurable so that the probability in the largeness condition is well-defined.
        \item \textbf{Largeness}: There is a constant $c>0$ such that,
        for all sufficiently large $n$,
        \begin{align*}
            \Pr_{\phi\sim\Haar_n}
            \sbra{\ket{\phi}\in P_n^{\mathrm{yes}}}
            \geq 2^{-cn}.
        \end{align*}

    \end{enumerate}

    For a quantum circuit model $\calC$, we define \emph{usefulness} against $\calC$ as follows.
    \begin{enumerate}
        \setcounter{enumi}{2}
        \item \textbf{Usefulness:} For a quantum circuit model $\calC$ and an error bound $\eps:\NN\to[0,1)$, the property is $\eps$-\emph{useful against} $\calC$ if every pure-state family $\{\ket{\psi_n}\}_{n\in\NN}$ that can be $\eps$-approximately prepared by $\calC$ is in $P_{n}^{\mr{no}}$. Formally, if there exists a circuit family $\{C_n\}_{n\in \NN} \in \calC$ with reduced states $\sigma_n$ on designated $n$-qubit output registers satisfying 
        \begin{align}
            \frac{1}{2}\tnorm{\ket{\psi_n}\bra{\psi_n} - \sigma_n} \leq \eps(n)
        \end{align}
        for every $n\in \NN$, then $\ket{\psi_n}\in P_n^{\mathrm{no}}$ for all sufficiently large $n$.
    \end{enumerate}
\end{definition}

While the natural property is a property of \emph{pure states}, the approximating states $\sigma_n$ prepared by $\calC$ may be \emph{mixed}. In this definition, each circuit acts on an all zero input, has a designated output register, and the remaining qubits can be traced out. In particular, the ancillary qubits need not be returned to $\ket{0}$.

\subsection{Quantum natural-proofs barrier from a \PRAL family}

To consider the ability of classical algorithms to distinguish between states, we give them a classical description of the state. Specifically, the algorithm receives $\al_T\pbra{\psi}$, in which the real and imaginary parts of every amplitude are truncated to $T$ bits of precision, as defined in \Cref{def:state-amp-list}.

\begin{definition}[Pseudorandom amplitude-list family]\label{def:PRAL}
    A family $\calD = \{D_{n,\lambda, T}\}_{n,T, \lambda \in\NN}$ of distributions on normalized $n$-qubit state vectors is a \emph{pseudorandom amplitude-list} (PRAL) family if, for every probabilistic polynomial-time classical distinguisher $A$, and every pair of polynomially bounded functions $\lambda, T:\NN\to\NN$,
    \begin{align*}
        \abs{
            \Pr_{\ket{\psi}\gets
                D_{n, \lambda(n), T(n)}}
                \sbra{
                    A\pbra{
                        1^n, 1^{\lambda(n)},
                        \al_{T(n)}(\psi)
                    }=1
                }
            -
            \Pr_{\ket{\phi}\sim\Haar_n}
                \sbra{
                    A\pbra{
                        1^n, 1^{\lambda(n)},
                        \al_{T(n)}(\phi)
                    }=1
                }
        }
        \leq 2^{-\lambda(n)}
    \end{align*}
    for all sufficiently large $n$.
    The running time of $A$ is measured with respect to its total input length. Efficient sampling from $D_{n,\lambda,T}$ is not required.
\end{definition}

The natural-proofs barrier only uses the linear-security regime $\lambda(n)=\Theta(n)$.
Throughout, $\Haar_n$ denotes Haar probability measure on the unit
sphere in $\CC^{2^n}$, so its samples are normalized state vectors.
A PRAL is purely a pseudorandomness notion and by itself places no restriction on state preparation. For the barrier, we additionally require the following nonuniform preparation condition.

\begin{definition}[Approximate circuit preparable]
\label{def:approximate-implementability}
    We say that a \PRAL family $\{D_{n,\lambda,T}\}_{n,\lambda,T\in\NN}$ is \emph{efficiently $\eps$-approximate $\calC$-preparable} if there exists a polynomial $q$ such that, for every $n,\lambda,T\in\NN$ and every $\ket{\psi}$ in the support of $D_{n,\lambda,T}$, there exists a circuit $C_\psi\in\calC$ of size at most $q(n,\lambda,T)$ that when applied to the all zero state, and tracing out the non-output register prepares an $n$-qubit state $\sigma_\psi$ satisfying
    \[
        \frac{1}{2}
        \tnorm{\ket{\psi}\bra{\psi}-\sigma_\psi}
        \leq\eps(n).
    \]
    The circuit $C_\psi$ and its output $\sigma_\psi$ may depend on $n,\lambda,T$, and $\ket{\psi}$, but the polynomial $q$ is common to all choices.
    The state $\sigma_\psi$ may be mixed and the discarded output qubits need not be returned to zero.

\end{definition}

\paragraph{Barrier intuition.}
When $\lambda$ and $T$ are polynomially bounded in $n$,
approximate preparation tells us there is a polynomial-size circuit that $\eps$-approximately prepares each state in the distribution $D_{n, \lambda, T}$. An $\eps$-useful natural property must therefore reject each of the states but also accept a nontrivial fraction of Haar-random states. Therefore, its constructivity algorithm distinguishes $D_{n, \lambda, T}$ from Haar random.

\begin{theorem}[Quantum natural-proofs barrier]
\label{thm:quantum-natural-proofs-barrier}
    Let $\calC$ be a quantum circuit model and let $\eps:\NN\to[0,1)$. Suppose there exists a \PRAL family $\{D_{n, \lambda, T}\}$ that is efficiently $\eps$-approximate $\calC$-preparable. Then there are no natural promise properties that are $\eps$-useful against $\calC_{\poly}$.
\end{theorem}

\begin{proof}
    Suppose for the sake of contradiction that such a property exists, and let $F$, $\tau$, and $c>0$ witness its constructivity and largeness. Set $T(n)=\tau(n)$ and $\lambda(n)=\ceil{cn} + 1$, and define $A\pbra{1^n,1^\lambda,t}=F\pbra{1^n,t}$. This is a deterministic polynomial-time classical distinguisher.

    Since $\lambda$ and $T$ are polynomially bounded, there is a polynomial $r$ such that every state in the support of $D_{n, \lambda(n), T(n)}$ is within trace distance $\eps(n)$ of a state preparable by a $\calC$ circuit of size at most $r(n)$. We claim that by the property's $\eps$-usefulness against $\calC_\poly$, for sufficiently large $n$, every state in this support belongs to $P_n^{\mathrm{no}}$. Otherwise, we can build a family of states $\{\psi_n\}$ that contradicts usefulness as follows. For infinitely many $n \in \NN$ we can select $\ket{\psi_n}$ to be a state in the support of $D_{n, \lambda(n), T(n)}$ that is not in $P_n^{\mathrm{no}}$, and for all other $n$ we set $\ket{\psi_n}$ as an arbitrary supported state. Now each of the states $\ket{\psi_n}$ can be $\eps(n)$-approximately prepared by a $\calC$ circuit of size $r(n)$, but there is no $n_0$ such that for all $n\geq n_0$ each $\ket{\psi_n}\in P_n^\mr{no}$. This contradicts $\eps$-usefulness.
    
    Therefore, by constructivity, for sufficiently large $n$,
    \begin{align*}
        \Pr_{\ket{\psi}\gets
                D_{n,\lambda(n),T(n)}}
                \sbra{
                    A\pbra{
                        1^n,1^{\lambda(n)},
                        \al_{T(n)}(\psi)
                    }=1
                } = 0
    \end{align*}
    By largeness, for all sufficiently large $n$, 
    \begin{align*}
        \Pr_{\ket{\phi}\sim\Haar_n}
                \sbra{A\pbra{1^n,1^{\lambda\pbra{n}},
                    \al_{T\pbra{n}}\pbra{\phi}}=1}
                \geq 2^{-cn}>2^{-\lambda(n)}.
    \end{align*}
    Therefore, $A$ distinguishes $D_{n,\lambda(n),T(n)}$ from the Haar measure with advantage strictly greater than $2^{-\lambda\pbra{n}}$, contradicting the security of $D_{n,\lambda(n),T(n)}$ in \Cref{def:PRAL}.
\end{proof}
 \section{Description of states in \PRAL family}
\label{sec:PRAL-target-state}

In this section we describe the states in our \PRAL family. 
To sample a truly Haar-random state, we could draw each amplitude independently from a complex Gaussian distribution, and then normalize. 
We mimic this construction using a
function
$H:\zo^n\to\zo^R$.
For each basis string $x$, we interpret $H(x)$ as a finite random tape
and apply a deterministic sampler $\Samp$ to obtain a complex number
$v_x=\Samp(H(x))$. Thus a single function $H$ determines an entire vector
$v=(v_x)_{x\in\zo^n}$, and hence a target quantum state after
normalization.

When $H$ is a uniformly random function, the strings $H(x)$ are
independent and uniform, so the values $v_x$ behave like independent
approximate Gaussian samples. Later we will replace the random function
by a pseudorandom one. 

We use the sampler in both the security analysis and the eventual
low-depth preparation, so it will be convenient to isolate the
quantitative guarantees we need from it. The parameters below track
its scaling, approximation error, and failure probability. We then use
such a sampler to associate a target state $\ket{\psi_H}$ with every
fixed function $H$.

\subsection{Finite complex Gaussian sampler}\label{sec:finite-complex-gaussian-sampler}

\begin{definition}[Complex Gaussian sampler]
\label{def:complex-gaussian-sampler}
    Fix positive integers $R,\ell$, a scale $s>0$,
    an approximation error $\xi\geq0$, and a failure
    probability $\eta_{\mathrm{samp}}\in[0,1]$.

    An $(s,\xi,\eta_{\mathrm{samp}})$-complex Gaussian
    sampler is a deterministic map $\Samp:\zo^R\to\zo^\ell$
    whose output encodes a complex number of magnitude
    at most one on every input.
    For a uniformly random tape $\omega\in\zo^R$,
    let $V=\Samp(\omega)$, identifying the output
    encoding with its decoded complex value.
    We require that there exists a coupling of $V$
    with a standard complex Gaussian random variable $G\sim\calN^\CC(0,1)$ such that
    \begin{equation}
        \Pr\sbra{\abs{sV-G}>\xi}\leq\eta_{\mathrm{samp}}
    \end{equation}
    We use the convention
    $\mathbb E\abs{G}^2=1$.
\end{definition}

Although we call $\Samp$ a sampler, it is a deterministic function of its input tape. The randomness comes from the choice of random input tape $\omega$. For now we treat $\Samp$ as a black box satisfying the guarantees above. We construct an explicit sampler in the magic hierarchy later in \Cref{ssec:complex-gaussian-sampler}.

\paragraph{Encoding convention.}
The sampler's output represents the complex number $\alpha = r e^{2\pi i \theta}$ by its polar coordinates $(r, \theta)$ up to $O(R)$ binary fractional bits. This internal encoding is distinct from the amplitude list $\al_T(\psi)$ which records the real and imaginary parts (cartesian coordinates).

\subsection{Target states from a complex-Gaussian sampler}\label{sec:target-state-definition}
Fix a function $H:\zo^n\to\zo^R$, and a sampler $\Samp: \zo^R \to \zo^\ell$. Let $v \in \CC^{N}$ be defined by $v_x = \Samp(H(x))$ for each $x\in \zo^n$. For nonzero $v\in\CC^N$, write $\ket{\wh{v}} \coloneq \frac{1}{\norm{v}}\sum_{x\in\zo^n}v_x\ket{x}$.

For a uniformly random $H$ and suitable sampler parameters, the
distribution of these normalized vectors approaches Haar measure.  Our preparation procedure, however,
requires a lower bound on $\norm{v}$ in order to prepare the state
efficiently. We therefore introduce a threshold $\gamma\in(0,1]$ and
define the target state as follows.

\begin{definition}[Target state $\ket{\psi_H}$ associated with $H$]
\label{def:target-state-for-H}
For $v$ defined above, and some threshold $\gamma \in (0, 1]$, let
\begin{equation}
\ket{\psi_H}
=
\begin{cases}
\displaystyle
\ket{\wh{v}},
& \text{if }\norm{v}\geq\gamma\sqrt N,\\
\ket{0^n},
& \text{if }\norm{v}<\gamma\sqrt N.
\end{cases}
\end{equation}
\end{definition}

The threshold is needed only for the preparation procedure: when $\norm{v}\geq\gamma\sqrt N$, our amplitude loader efficiently prepares $\ket{\widehat v}$; otherwise the target is defined to be the trivially preparable state $\ket{0^n}$. This case distinction is part of the target-state definition and need not itself be evaluated by the circuit.

In \Cref{sec:security}, we analyze distributions over these states obtained first from a uniformly random function and then from a pseudorandom function. The preparation of these states in the magic hierarchy is deferred to \Cref{sec:implementing-PRAL}.

 \section{Security of the \PRAL construction} \label{sec:security}
We prove security in two steps. First, we show that for a uniformly random function $H$, the corresponding distribution over target states $\ket{\psi_H}$ have truncated amplitude lists statistically close to Haar. The main subtlety is that nearby amplitudes can have different truncations. Second, we show how to reconstruct these lists from classical oracle access to $H$, allowing us to replace the random function by a sufficiently secure PRF.

\subsection{Security with a truly random function}\label{ssec:security:trulyrandfunc}

Let $H:\zo^n\to\zo^R$ be a uniformly random function. We write
$D_{\mathrm{rf}}$ for the resulting distribution of target states
$\ket{\psi_H}$ from \Cref{def:target-state-for-H}.

\begin{theorem}[Random-function security]
\label{thm:random-function-security}
    Fix $n\geq1$, write $N=2^n$, and let $\gamma\in(0,1]$. Let $\Samp$ be an $(s,\xi,\eta_{\mathrm{samp}})$-complex Gaussian sampler, and let $D_{\mathrm{rf}}$ be the target-state distribution defined by this sampler and threshold for uniformly random $H$.

    Suppose there exists $\beta\in(0,1)$ such that
    $s\gamma+\xi\leq\beta$.
    Then, for every truncation precision $T\in\NN$,
    \begin{equation}
         \TVD\!\left(
            \al_T(D_{\mathrm{rf}}),\al_T(\Haar_n)
        \right)
        \leq
        N\eta_{\mathrm{samp}}
        +e^{-N(1-\beta^2)^2/2}
        +C N^{3/2}2^T\frac{2\xi}{\beta},
    \end{equation}
    where $C>0$ is a universal constant.
\end{theorem}
The three error terms account for sampler failure, an unusually small Gaussian norm, and proximity to a truncation boundary.

Before proving this theorem, it ios useful to first prove the following lemma that bounds the probability that two nearby states will be truncated to different values in their \amplistswords.

\begin{lemma}[Haar amplitudes avoid truncation boundaries]
\label{lem:haar-truncation-boundaries}
Let $\ket{\phi}\sim\Haar_n$, where $N=2^n\geq2$, and let
\[
    \calB_T
    \coloneq
    \{j2^{-T}:j\in\ZZ,\,-2^T\leq j\leq2^T\}
\]
be the set of $T$-bit truncation boundaries on $[-1,1]$.
For every $\rho\geq0$, the probability that the real or imaginary part
of any amplitude $\phi_x$ lies within distance $\rho$ of a boundary in
$\calB_T$ is at most
\begin{equation}
    C N^{3/2}2^T\rho,
\end{equation}
where $C>0$ is a universal constant.
\end{lemma}

\begin{proof}
Fix $x\in\zo^n$, and let $X$ be either
$\operatorname{Re}(\phi_x)$ or $\operatorname{Im}(\phi_x)$.
Viewing $\CC^N$ as $\RR^{2N}$, a Haar-random state is uniformly
distributed on the real unit sphere. By the standard slice-integration
formula for the sphere~\cite[Corollary~A.5]{axler2001harmonic}, $X$
has density
$f(y)=c_{2N}(1-y^2)^{(2N-3)/2}$ for $\abs{y}\leq1$, where
$c_{2N}=O(\sqrt N)$. Thus the density of $X$ is bounded everywhere by
$C\sqrt N$ for a universal constant $C$.

There are $O(2^T)$ truncation boundaries, so the union of the intervals
within distance $\rho$ of these boundaries has total length
$O(2^T\rho)$. Hence, for one fixed real or imaginary amplitude
coordinate, the probability of lying within distance $\rho$ of any
boundary is at most $C\sqrt N\,2^T\rho$.

There are $2N$ real and imaginary amplitude coordinates in total.
A union bound therefore gives
$\Pr[\text{some amplitude coordinate lies within $\rho$ of a truncation boundary}]
\leq C N^{3/2}2^T\rho$,
after absorbing constant factors into $C$.
\end{proof}

We now proceed to the proof of \Cref{thm:random-function-security}.
\begin{proof}[Proof of \Cref{thm:random-function-security}]
We first couple the sampled amplitude vector to a true Gaussian vector. Under $D_{\mathrm{rf}}$, the function $H$ is uniformly random, so the tapes $\{H(x)\}_{x\in\zo^n}$ are independent and uniform. Let $v\in\CC^N$ be defined by $v_x=\Samp(H(x))$
for each $x\in\zo^n$. By the sampler guarantee in \Cref{def:complex-gaussian-sampler}, we may couple each $v_x$ independently to a standard complex Gaussian $g_x$ such that
\begin{equation}
    \Pr\!\left[\abs{s v_x-g_x}>\xi\right]
    \leq \eta_{\mathrm{samp}}.
\end{equation}
Let $g=(g_x)_{x\in\zo^n}$, so that $g\sim\calN^\CC(0,1)^N$. In particular, $\ket{\wh g}$ is Haar random. By a union bound,
\begin{align}
    \Pr[\norm{sv-g}_2>\xi\sqrt N]
\leq
    N\eta_{\mathrm{samp}}.
    \label{eq:secpf:twonormboundgs}
\end{align}
We now define the following ``good'' event
\begin{equation}
    E
    \coloneq
    \left\{
        \norm{sv-g}_2\leq\xi\sqrt N
        \ \text{and}\ 
        \norm{g}_2\geq\beta\sqrt N
    \right\}.
\end{equation}
We show that this good event guarantees that our state $\ket{\psi_F}$ actually corresponds to the state $\ket{\wh{v}}$ and that it is close to the Gaussian sample.

\begin{claim}[Good event implies high norm and Haar closeness]
\label{lem:good-event-high-norm-haar-closeness}
On the event $E$, we have $\norm{v}_2\geq\gamma\sqrt N$, and hence
$\ket{\psi_H}=\ket{\wh v}$. Moreover,
\begin{equation}
    \norm{\ket{\wh v}-\ket{\wh g}}_2
    \leq
    \frac{2\xi}{\beta}.
\end{equation}
\end{claim}

\begin{proof}
On $E$, the reverse triangle inequality gives
$s\norm{v}_2=\norm{sv}_2\geq\norm{g}_2-\norm{sv-g}_2
\geq(\beta-\xi)\sqrt N$.
Since $s\gamma+\xi\leq\beta$, we have
$\beta-\xi\geq s\gamma$, and therefore
$\norm{v}_2\geq\gamma\sqrt N$. Moreover, on $E$ we have $\norm{g}_2\geq\beta\sqrt N>0$,
so both $\ket{\wh v}$ and $\ket{\wh g}$ are well defined. By \Cref{def:target-state-for-H}, this implies
$\ket{\psi_F}=\ket{\wh v}$.

Let $u:=sv$. Since $s>0$ and $v\neq0$ on $E$, we have
$\ket{\wh u}=\ket{\wh v}$. Therefore the triangle and reverse triangle inequalities give us
\begin{align}
    \norm{\ket{\wh u}-\ket{\wh g}}_2
&\leq
    \frac{2\norm{u-g}_2}{\norm{g}_2}
    \leq
    \frac{2\xi\sqrt N}{\beta\sqrt N}
    =
    \frac{2\xi}{\beta}.
\end{align}
The final inequality uses the two conditions defining $E$.
\end{proof}

\begin{claim}[The good event occurs with high probability]
\label{lem:good-event-probability}
The event $E$ satisfies
\begin{equation}
    \Pr[E]
    \geq
    1-N\eta_{\mathrm{samp}}
    -\exp\!\left(-\frac{N(1-\beta^2)^2}{2}\right).
\end{equation}
\end{claim}

\begin{proof}
By \Cref{eq:secpf:twonormboundgs},
$\norm{sv-g}_2\leq\xi\sqrt N$ with probability at least
$1-N\eta_{\mathrm{samp}}$.

Furthermore, $\norm{g}_2^2=\sum_{i=1}^N\abs{g_i}^2$ is a sum of
$N$ independent mean-one exponential random variables. A standard
Chernoff bound therefore gives
\[
    \Pr[\norm{g}_2^2\leq\beta^2N]
    \leq
    \exp\!\left(-\frac{N(1-\beta^2)^2}{2}\right).
\]

The event $E$ fails only if either
$\norm{sv-g}_2>\xi\sqrt N$ or
$\norm{g}_2<\beta\sqrt N$.
Therefore by union bound,
$\Pr[E^c]\leq
N\eta_{\mathrm{samp}}
+\exp(-N(1-\beta^2)^2/2)$,
\end{proof}

Now we combine preceding bounds to complete the proof of the random-function security theorem.

Let
$\rho=2\xi/\beta$.
By Lemma~\ref{lem:good-event-high-norm-haar-closeness}, on the event $E$ we have
$\ket{\psi_F}=\ket{\wh v}$ and
$\norm{\ket{\psi_F}-\ket{\wh g}}_2\leq\rho$.
Since $\ket{\wh g}$ is Haar random, Lemma~\ref{lem:haar-truncation-boundaries}
implies that, except with probability at most
$C N^{3/2}2^T\rho$, no real or imaginary amplitude of
$\ket{\wh g}$ lies within distance $\rho$ of a truncation boundary.

On the intersection of $E$ with this no-boundary event, the bound
$\norm{\ket{\psi_F}-\ket{\wh g}}_2\leq\rho$
implies that every real and imaginary amplitude coordinate differs by at most
$\rho$. Hence the two states have the same $T$-bit truncated \amplistwords.

Therefore, under this coupling, the two truncated lists can differ only if
$E$ fails or if a Haar amplitude lies within distance $\rho$ of a truncation
boundary. By the coupling characterization of total variation distance and
Lemma~\ref{lem:good-event-probability},
\begin{equation}
\TVD\!\left(
    \al_T(D_{\mathrm{rf}}),
    \al_T(\Haar_n)
\right)
\leq
N\eta_{\mathrm{samp}}
+
\exp\!\left(-\frac{N(1-\beta^2)^2}{2}\right)
+
C N^{3/2}2^T\frac{2\xi}{\beta}.
\end{equation}
\end{proof}

We give a concrete parameter setting under which the random-function security theorem yields error at most \(\eps\).

\begin{corollary}[Parameters for random function security]\label{cor:random-function-security-parameters}
    Let $N=2^n\geq2$, let $T\in\NN$, and let $\eps\in(0,1)$ satisfying $\eps \geq 2 e^{-N/8}$. Let $C\geq1$ be a universal constant for which \Cref{thm:random-function-security} holds, and let $\Samp$ be a complex Gaussian sampler as in \Cref{def:complex-gaussian-sampler} with parameters
    \begin{align}
        s\geq\frac12, \qquad \xi\leq\frac{\eps}{16CN^{3/2}2^T}, \qquad \eta_{\mathrm{samp}}\leq\frac{\eps}{4N}.
    \end{align}
    Then for $\gamma=1/(2s)$, the state ensemble $D_{\mathrm{rf}}$ defined using $\Samp$ and $\gamma$ is $\eps$-close to the Haar measure when viewing only truncated \amplistswords at precision $T$:
    \begin{align}
        \TVD\pbra{\al_T\pbra{D_{\mathrm{rf}}},\al_T\pbra{\Haar_n}}\leq\eps.
    \end{align}
\end{corollary}

\begin{proof}
    We set $\beta=1/\sqrt{2}$. The assumed bounds give $\xi \leq 1/16$, so $s\gamma + \xi = 1/2 + \xi < \beta$.  
By \Cref{thm:random-function-security},
    \begin{align}
        \TVD\pbra{\al_T\pbra{D_{\mathrm{rf}}},\al_T\pbra{\Haar_n}}
        &\leq N\eta_{\mathrm{samp}}+e^{-N(1-\beta^2)^2/2}+CN^{3/2}2^T\frac{2\xi}{\beta}\\
&\leq \frac{\eps}{4}+e^{-N/8}+\frac{\eps}{4}
        \leq \eps.
    \end{align}
    The final inequality follows from  $\eps \geq 2 e^{-N/8}$.
\end{proof}

In particular, for any polynomially bounded $b:\NN\to\NN$, setting $\eps(n)=2^{-b(n)}$ satisfies the dimension condition for all sufficiently large $n$; hence samplers with the stated parameters give \amplistwords total variation distance at most $2^{-b(n)}$.

For the remainder of this section, whenever $n,T,\eps$ are fixed, we
use the sampler and threshold chosen in
\Cref{cor:random-function-security-parameters}. We suppress this
dependence from the notation.

\subsection{Reconstructing the amplitude list}\label{ssec:security:pseudorandomfunc}

We now replace the uniformly random function $H$ by a pseudorandom
function $F_k$. Fix $n,T,\eps$ and use the sampler and target-state
parameters fixed above. Let $D_{\mathrm{prf}}$ denote the resulting
distribution of $\ket{\psi_{F_k}}$ for a uniformly random key $k$.

Our goal is to show that no efficient classical distinguisher can
distinguish $D_{\mathrm{prf}}$ from the random-function distribution
$D_{\mathrm{rf}}$, assuming that $\mathcal F$ is a sufficiently secure PRF family.
Together with the random-function security result above, this will
imply that $D_{\mathrm{prf}}$ is indistinguishable from Haar at the
level of truncated \amplistswords.

The only additional ingredient we need is the ability to reconstruct
the amplitude list from classical oracle access to the underlying
function.

\begin{lemma}[Classical reconstruction of the truncated amplitude list]
\label{lem:classical-amp-list-reconstruction}
Fix $n,T,\eps$ and a sampler satisfying the parameter conditions of \Cref{cor:random-function-security-parameters}. Assume that the tape length $R$ and classical evaluation time of the sampler are polynomial in $n,T,\log(1/\eps)$.

Given classical oracle access to any
$H:\zo^n\to\zo^R$, there is a classical procedure that either outputs
the exact list $\al_T(\psi_H)$ or reports failure.  It makes $N=2^n$
oracle queries and runs in time
$\poly(N,R,T,\log(1/\eps))$.  If $H$ is uniformly random, the failure
probability is at most $\eps$.
\end{lemma}

\begin{proof}
Query $H(x)$ for every $x\in\zo^n$ and evaluate the sampler to obtain
$v_x=\Samp(H(x))$ for all $x$. Compute $\norm{v}_2^2$ and compare it
with $\gamma^2N$. If $\norm{v}_2^2<\gamma^2N$, then
$\ket{\psi_H}=\ket{0^n}$ by \Cref{def:target-state-for-H}, so the
procedure returns the known truncated \amplistwords of $\ket{0^n}$.

Otherwise, $\ket{\psi_H}=\ket{\wh v}$. The procedure converts each
normalized amplitude $v_x/\norm{v}_2$ from the sampler's polar encoding
to Cartesian coordinates. For each real and imaginary component, it
computes an interval of length at most $2^{-p}$ containing the exact
value.

If any such interval intersects a truncation boundary in $\calB_T$, the
procedure reports failure. Otherwise, every interval lies entirely within
a single truncation bin and therefore determines the corresponding
$T$-bit encoding uniquely. Concatenating these encodings gives exactly
$\al_T(\psi_H)$.

We now bound the failure probability when $H$ is uniformly random.
By \Cref{cor:random-function-security-parameters},
$s\gamma+\xi\leq 9/16<1/\sqrt{2}$, so the good-event lemmas apply
with $\beta=1/\sqrt{2}$. Thus, by
~\Cref{lem:good-event-probability},
$\Pr[E^c]\leq N\eta_{\mathrm{samp}}+e^{-N/8}$, and by
~\Cref{lem:good-event-high-norm-haar-closeness}, on $E$ we have
$\norm{\ket{\psi_H}-\ket{\wh g}}_2\leq2\sqrt{2}\,\xi$.

Suppose the reconstruction procedure reports failure. Then some interval
containing a real or imaginary coordinate of $\ket{\psi_H}$ intersects
a truncation boundary, so that coordinate lies within distance $2^{-p}$
of the boundary. On $E$, the corresponding coordinate of
$\ket{\wh g}$ therefore lies within distance
$2\sqrt{2}\,\xi+2^{-p}$ of the same boundary. 
Choose $p
    =
    \left\lceil
        \log_2\!\left(
            \frac{16CN^{3/2}2^T}{\eps}
        \right)
    \right\rceil$.
By
Lemma~\ref{lem:haar-truncation-boundaries} and our parameter choices from \Cref{cor:random-function-security-parameters}
\begin{align}
    \Pr[\text{failure}] 
    &\leq
    N\eta_{\mathrm{samp}}+e^{-N/8} + CN^{3/2}2^T(2\sqrt2\,\xi+2^{-p})\\
    &\leq
    \frac{\eps}{4}+\frac{\eps}{2}
   +\frac{2\sqrt2+1}{16}\eps
 \leq\eps.
\end{align}

It remains to bound the running time. The procedure makes $N$ oracle queries. Since $p=O(n+T+\log(1/\eps)$, sampler evaluation and the required finite-precision arithmetic take time polynomial in $N,R,T,\log(1/\eps)$.
\end{proof}

\begin{theorem}[PRF replacement]
\label{thm:prf-replacement}
    Fix $n,T,\lambda\in\NN$, and let
    $\mathcal F=\{F_k:\zo^n\to\zo^R\}$ be a keyed function family.
    Let $D_{\mathrm{prf}}$ be the distribution of $\ket{\psi_{F_k}}$
    for a uniformly random key $k$, and let $D_{\mathrm{rf}}$ be the
    corresponding distribution for a uniformly random $H:\zo^n\to\zo^R$.
    Suppose $\mathcal F$ is secure against classical oracle distinguishers
    running in time at most $t_{\mathrm{PRF}}$, with distinguishing
    advantage at most $\eps_{\mathrm{PRF}}$.
    Assume moreover that the reconstruction procedure of
    \Cref{lem:classical-amp-list-reconstruction} runs in time at most
    $t_{\mathrm{rec}}$ and, for uniformly random $H$, reports failure with
    probability at most $\eps_{\mathrm{rec}}$.
    
    Then every classical distinguisher $A$ running in time $t_A$ with $t_A+t_{\mathrm{rec}}\leq t_{\mathrm{PRF}}$
    satisfies
    \begin{align*}
    \begin{aligned}
    \Big|
    \Pr_{\psi\gets D_{\mathrm{prf}}}
      [A(1^n,1^\lambda,\al_T(\psi))=1] 
      -
    \Pr_{\psi\gets D_{\mathrm{rf}}}
      [A(1^n,1^\lambda,\al_T(\psi))=1]
    \Big|
    \leq
    \eps_{\mathrm{PRF}}+\eps_{\mathrm{rec}}.
    \end{aligned}
    \end{align*}
\end{theorem}

\begin{proof}
    Fix a classical distinguisher $A$ satisfying the stated running-time bound.
By Lemma~\ref{lem:classical-amp-list-reconstruction}, given classical
oracle access to a function $H:\zo^n\to\zo^R$, we can either reconstruct
$\al_T(\psi_H)$ exactly or report failure. When $H$ is uniformly random,
the failure probability is at most $\eps_\mr{rec}$.

We use this reconstruction procedure to build two classical oracle
distinguishers $B_0$ and $B_1$. For $b\in\{0,1\}$, the distinguisher
$B_b$ runs the reconstruction procedure. If reconstruction succeeds and
returns a list $z$, it outputs
$A(1^n,1^\lambda,z)$. If reconstruction fails, it outputs $b$.
By the running-time bound in
Lemma~\ref{lem:classical-amp-list-reconstruction}, each $B_b$ runs in
time at most
$t_A+t_\mr{rec}\leq t_{\mathrm{PRF}}$.

Let $a_{\mathrm{prf}}$ and $a_{\mathrm{rf}}$ denote the acceptance
probabilities of $A$ on exact \amplistswords drawn from
$D_{\mathrm{prf}}$ and $D_{\mathrm{rf}}$, respectively.
Whenever reconstruction succeeds, $B_0$ and $B_1$ both run $A$ on the
exact \amplistwords. On failure, $B_0$ outputs $0$ and $B_1$ outputs $1$.
Therefore,
$\Pr[B_0^{F_k}=1]\leq a_{\mathrm{prf}}\leq\Pr[B_1^{F_k}=1]$.
When the oracle is a uniformly random function, reconstruction fails with
probability at most $\eps_\mr{rec}$. Hence forcing the output on failure to be
$0$ or $1$ changes the acceptance probability by at most $\eps_\mr{rec}$, so
\[
    a_{\mathrm{rf}}-\eps_\mr{rec}
    \leq
    \Pr[B_0^{H}=1]
    \leq
    \Pr[B_1^{H}=1]
    \leq
    a_{\mathrm{rf}}+\eps_\mr{rec},
\]
where $H$ is uniformly random.
Since each $B_b$ runs within the PRF security time bound,
for
$b\in\{0,1\}$ we have
\[\abs{\Pr[B_b^{F_k}=1]-\Pr[B_b^H=1]}\leq\eps_{\mathrm{PRF}}.\]
Combining these inequalities gives $a_{\mathrm{rf}}-\eps_\mr{rec}-\eps_{\mathrm{PRF}}
    \leq
    a_{\mathrm{prf}}
    \leq
    a_{\mathrm{rf}}+\eps_\mr{rec}+\eps_{\mathrm{PRF}}$ and therefore
\begin{equation}
    \abs{a_{\mathrm{prf}}-a_{\mathrm{rf}}}
    \leq
    \eps_\mr{rec}+\eps_{\mathrm{PRF}}.
\end{equation}
\end{proof}

\subsection{\PRAL security from a subexponentially-secure PRF}

The reconstruction procedure above runs in time polynomial in
$N=2^n$, and is therefore exponential in the number of qubits $n$.
Consequently, ordinary polynomial-time PRF security is insufficient
for our reduction. We therefore require a PRF family that remains secure
against classical distinguishers of subexponential size in its security
parameter $m$. By choosing $m$ to be a sufficiently large polynomial
in $n$, this security budget will cover the $2^{O(n)}$-time reduction.

\begin{definition}[Classically subexponentially secure PRFs]
\label{def:subexp-prf}
Let
$\mathcal G=\{G_{m,k}:\zo^m\to\zo^m\}_{m\in\NN,\;k\in\mathcal K_m}$
be a keyed function family whose uniform keys can be sampled in
$\poly(m)$ time and whose functions can be evaluated in $\poly(m)$ time.

We call $\mathcal G$ \emph{classically subexponentially secure} if
there exists a constant $a\in(0,1)$ such that, for all sufficiently
large $m$, every classical oracle algorithm $B$ running in time at most
$2^{m^a}$ satisfies
\begin{equation}
    \abs{
        \Pr_{k\gets\mathcal K_m}[B^{G_{m,k}}=1]
        -
        \Pr_{U_m}[B^{U_m}=1]
    }
    \leq
    2^{-m^a},
\end{equation}
where $U_m:\zo^m\to\zo^m$ is a uniformly random function.
\end{definition}

\begin{remark}
    Classically subexponentially secure PRFs can already be constructed from
    subexponentially secure length-doubling pseudorandom generators.
    Concretely, the GGM construction~\cite{goldreich1986how} constructs
    a PRF from a length-doubling PRG and \cite[Proposition~5]{jain2012hardness}
    shows that if the PRG has
    $c m^\epsilon$ bits of security, then the resulting GGM PRF has
    $c'm^\epsilon$ bits of security for every constant $c'<c/2$.
    Thus, for any fixed $a<\epsilon$, this yields a PRF satisfying \Cref{def:subexp-prf} for all sufficiently
    large $m$.

    This is sufficient for the security argument in this section. However, for our ultimate magic-hierarchy implementation, we need the PRF itself to be computable by polynomial-size \textit{constant-depth} threshold circuits. The standard GGM construction does not preserve constant depth: its evaluation consists of a sequential chain of PRG evaluations. We therefore assume the existence of the PRF directly when instantiating the construction in $\TCZ$.
\end{remark}

\begin{theorem}[\PRAL security from subexponentially secure PRFs]
\label{thm:PRAL-security}
Let $\mathcal G$ be a classically subexponentially secure PRF family
in the sense of \Cref{def:subexp-prf}. 

Suppose that for every requested error tolerance $\eps\in(0,1)$,
there is a complex Gaussian sampler satisfying the parameter bounds of
\Cref{cor:random-function-security-parameters}, with tape length and
classical evaluation time polynomial in
$n,T,$ and $\log(1/\eps)$.

Then the sampler and PRF parameters can be chosen so that the resulting
target-state distributions
$\{D_{n,\lambda,T}\}_{n,\lambda,T\in\NN}$
form a \PRAL family.
\end{theorem}

The proof has three steps: (1) Choose the sampler parameters and a
larger internal PRF security parameter $m$, and use the PRF to define the
tape functions $F_k$ underlying our target-state ensemble. (2) Show
that the classical reduction from \Cref{thm:prf-replacement} runs within
the subexponential security budget of the PRF, which allows us to replace
the pseudorandom tape function by a uniformly random one. (3) Combine this replacement with the random-function security result from
\Cref{cor:random-function-security-parameters} to conclude that the
resulting ensemble is indistinguishable from Haar in the sense required
by \PRAL security.

\begin{proof}
\medskip
\noindent\textbf{(1) Defining the ensemble.}
Let $a\in(0,1)$ be the constant from
\Cref{def:subexp-prf}. Fix an integer $d$ such that $ad > 2$. We choose our parameters so that the random-function error from \Cref{thm:random-function-security}, reconstruction failure from \Cref{thm:prf-replacement}, and PRF distinguishing advantage are each at most $\eps = 2^{-\lambda -3}$.

For each $n, \lambda, T$, choose the sampler and target threshold according to parameter choice in \Cref{cor:random-function-security-parameters} with $\eps = 2^{-\lambda -3}$. 
By assumption, the tape length $R$ and classical evaluation time are polynomially bounded in $n,\lambda,$ and $T$.

Set $m=(n+\lambda+T+R+4)^d$. 
Then $m$ is polynomially bounded in $n,\lambda$, and $T$ and satisfies
$m\geq\max\{n,R\}$.
For a key $k\in\mathcal K_m$, define $F_k(x)
    =
    \operatorname{prefix}_R
    \left(
        G_{m,k}(0^{m-n}x)
    \right)$
for each $x\in\zo^n$, where $\operatorname{prefix}_R$ keeps the first $R$ output bits. The padding and truncation adapt the input and output lengths of $\calG$ to those required by the sampler.

Let $D_{n,\lambda,T}$ be the distribution of $\ket{\psi_{F_k}}$ for a uniformly random key $k\in\mathcal K_m$. Let $D_{\mathrm{rf}}$ denote the random-function distribution
using the same sampler parameters and target threshold as
$D_{n,\lambda,T}$.

To verify security, fix a polynomial-time classical distinguisher $A$ and polynomially bounded functions $\lambda(n),T(n)$. We suppress their dependence on $n$ below. Since $\lambda(n)$ is polynomially bounded, $\eps=2^{-\lambda-3}\geq2e^{-2^n/8}$ for all sufficiently large $n$, so \Cref{cor:random-function-security-parameters} applies.

\medskip
\noindent\textbf{(2) Runtime of the PRF reduction.}
We next verify that the PRF-replacement reduction runs within the time bound covered by the security guarantee of $\mathcal G$. The input to $A$ has length
$O(n+\lambda+N(T+1))$, so since $A$ runs in polynomial time and the
reconstruction procedure runs in time polynomial in
$N,R,T,$ and $\log(1/\eps)$, their combined running time is polynomial
in $N,n,\lambda,T,$ and $R$.

To simulate an oracle query to the tape function $F_k$, the reduction
queries $G_{m,k}$ on the padded input $0^{m-n}x$ and keeps the first
$R$ output bits. This adds only polynomial overhead in $m$. Hence the
entire PRF distinguisher runs in time
$\poly(N,n,\lambda,T,R,m)$.

Along the fixed polynomially bounded functions $\lambda(n)$ and $T(n)$,
the quantities $R$ and $m$ are polynomially bounded in $n$, while
$N=2^n$. Therefore the running time above is $2^{O(n)}$.

On the other hand, our choice of $m$ gives
$m^a\geq(n+4)^{da}$. Since $da>2$, for all sufficiently large $n$ we
have $2^{O(n)}\leq2^{m^a}$. Thus the reduction is covered by the
subexponential PRF security guarantee in \Cref{def:subexp-prf}.

\medskip
\noindent\textbf{(3) \PRAL security from PRF and random-function security.}
When $G_{m,k}$ is replaced by a uniformly random function
$U_m:\zo^m\to\zo^m$, the padded inputs $0^{m-n}x$ are distinct for
distinct $x\in\zo^n$. Therefore, the values $U_m(0^{m-n}x)$ are
independent and uniform in $\zo^m$. 

Consequently, the induced function
$F(x)
    =
    \operatorname{prefix}_R
    \left(
        U_m(0^{m-n}x)
    \right)$
is a uniformly random function from $\zo^n$ to $\zo^R$. Thus, the
random-function side of the PRF experiment induces exactly the
distribution $D_{\mathrm{rf}}$ analyzed above.

Applying \Cref{thm:prf-replacement} with
$t_{\mathrm{PRF}}=2^{m^a}$ and
$\eps_{\mathrm{PRF}}=2^{-m^a}$ gives
\begin{equation}
\abs{
\Pr_{\psi\gets D_{n,\lambda,T}}
[A(1^n,1^\lambda,\al_T(\psi))=1]
-
\Pr_{\psi\gets D_{\mathrm{rf}}}
[A(1^n,1^\lambda,\al_T(\psi))=1]
}
\leq
2^{-m^a}+\eps.
\end{equation}
Our choice of $m$ also gives $m^a\geq\lambda+3$, so
$2^{-m^a}\leq\eps$. Therefore the PRF-replacement bound is at most
$2\eps$.

By the random-function security bound above (\Cref{cor:random-function-security-parameters}), $\TVD\pbra{\al_T(D_{\mathrm{rf}}),\al_T(\Haar_n)}\leq\eps$.
By the defining property of total variation distance, the acceptance
probability of $A$ on these two distributions therefore differs by at
most $\eps$.

Combining the two bounds with the triangle inequality gives
\begin{equation}
\abs{
\Pr_{\psi\gets D_{n,\lambda,T}}
[A(1^n,1^\lambda,\al_T(\psi))=1]
-
\Pr_{\phi\sim\Haar_n}
[A(1^n,1^\lambda,\al_T(\phi))=1]
}
\leq
3\eps
<
2^{-\lambda},
\end{equation}
where the final inequality uses
$\eps=2^{-\lambda-3}$.

This holds for every polynomial-time classical distinguisher $A$ and
every polynomially bounded pair of functions
$\lambda,T:\NN\to\NN$, for all sufficiently large $n$.
Therefore the ensemble
$\{D_{n,\lambda,T}\}_{n,\lambda,T\in\NN}$
satisfies the \PRAL security condition in \Cref{def:PRAL}.
\end{proof}

 \section{Preparing the states in the \PRAL family in \MH}\label{sec:implementing-PRAL}

In this section, we show that the target states underlying the \PRAL construction can be prepared at a constant level of the magic hierarchy. We first work with an arbitrary function $H:\zo^n\to\zo^R$ and then instantiate $H$ with the pseudorandom functions used in \Cref{thm:PRAL-security}.

Let $N=2^n$. Recall from \Cref{def:target-state-for-H} that, for each $x\in\zo^n$, we interpret $H(x)$ as the random tape of a finite complex Gaussian sampler and set $v_x=\Samp(H(x))$. Writing $v=(v_x)_{x\in\zo^n}$, the target state is $\ket{\psi_H}=\ket{\wh v}=\norm{v}^{-1}\sum_xv_x\ket{x}$ when $\norm{v}\geq\gamma\sqrt N$, and is $\ket{0^n}$ otherwise. 

We prepare $\ket{\wh v}$ in the following two steps. First, we coherently compute the bounded amplitude $v_x$ from $H(x)$. Second, we load these computed amplitudes into the data register and normalize them. The construction follows the Gaussian-sampling and rejection-sampling approach of~\cite{brakerski2020scalable}, but we implement both steps at low magic-hierarchy depth.

We describe how to implement amplitude computation  in~\Cref{ssec:complex-gaussian-sampler} and amplitude loading in~\Cref{ssec:amplitude-loading-in-mh}.

Our procedure for preparing $\ket{\wh v}$ is inspired by the approach in \cite{brakerski2020scalable} for preparing scalable PRS, which consists of multivariate Gaussian sampling and quantum rejection sampling. They show how to implement these two steps in polynomial time, but for our purposes we need more efficient state preparation in the magic hierarchy.

\subsection{Complex Gaussian sampler}
\label{ssec:complex-gaussian-sampler}
Here we construct the finite complex Gaussian sampler from
\Cref{def:complex-gaussian-sampler} and show how to implement it in the magic
hierarchy.

Our goal is to generate the amplitudes of a Haar-random state. Since a normalized vector of independent standard complex Gaussians is Haar random, it suffices to generate amplitudes from sampling standard complex Gaussians. However, there are two issues with directly doing so. First, a Gaussian is continuous, whereas a quantum circuit has to work with a finite number of bits. We will therefore need to discretize the Gaussian. Second, a Gaussian is unbounded, whereas to prepare a state with some desired amplitudes (as in \Cref{ssec:amplitude-loading-in-mh}), we require these amplitudes to be represented on quantum registers, which means the amplitudes need to be bounded. 

Therefore, we construct a finite, bounded sampler whose output distribution approximates a standard complex Gaussian after a known rescaling. The sampler itself is a deterministic map from a finite random input tape to a bounded complex number, and all its randomness is supplied through the input tape.

To construct this sampler, we first make the following observation: for a standard complex Gaussian $G\sim\mc{N}^{\CC}(0,1)$, its phase is uniformly random, while its radius satisfies $\Pr[|G|>r]=e^{-r^2}$. Therefore to approximate $G$ given uniform randomness, we just need to reproduce this radial distribution. To do so, we use a geometric construction based on a Poisson number of uniformly random points in the unit square.

In~\Cref{sssec:ideal-gaussian-sampler}, we describe an ideal continuous version of this construction and show that its rescaled output approximates a standard complex Gaussian. We then replace the continuous randomness with finite-precision randomness and truncate the unbounded Poisson distribution in~\Cref{sssec:finite-gaussian-sampler}, and we show that the resulting finite sampler still approximates a standard complex Gaussian. Finally, in
\Cref{sssec:gaussian-coherent-implementation}, we show how to evaluate this
finite deterministic map coherently in the magic hierarchy.

\subsubsection{Ideal continuous sampler}
\label{sssec:ideal-gaussian-sampler}

We first consider an ideal continuous construction for reproducing the radial distribution of a standard complex Gaussian, which satisfies $\Pr[|G|>r]=e^{-r^2}$. 

We begin by drawing $K\sim\Pois(\mu)$, where $\mu$ is the expected number of points. Conditioned on $K$, we then draw $K$ points independently and uniformly from the unit square $[0,1)^2$. At this stage, $K$ can be arbitrarily large and the coordinates of the points are sampled continuously. We will introduce the finite-precision and bounded version later.

To see why this construction works, we give the following geometric intuition. 
Imagine expanding a square from the origin until it first reaches
one of the sampled points. The required side length is
$\min_i\max\{U_{i,0},U_{i,1}\}$. With a Poisson number of points,
the probability that $[0,t]^2$ contains no point is
$e^{-\mu t^2}$, matching the Gaussian radial tail after rescaling.
The definition below sends the exceptional case of no sampled
points to zero.

Below we formalize this sampler construction and prove that it approximates the standard complex Gaussian. We also account for the case when the Gaussian samples are greater than the truncation threshold and bound the error from this.

\begin{definition}[Ideal Poisson sampler]
\label{def:ideal-poisson-sampler}
For $\mu\geq1$, the ideal Poisson sampler
$\Samp_{\mathrm{id}}$ is defined as follows.
\begin{enumerate}
    \item Draw $K\sim\Pois(\mu)$.
    \item Conditioned on $K$, draw $K$ independent points
    $(U_{i,0},U_{i,1})$ uniformly from $[0,1)^2$.
    \item If $K=0$, set $R_{\mathrm{id}}=0$. Otherwise, set
    \begin{equation}
        R_{\mathrm{id}}
        =
        \min_{i\in[K]}\max\{U_{i,0},U_{i,1}\}.
    \end{equation}
    \item Draw $\Theta_{\mathrm{id}}$ uniformly from $[0,1)$,
    independently of all previous randomness.
    \item Output
    $V_{\mathrm{id}}
    =
    R_{\mathrm{id}}e^{2\pi i\Theta_{\mathrm{id}}}$.
\end{enumerate}
\end{definition}
Next, we show that the ideal Poisson sampler has the same distribution as a standard complex Gaussian scaled down by $1/\sqrt{\mu}$, except that samples with magnitude at least $\sqrt{\mu}$ are replaced by zero.
\begin{lemma}[Ideal sampler as a truncated Gaussian]
\label{lem:ideal-poisson-sampler}
Let $V_{\mathrm{id}}$ be the output of
$\Samp_{\mathrm{id}}$, and let
$G\sim\mc{N}^{\CC}(0,1)$. Define the truncation and rescaling map
\begin{equation}
T_\mu(x)
=
\begin{cases}
x/\sqrt{\mu}, & \abs{x}<\sqrt{\mu},\\
0, & \abs{x}\geq\sqrt{\mu}.
\end{cases}
\end{equation}
Then $V_{\mathrm{id}}$ has the same distribution as $T_\mu(G)$.
In particular, $V_{\mathrm{id}}$ and $G$ can be coupled by first
sampling $G$ and setting $V_{\mathrm{id}}=T_\mu(G)$. Under this
coupling, $\sqrt{\mu}V_{\mathrm{id}}=G$ except with probability
$\eta_{\mathrm{GaussTrunc}}=e^{-\mu}$.
\end{lemma}
\begin{proof}
We first determine the distribution of the magnitude
$R_{\mathrm{id}}$. Fix $t\in[0,1)$. For a single sampled point
$(U_{i,0},U_{i,1})$, the condition
$\max\{U_{i,0},U_{i,1}\}>t$ holds exactly when the point lies outside
$[0,t]^2$. Since this square has area $t^2$,
$\Pr[\max\{U_{i,0},U_{i,1}\}>t]=1-t^2$.

Conditioned on $K=k\geq1$, the event $R_{\mathrm{id}}>t$ occurs
exactly when all $k$ independent points lie outside $[0,t]^2$. Hence $\Pr[R_{\mathrm{id}}>t\mid K=k]=(1-t^2)^k$. Since $K\sim\Pois(\mu)$, by definition
$\Pr[K=k]=e^{-\mu}\mu^k/k!$ for $k=0,1,2,\ldots$. Thus, averaging over $K\sim\Pois(\mu)$ and excluding $K=0$, for which
$R_{\mathrm{id}}=0$, gives
\begin{equation}
\Pr[R_{\mathrm{id}}>t]
=
\sum_{k=1}^{\infty} e^{-\mu}\frac{\mu^k}{k!}(1-t^2)^k
=
e^{-\mu}\sum_{k=1}^{\infty}\frac{\bigl(\mu(1-t^2)\bigr)^k}{k!}
=
e^{-\mu}\left(e^{\mu(1-t^2)}-1\right)
=
e^{-\mu t^2}-e^{-\mu}.
\end{equation}
We now compare the rescaled sampler magnitude with the magnitude of a standard
complex Gaussian. For $G\sim\mc{N}^{\CC}(0,1)$,
$\Pr[\abs{G}>r]=e^{-r^2}$. To match this with the expression
$e^{-\mu t^2}$ above, fix $0\leq r<\sqrt{\mu}$ and set
$t=r/\sqrt{\mu}$. Substituting gives $\Pr[\sqrt{\mu}R_{\mathrm{id}}>r]
=
e^{-r^2}-e^{-\mu}$.
Since $\Pr[\abs{G}\geq\sqrt{\mu}]=e^{-\mu}$, we also have $\Pr[r<\abs{G}<\sqrt{\mu}]
=
e^{-r^2}-e^{-\mu}$.
Thus $\sqrt{\mu}R_{\mathrm{id}}$ matches the Gaussian radial distribution
up to radius $\sqrt{\mu}$, while the remaining Gaussian tail of probability
$e^{-\mu}$ is mapped to zero.

The phase $\Theta_{\mathrm{id}}$ is independent and uniform on
$[0,1)$, as is the phase of a standard complex Gaussian. Combining the
magnitude and phase therefore shows that
$V_{\mathrm{id}}$ has the same distribution as $T_\mu(G)$.

We can now construct the claimed coupling explicitly. First sample
$G\sim\mc{N}^{\CC}(0,1)$ and set
$V_{\mathrm{id}}=T_\mu(G)$. By the distributional identity above,
$V_{\mathrm{id}}$ in this joint experiment has exactly the output
distribution of $\Samp_{\mathrm{id}}$, so this is a valid
coupling. Under this coupling,
$\sqrt{\mu}V_{\mathrm{id}}=G$ whenever
$\abs{G}<\sqrt{\mu}$. The only failure event is therefore
$\abs{G}\geq\sqrt{\mu}$, which has probability
$\eta_{\mathrm{GaussTrunc}}=e^{-\mu}$.
\end{proof}
Thus, after rescaling, the ideal Poisson sampler exactly reproduces a standard complex Gaussian below the truncation threshold $\sqrt{\mu}$. The only difference is the Gaussian tail beyond this threshold, which is replaced by zero and occurs with probability $\eta_{\mathrm{GaussTrunc}}=e^{-\mu}$.

\subsubsection{Finite-precision sampler}
\label{sssec:finite-gaussian-sampler}

In this section, we describe a finite version of the ideal Poisson sampler, which is what we will actually implement to prepare the states in the \PRAL family. 

There are two obstacles to implementing the ideal construction directly. First, the Poisson count $K$ is unbounded, so the ideal sampler may require arbitrarily many sampled points. Second, the point coordinates and phase are continuous random variables in $[0,1)$ and therefore cannot be represented exactly using a finite number of bits. We therefore cap the number of candidate pairs at $K_{\max}$ and represent each uniform random variable using $q$ bits of precision. 

To distinguish the ideal Poisson count from the finite one, throughout this
section we write $K_{\mathrm{id}}$ for the random variable denoted by
$K$ in~\Cref{sssec:ideal-gaussian-sampler}, and $K_{\mathrm{fin}}$ for the
corresponding count generated by the finite sampler.

The finite sampler receives a uniformly random tape
$\Omega\sim\Unif(\{0,1\}^R)$. The tape contains all of the randomness needed
to reproduce the ideal sampling procedure with finite precision. It contains a $q$-bit integer $Z$ used to determine the finite Poisson count $K_{\mathrm{fin}}$, two $q$-bit binary fractions $U^{\mathrm{fin}}_{i,0},U^{\mathrm{fin}}_{i,1}$ for each $i=1,\ldots,K_{\max}$ giving the coordinates of the $K_{\max}$ candidate points, and a $q$-bit binary fraction $\Theta_{\mathrm{fin}}$ used for the phase. Thus we may write
\begin{equation}
    \Omega=(Z,
U^{\mathrm{fin}}_{1,0},U^{\mathrm{fin}}_{1,1},
\ldots,
U^{\mathrm{fin}}_{K_{\max},0},U^{\mathrm{fin}}_{K_{\max},1},
\Theta_{\mathrm{fin}}),
\end{equation}
where the total tape length is $R=(2K_{\max}+2)q$. Here $Z\in\{0,\ldots,2^q-1\}$, while each phase or coordinate value is
obtained by interpreting its corresponding $q$-bit string as an integer and
dividing by $2^q$, giving a binary fraction in $[0,1)$.

We next describe how to use the uniform integer $Z$ to simulate the Poisson count. For each $i=1,\ldots,K_{\max}$, we want the comparison with
$Z$ to reproduce the event that the ideal Poisson draw contains at least $i$
points, namely $K_{\mathrm{id}}\geq i$. Let $P_i=\Pr[K_{\mathrm{id}}<i]$. Choose an integer threshold
$T_i$ such that $T_i/2^q$ approximates $P_i$. More precisely, for fixed parameters $\mu\geq1$, $K_{\max}\in\NN$, and $q\in\NN$, choose $T_i\in\{0,\ldots,2^q\}$ so that, for every
$i\in[K_{\max}]$,
\begin{equation}
\abs{
\frac{T_i}{2^q}
-
P_i
}
\leq 2^{-q}.
\label{eq:poisson-threshold-accuracy}
\end{equation}
The thresholds depend only on the sampler parameters and can be precomputed
classically, which we describe in the next section. Since $Z$ is uniform on $\{0,\ldots,2^q-1\}$, we have
$\Pr[Z\geq T_i]=1-T_i/2^q$. On the other hand,
$\Pr[K_{\mathrm{id}}\geq i]=1-P_i$. Therefore,
\eqref{eq:poisson-threshold-accuracy} implies that the comparison
$Z\geq T_i$ reproduces the ideal event $K_{\mathrm{id}}\geq i$ up to
probability error at most $2^{-q}$. The probabilities $P_i$ are nondecreasing in $i$, so the thresholds $T_i$
can be chosen nondecreasing as well. Therefore, for any fixed value of $Z$,
the comparisons $Z\geq T_i$ consist of some number of initial true
comparisons followed by only false comparisons. We set $K_{\mathrm{fin}}$
to be the number of these true comparisons. Thus $K_{\mathrm{fin}}$ is the finite sampler's approximation to the ideal
Poisson count and satisfies
$0\leq K_{\mathrm{fin}}\leq K_{\max}$.

We now formalize the finite sampler and then show that it approximates the
ideal sampler from~\Cref{sssec:ideal-gaussian-sampler}, and hence, after
rescaling, a standard complex Gaussian in~\Cref{lem:finite-complex-gaussian-sampler} and~\Cref{cor:finite-sampler-parameter-choice}.

\begin{definition}[Finite Poisson sampler]
\label{def:finite-poisson-sampler}
Fix $\mu\geq1$, $K_{\max}\in\NN$, and $q\in\NN$, and let
$T_1,\ldots,T_{K_{\max}}$ be the corresponding thresholds satisfying
\eqref{eq:poisson-threshold-accuracy}. Let $R=(2K_{\max}+2)q$, and sample
the input tape uniformly as $\Omega\sim\Unif(\{0,1\}^R)$.

Conditioned on the sampled tape $\Omega$, the finite Poisson sampler
$\Samp_{\mathrm{fin}}$ is deterministic and proceeds as follows.
\begin{enumerate}
    \item Set $Z$ to be the integer represented by the
    first $q$ bits of $\Omega$. Set $K_{\mathrm{fin}}
        =
        \abs{\cbra{i\in[K_{\max}] : Z\geq T_i}}$.

    \item For each $i=1,\ldots,K_{\max}$, set
    $(U^{\mathrm{fin}}_{i,0},U^{\mathrm{fin}}_{i,1})\in[0,1)^2$
    to be the next pair of $q$-bit binary fractions in $\Omega$.

    \item If $K_{\mathrm{fin}}=0$, set $R_{\mathrm{fin}}=0$. Otherwise, set
    \begin{equation}
        R_{\mathrm{fin}}
        =
        \min_{1\leq i\leq K_{\mathrm{fin}}}
        \max\{U^{\mathrm{fin}}_{i,0},U^{\mathrm{fin}}_{i,1}\}.
    \end{equation}

    \item Set $\Theta_{\mathrm{fin}}\in[0,1)$ to be the binary fraction
    represented by the final $q$ bits of $\Omega$.

    \item Output
    $\Samp_{\mathrm{fin}}(\Omega)
    =(R_{\mathrm{fin}},\Theta_{\mathrm{fin}})$,
    representing the complex number
    $V_{\mathrm{fin}}
    =R_{\mathrm{fin}}e^{2\pi i\Theta_{\mathrm{fin}}}$.
\end{enumerate}
\end{definition}

In the state-preparation construction,
the tape associated with each basis string $x$ is $H(x)$. We therefore write
$\Samp_{\mathrm{fin}}(H(x))=(r_x,\theta_x)$ and
$v_x=r_xe^{2\pi i\theta_x}$. 

We now quantify how closely the finite sampler approximates a standard
complex Gaussian after rescaling. There are two effects to account for. Rounding the candidate coordinates and phase to $q$ bits introduces a small
numerical error. Separately, the coupling guarantee can fail if the Gaussian sample
lies in the truncated tail, if the ideal Poisson count exceeds $K_{\max}$,
or if the finite threshold comparisons produce a different count from the
ideal Poisson count. The first of these, from the Gaussian truncation, was already present in the ideal sampler. The following lemma bounds both the numerical error and the total
coupling-failure probability.

\begin{lemma}[Finite sampler approximation]
\label{lem:finite-complex-gaussian-sampler}
Let $V_{\mathrm{fin}}$ be the output of the finite Poisson sampler
$\Samp_{\mathrm{fin}}$, and let $G\sim\mc{N}^{\CC}(0,1)$. The rescaled
finite-sampler output $\sqrt{\mu}V_{\mathrm{fin}}$ approximates the
distribution of $G$ in the following sense: $V_{\mathrm{fin}}$ and $G$
can be coupled so that
\begin{equation}
\abs{\sqrt{\mu}V_{\mathrm{fin}}-G}
\leq
\varepsilon_{\mathrm{approx}}
\end{equation}
except with probability at most $\eta_{\mathrm{fail}}$, where $\varepsilon_{\mathrm{approx}}
=
(1+2\pi)\sqrt{\mu}\,2^{-q}$ and $\eta_{\mathrm{fail}}
=
\eta_{\mathrm{GaussTrunc}}
+
\eta_{\mathrm{PoisTrunc}}
+
\eta_{\mathrm{Thresh}}$. Here
$\eta_{\mathrm{GaussTrunc}}=e^{-\mu}$,
$\eta_{\mathrm{PoisTrunc}}
=\Pr[K_{\mathrm{id}}>K_{\max}]$, and
$\eta_{\mathrm{Thresh}}=K_{\max}2^{-q}$.
\end{lemma}

\begin{proof}
We first show that the finite sampler uses the same number of candidate
points as the ideal sampler except with small probability. There are two
possible sources of disagreement. The ideal Poisson count
$K_{\mathrm{id}}$ may exceed the maximum $K_{\max}$ available to the finite
sampler, or the finite thresholds may cause $K_{\mathrm{fin}}$ to differ
from $K_{\mathrm{id}}$ even when $K_{\mathrm{id}}\leq K_{\max}$.

To bound the latter probability, that the finite thresholds cause
$K_{\mathrm{fin}}$ to differ from $K_{\mathrm{id}}$, we compare the ideal
and finite counts using the same underlying randomness. We first give an
equivalent way to generate the ideal Poisson count from a single uniform
random variable. Draw
$U_{\mathrm{count}}\sim\Unif[0,1)$. Recall that
$P_i=\Pr[K_{\mathrm{id}}<i]$, and set $P_0=0$. Define
$K_{\mathrm{id}}=k$ whenever
$P_k\leq U_{\mathrm{count}}<P_{k+1}$. Since $U_{\mathrm{count}}$ is
uniform, the probability of obtaining $K_{\mathrm{id}}=k$ is the length
of this interval,
$P_{k+1}-P_k=\Pr[K_{\mathrm{id}}=k]$. Thus this procedure indeed generates
$K_{\mathrm{id}}\sim\Pois(\mu)$. In particular,
$K_{\mathrm{id}}\geq i$ if and only if
$U_{\mathrm{count}}\geq P_i$.

Recall that in the finite sampler, $Z$ is required only to be uniform on
$\{0,\ldots,2^q-1\}$. Since we are already using
$U_{\mathrm{count}}\sim\Unif[0,1)$ to generate the ideal Poisson count, we
use the same random variable to generate $Z$ by setting
$Z=\floor{2^qU_{\mathrm{count}}}$. This has exactly the required uniform
distribution, because each value $j\in\{0,\ldots,2^q-1\}$ corresponds to an
interval of length $2^{-q}$ in $[0,1)$. Moreover, $Z\geq T_i$ if and only if
$U_{\mathrm{count}}\geq T_i/2^q$.

Thus the ideal and finite samplers make their $i$th count decision by
comparing the same uniform random variable against the two thresholds
$P_i$ and $T_i/2^q$. By~\eqref{eq:poisson-threshold-accuracy}, these
thresholds differ by at most $2^{-q}$. Hence the two comparisons disagree
only when $U_{\mathrm{count}}$ lies between them, which occurs with
probability at most $2^{-q}$. Taking a union bound over
$i=1,\ldots,K_{\max}$, the probability that any threshold comparison
disagrees is at most
$\eta_{\mathrm{Thresh}}=K_{\max}2^{-q}$.

If no threshold comparison disagrees and
$K_{\mathrm{id}}\leq K_{\max}$, then
$K_{\mathrm{fin}}=K_{\mathrm{id}}$. The other possible disagreement is the
Poisson truncation event $K_{\mathrm{id}}>K_{\max}$, which occurs with
probability
$\eta_{\mathrm{PoisTrunc}}
=\Pr[K_{\mathrm{id}}>K_{\max}]=\sum_{k=K_{\max}+1}^{\infty}
e^{-\mu}\frac{\mu^k}{k!}$.

Assume now that $K_{\mathrm{fin}}=K_{\mathrm{id}}$. We next bound the
numerical error introduced by replacing the ideal continuous coordinates and
phase with their $q$-bit approximations. This will give a bound on
$\abs{V_{\mathrm{fin}}-V_{\mathrm{id}}}$.

To compare the two outputs directly, we generate their coordinates and phases
from the same underlying continuous randomness. For this purpose, we may equivalently realize the ideal sampler by drawing an
independent sequence of uniform coordinates $U_{i,a}\sim\Unif[0,1)$ and using
only the first $K_{\mathrm{id}}$ pairs. On the event
$K_{\mathrm{id}}\leq K_{\max}$, only the first $K_{\max}$ pairs are
relevant. Since we are assuming $K_{\mathrm{fin}}=K_{\mathrm{id}}$, the two
samplers use the same first $K_{\mathrm{id}}$ candidate pairs. We can
therefore compare their coordinates pairwise.

Take the variables $U_{\mathrm{count}}$, $\Theta_{\mathrm{id}}$, and
$U_{i,a}$ for all $i,a$ to be mutually independent. For each
$i=1,\ldots,K_{\max}$ and $a\in\{0,1\}$, we obtain the finite coordinate
by rounding the corresponding ideal coordinate down to $q$ bits: $U^{\mathrm{fin}}_{i,a}
=
2^{-q}\floor{2^qU_{i,a}}$. Similarly, we set
$\Theta_{\mathrm{fin}}
=2^{-q}\floor{2^q\Theta_{\mathrm{id}}}$.
Since rounding a uniform random variable in $[0,1)$ in this way produces a
uniform $q$-bit binary fraction, these variables have exactly the
distribution required by the finite sampler. By construction,
$\abs{U^{\mathrm{fin}}_{i,a}-U_{i,a}}\leq2^{-q}$ and
$\abs{\Theta_{\mathrm{fin}}-\Theta_{\mathrm{id}}}\leq2^{-q}$.

We now propagate these rounding errors to the sampler output, and we first bound the resulting error in the magnitude. The finite coordinates are obtained by applying the same lower $q$-bit
rounding operation to every ideal coordinate. This rounding operation is
monotone, so it preserves the ordering of the coordinates and therefore
commutes with the maxima and minimum used to construct the magnitude.
Hence $R_{\mathrm{fin}}$ is exactly the lower $q$-bit approximation of
$R_{\mathrm{id}}$, $R_{\mathrm{fin}}
=
2^{-q}\floor{2^qR_{\mathrm{id}}}$.
Therefore, $\abs{R_{\mathrm{fin}}-R_{\mathrm{id}}}
\leq2^{-q}$,

It remains to bound the error in the phase factor. Since
$\abs{\Theta_{\mathrm{fin}}-\Theta_{\mathrm{id}}}\leq2^{-q}$, and
$\abs{e^{ia}-e^{ib}}\leq\abs{a-b}$ for real $a,b$, we have $\abs{
e^{2\pi i\Theta_{\mathrm{fin}}}
-
e^{2\pi i\Theta_{\mathrm{id}}}
}
\leq
2\pi\abs{\Theta_{\mathrm{fin}}-\Theta_{\mathrm{id}}}
\leq
2\pi2^{-q}$.
Since
$V_{\mathrm{fin}}
=R_{\mathrm{fin}}e^{2\pi i\Theta_{\mathrm{fin}}}$,
$V_{\mathrm{id}}
=R_{\mathrm{id}}e^{2\pi i\Theta_{\mathrm{id}}}$,
and $R_{\mathrm{id}}<1$, the triangle inequality gives
\begin{equation}
\abs{V_{\mathrm{fin}}-V_{\mathrm{id}}}
\leq
\abs{R_{\mathrm{fin}}-R_{\mathrm{id}}}
+
R_{\mathrm{id}}
\abs{
e^{2\pi i\Theta_{\mathrm{fin}}}
-
e^{2\pi i\Theta_{\mathrm{id}}}
} \leq
(1+2\pi)2^{-q}.
\end{equation}

Finally, by~\Cref{lem:ideal-poisson-sampler}, the rescaled ideal output
$\sqrt{\mu}V_{\mathrm{id}}$ agrees exactly with a standard complex Gaussian
$G\sim\mc{N}^{\CC}(0,1)$ except on the Gaussian truncation event, which has
probability $\eta_{\mathrm{GaussTrunc}}=e^{-\mu}$. Combining this guarantee
with the finite-to-ideal comparison above, whenever none of the three failure
events occurs, $\abs{\sqrt{\mu}V_{\mathrm{fin}}-G} =
\sqrt{\mu}\abs{V_{\mathrm{fin}}-V_{\mathrm{id}}} \leq
(1+2\pi)\sqrt{\mu}\,2^{-q} =
\varepsilon_{\mathrm{approx}}$.
By a union bound, the probability that at least one of the three failure
events occurs is at most $\eta_{\mathrm{fail}}
=
\eta_{\mathrm{GaussTrunc}}
+
\eta_{\mathrm{PoisTrunc}}
+
\eta_{\mathrm{Thresh}}$.
Therefore, the coupling satisfies
$\abs{\sqrt{\mu}V_{\mathrm{fin}}-G}\leq\varepsilon_{\mathrm{approx}}$
except with probability at most $\eta_{\mathrm{fail}}$.
\end{proof}

\begin{corollary}[Parameter choice]
\label{cor:finite-sampler-parameter-choice}
For any target approximation error
$\varepsilon_{\mathrm{approx}}^\star>0$ and target failure probability
$\eta_{\mathrm{fail}}^\star\in(0,1)$, the parameters
$\mu$, $K_{\max}$, and $q$ can be chosen so that $\varepsilon_{\mathrm{approx}}
\leq
\varepsilon_{\mathrm{approx}}^\star$ and $\eta_{\mathrm{fail}}
\leq
\eta_{\mathrm{fail}}^\star$. Moreover, as
$\varepsilon_{\mathrm{approx}}^\star,\eta_{\mathrm{fail}}^\star\to0$,
the parameters can be chosen so that $\mu,K_{\max}
=
O\left(
\log(1/\eta_{\mathrm{fail}}^\star)
\right)$
and $q
=
O\left(
\log(1/\varepsilon_{\mathrm{approx}}^\star)
+
\log(1/\eta_{\mathrm{fail}}^\star)
\right)$.
\end{corollary}

\begin{proof}
Fix a target approximation error
$\varepsilon_{\mathrm{approx}}^\star>0$ and a target failure probability
$\eta_{\mathrm{fail}}^\star\in(0,1)$. By
\Cref{lem:finite-complex-gaussian-sampler}, it is enough to choose the
parameters so that
$\varepsilon_{\mathrm{approx}}
\leq\varepsilon_{\mathrm{approx}}^\star$
and each of the three contributions to $\eta_{\mathrm{fail}}$ is at most
$\eta_{\mathrm{fail}}^\star/3$.

Choose
\begin{equation}
\mu
=
\ceil{\ln\frac{3}{\eta_{\mathrm{fail}}^\star}},
\qquad
K_{\max}=4\mu,
\qquad
q
=
\ceil{
\log_2
\max\left\{
\frac{3K_{\max}}{\eta_{\mathrm{fail}}^\star},
\frac{(1+2\pi)\sqrt{\mu}}
{\varepsilon_{\mathrm{approx}}^\star}
\right\}
}.
\end{equation}
We verify that these choices give the desired bounds.

First, the choice of $\mu$ gives
$\eta_{\mathrm{GaussTrunc}}=e^{-\mu}
\leq\eta_{\mathrm{fail}}^\star/3$.

Next, we bound the probability that the ideal Poisson count exceeds
$K_{\max}$. Since $K_{\mathrm{id}}\sim\Pois(\mu)$, $\mathbb{E}[2^{K_{\mathrm{id}}}]
=
e^{-\mu}
\sum_{k=0}^{\infty}
\frac{(2\mu)^k}{k!}
=
e^\mu$. Therefore, by Markov's
inequality, $\eta_{\mathrm{PoisTrunc}}
=
\Pr[K_{\mathrm{id}}>K_{\max}]
\le e^\mu 2^{-K_{\max}}
=
e^\mu 2^{-4\mu}
=
e^{-(4\ln 2-1)\mu}
\le e^{-\mu}
\le \frac{\eta_{\mathrm{fail}}^\star}{3}$,
where the penultimate inequality uses $4\ln 2-1>1$.

It remains to control the threshold error and the numerical approximation
error. The value of $q$ was chosen to satisfy both the threshold-error and
numerical-error requirements simultaneously. Since
$2^q\geq 3K_{\max}/\eta_{\mathrm{fail}}^\star$, 
$\eta_{\mathrm{Thresh}}
=K_{\max}2^{-q}
\leq\eta_{\mathrm{fail}}^\star/3$.
Together with the bounds above,
$\eta_{\mathrm{fail}}
\leq\eta_{\mathrm{fail}}^\star$.

The definition of $q$ also gives
$2^q\geq
(1+2\pi)\sqrt{\mu}/\varepsilon_{\mathrm{approx}}^\star$.
Therefore
$\varepsilon_{\mathrm{approx}}
=(1+2\pi)\sqrt{\mu}\,2^{-q}
\leq\varepsilon_{\mathrm{approx}}^\star$.

Finally, we prove the parameter scaling as the targets $\varepsilon_{\mathrm{approx}}^\star$ and $\eta_{\mathrm{fail}}^\star$ approach 0. From the choice of $\mu$,
we have
$\mu=O(\log(1/\eta_{\mathrm{fail}}^\star))$, and since
$K_{\max}=4\mu$, the same bound holds for $K_{\max}$.
Moreover,
$q=
O(
\log(K_{\max}/\eta_{\mathrm{fail}}^\star)
+
\log(\sqrt{\mu}/\varepsilon_{\mathrm{approx}}^\star)
)$.
Since
$\log K_{\max}$ and $\log\mu$ are
$O(\log\log(1/\eta_{\mathrm{fail}}^\star))$, this simplifies to $q
=
O\left(
\log(1/\varepsilon_{\mathrm{approx}}^\star)
+
\log(1/\eta_{\mathrm{fail}}^\star)
\right)$.
\end{proof}

\subsubsection{Coherent implementation in \texorpdfstring{$\MH$}{MH}}
\label{sssec:gaussian-coherent-implementation}

We now give a coherent implementation of the finite Poisson sampler from
\Cref{def:finite-poisson-sampler}. For a fixed input tape $\Omega$, the
sampler is a deterministic computation of
$(R_{\mathrm{fin}},\Theta_{\mathrm{fin}})$. The phase
$\Theta_{\mathrm{fin}}$ is already contained in the input tape, so the main
task is to compute the magnitude $R_{\mathrm{fin}}
=
\min_{1\leq i\leq K_{\mathrm{fin}}}
\max\{U^{\mathrm{fin}}_{i,0},U^{\mathrm{fin}}_{i,1}\}$.

A straightforward implementation would first determine
$K_{\mathrm{fin}}$, compute the maximum within each of the first
$K_{\mathrm{fin}}$ candidate pairs, and then compare these maxima to find
their minimum. If the minimum is computed by updating a running minimum, the
comparisons are necessarily sequential, giving nonlinear depth that grows
with $K_{\max}$. A balanced comparison tree reduces this to $O(\log K_{\max})$ sequential
comparison stages, but this still does not give a constant magic-hierarchy
level.

Instead of reducing the candidate values through a sequence of pairwise
comparisons, we trade depth for size. Every candidate is tested against all
other active candidates in parallel to determine whether it is the minimum. We first determine which of the $K_{\max}$ candidate pairs are active from
the comparisons $Z\geq T_i$ and which coordinate is the maximum within each
active pair. We then compare every resulting candidate against all other
active pairs simultaneously. A fixed tie-breaking rule ensures that exactly
one candidate is selected, and the bits of this winning candidate give
$R_{\mathrm{fin}}$. We refer to this as the parallel winner-selection construction.

The following lemma gives the resulting coherent implementation of the finite sampler in the magic hierarchy.

\begin{lemma}[Coherent implementation of the finite Poisson sampler in $\MH$]
\label{lem:finite-poisson-coherent-implementation}
Fix sampler parameters $\mu$, $K_{\max}$, and $q$, and let
$R=(2K_{\max}+2)q$ be the random-tape length. There exists an exact nonclean
coherent implementation
$U_{\mathrm{samp}}$ for the finite Poisson sampler $\Samp_{\mathrm{fin}}:\{0,1\}^R\to\{0,1\}^{2q}$
defined in \Cref{def:finite-poisson-sampler} that, for every input tape
$\Omega\in\{0,1\}^R$, acts in the following form:
\begin{equation}
U_{\mathrm{samp}}
\ket{\Omega}\ket{0^{n_{\mathrm{anc}}}}\ket{0^{2q}}
=
\ket{\Omega}\ket{\mr{junk}_{\Omega}}
\ket{\Samp_{\mathrm{fin}}(\Omega)}.
\end{equation}
Here $n_{\mathrm{anc}}$ is the number of ancilla qubits, $\ket{\mr{junk}_{\Omega}}$ denotes the resulting state of
these ancilla qubits, which may depend on $\Omega$, and the final $2q$
qubits contain the sampler output
$\Samp_{\mathrm{fin}}(\Omega)
=(R_{\mathrm{fin}},\Theta_{\mathrm{fin}})$, with $q$ bits for each
component.

The circuit $U_{\mathrm{samp}}$ lies in
$\AxCQ{8}\subseteq\MH_8$, with $n_{\mathrm{anc}}
=
O\left(
K_{\max}^2q^4\log(q+1)
+
K_{\max}^2q\log(K_{\max}+1)
\right)$ ancilla qubits and circuit size.
\end{lemma}

To prove this lemma, we first formalize the efficient parallel computation of the
magnitude $R_{\mathrm{fin}}$ described above. Recall that when the number of candidate points $K_{\mathrm{fin}}>0$,
$R_{\mathrm{fin}}$ is obtained by taking the maximum within each of the
first $K_{\mathrm{fin}}$ candidate pairs and then taking the minimum of
these pair maxima. We therefore need to identify which pairs are active,
determine the maximum within each active pair, and then select the smallest
of these maxima in parallel. Recall that the input tape has the form
$\Omega=(Z,U^{\mathrm{fin}}_{1,0},U^{\mathrm{fin}}_{1,1},\ldots,
U^{\mathrm{fin}}_{K_{\max},0},U^{\mathrm{fin}}_{K_{\max},1},
\Theta_{\mathrm{fin}})$, where $Z$ determines the active candidate pairs,
the $U^{\mathrm{fin}}_{i,a}$ are the candidate coordinates, and
$\Theta_{\mathrm{fin}}$ is the phase. 

For each candidate pair $i\in[K_{\max}]$, define the activity bit
$A_i\in\{0,1\}$ by $A_i=1$ if and only if $Z\geq T_i$, where $T_i$ is
the fixed threshold associated with pair $i$ in the finite sampler. Thus
$A_i=1$ exactly when pair $i$ is one of the first $K_{\mathrm{fin}}$
pairs participating in the minimum defining $R_{\mathrm{fin}}$. Since the
thresholds $T_i$ are nondecreasing, the activity bits consist of
$K_{\mathrm{fin}}$ initial ones followed by zeros.

For each pair $i$, let
$M_i=\max\{U^{\mathrm{fin}}_{i,0},U^{\mathrm{fin}}_{i,1}\}$ denote its
pair maximum. For each coordinate $a\in\{0,1\}$, define the selector bit
$S_{i,a}\in\{0,1\}$ so that $S_{i,a}=1$ exactly when
$U^{\mathrm{fin}}_{i,a}$ is chosen to represent $M_i$. We set
$S_{i,0}=1$ if and only if
$U^{\mathrm{fin}}_{i,0}\geq U^{\mathrm{fin}}_{i,1}$ and
$S_{i,1}=1$ if and only if
$U^{\mathrm{fin}}_{i,1}>U^{\mathrm{fin}}_{i,0}$.
Thus exactly one selector bit is $1$, with coordinate $0$ chosen in the
case of a tie. 

Among the active pairs, $R_{\mathrm{fin}}$ is the smallest pair maximum
$M_i$. We select one pair attaining this minimum, breaking ties in favor of
the smaller pair index. To determine whether a candidate
$U^{\mathrm{fin}}_{i,a}$ can be selected, we record whether another active
pair blocks it. For each $j\neq i$, define the blocking bit $B_{j;i,a}\in\{0,1\}$ as
follows. If $j<i$, set
$B_{j;i,a}
=
A_j
\wedge
\mathbf{1}\{U^{\mathrm{fin}}_{j,0}\leq U^{\mathrm{fin}}_{i,a}\}
\wedge
\mathbf{1}\{U^{\mathrm{fin}}_{j,1}\leq U^{\mathrm{fin}}_{i,a}\}$.
If $j>i$, set
$B_{j;i,a}
=
A_j
\wedge
\mathbf{1}\{U^{\mathrm{fin}}_{j,0}<U^{\mathrm{fin}}_{i,a}\}
\wedge
\mathbf{1}\{U^{\mathrm{fin}}_{j,1}<U^{\mathrm{fin}}_{i,a}\}$,
where $\mathbf{1}\{\cdot\}$ denotes the bit indicating whether the stated
condition holds. These conditions compare the pair maximum $M_j$ with the candidate
$U^{\mathrm{fin}}_{i,a}$. The non-strict comparison for $j<i$ enforces the
tie-breaking rule by allowing an earlier pair to block a candidate with the
same value, while a later pair blocks the candidate only if its pair maximum
is strictly smaller. The selector bit $S_{i,a}$ ensures that only the
coordinate equal to the pair maximum $M_i$ can ultimately be selected.

Finally, define the winner bit $W_{i,a}\in\{0,1\}$ by
\begin{equation}
W_{i,a}
=
A_iS_{i,a}
\prod_{j\neq i}(1\oplus B_{j;i,a}).
\label{eq:poisson-winner-bit}
\end{equation}
Thus $W_{i,a}=1$ exactly when pair $i$ is active,
$U^{\mathrm{fin}}_{i,a}=M_i$, and no other active pair blocks it.

\begin{claim}[Correctness of parallel winner selection]
\label{clm:poisson-winner-correctness}
The winner bits $W_{i,a}$ give an exact parallel representation of the
magnitude $R_{\mathrm{fin}}$. For every input tape $\Omega$, if
$K_{\mathrm{fin}}>0$, exactly one winner bit $W_{i,a}$ is $1$, and its
corresponding value $U^{\mathrm{fin}}_{i,a}$ is $R_{\mathrm{fin}}$.
If $K_{\mathrm{fin}}=0$, all winner bits are zero. Consequently, for every
$b\in[q]$,
\begin{equation}
(R_{\mathrm{fin}})_b
=
\bigoplus_{i=1}^{K_{\max}}
\bigoplus_{a\in\{0,1\}}
(U^{\mathrm{fin}}_{i,a})_b W_{i,a}.
\label{eq:poisson-minimum-selection}
\end{equation}
\end{claim}
\begin{proof}[Proof of~\Cref{clm:poisson-winner-correctness}]
Suppose first that at least one pair is active. Let $i^\star$ be the
smallest index among the active pairs attaining the minimum pair maximum,
and let $a^\star$ be the coordinate selected within pair $i^\star$.
No earlier active pair has maximum at most
$U^{\mathrm{fin}}_{i^\star,a^\star}$, and no later active pair has a
strictly smaller maximum. Hence all blocking bits appearing in
$W_{i^\star,a^\star}$ are zero. Since
$A_{i^\star}=S_{i^\star,a^\star}=1$,
\eqref{eq:poisson-winner-bit} gives
$W_{i^\star,a^\star}=1$. Every other winner bit is zero. An inactive pair has $A_i=0$, the
unselected coordinate within any pair has $S_{i,a}=0$, and any other
selected active maximum is either strictly larger than the minimum or is an
equal minimum with a larger pair index. In either case it has a blocking bit
equal to $1$. Thus $W_{i^\star,a^\star}$ is the unique winner.

If no pair is active, every $A_i=0$, so every winner bit is zero.
Equation~\eqref{eq:poisson-minimum-selection} then copies the bits of the
unique winner when one exists and outputs zero otherwise, exactly giving
$R_{\mathrm{fin}}$.
\end{proof}

The circuit uses the thresholds $T_i$ as hardwired classical constants.
We therefore also need to verify that these thresholds can be computed
efficiently from the sampler parameters.

\begin{claim}[Efficient threshold precomputation]
\label{clm:poisson-threshold-precomputation}
For fixed $\mu$, $K_{\max}$, and $q$, the thresholds
$T_1,\ldots,T_{K_{\max}}$ satisfying
\eqref{eq:poisson-threshold-accuracy} can be chosen nondecreasing and
computed classically in time polynomial in $\mu$, $K_{\max}$, and $q$.
\end{claim}

\begin{proof}[Proof of~\Cref{clm:poisson-threshold-precomputation}]
Recall that
$P_i=\Pr[K_{\mathrm{id}}<i]
=e^{-\mu}\sum_{k=0}^{i-1}\mu^k/k!$.
Since the finite sum is explicit, it suffices to approximate the factor
$e^{-\mu}$ to sufficiently high precision. Choose a nonnegative rational approximation $\widehat e_\mu$ satisfying
$\abs{\widehat e_\mu-e^{-\mu}}\leq2^{-q-2-2\mu}$, and set
$\widehat P_i=\widehat e_\mu\sum_{k=0}^{i-1}\mu^k/k!$.
Since
$\sum_{k=0}^{i-1}\mu^k/k!\leq e^\mu\leq2^{2\mu}$,
we have $\abs{\widehat P_i-P_i}\leq2^{-q-2}$. Let $T_i$ be the nearest integer to $2^q\widehat P_i$, clipped to
$\{0,\ldots,2^q\}$ using a fixed rule for ties. Since $P_i\in[0,1]$,
clipping cannot increase the distance from $P_i$. Nearest-integer rounding
contributes at most $2^{-q-1}$ after rescaling by $2^{-q}$, so
$\abs{T_i/2^q-P_i}<2^{-q}$, as required. Since $\widehat P_i$ is nondecreasing in $i$, nearest-integer rounding with
a fixed tie rule, followed by clipping, preserves this ordering. Hence the
thresholds $T_i$ are nondecreasing. The approximation $\widehat e_\mu$ and
the finite sums above can be computed to the required precision in time
polynomial in $\mu$, $K_{\max}$, and $q$. This proves the claim.
\end{proof}

We now use the two claims above to give the magic hierarchy implementation of the sampler.

\begin{proof}[Proof of \Cref{lem:finite-poisson-coherent-implementation}]
We now give the full implementation for the finite sampler from~\Cref{def:finite-poisson-sampler} coherently in the magic hierarchy, which outputs the magnitude $R_{\mathrm{fin}}$ and phase $\Theta_{\mathrm{fin}}$ on $q$-bit registers, given an input tape $\Omega$. 

We compute the magnitude
$R_{\mathrm{fin}}$ using the parallel winner-selection procedure as described above in \Cref{clm:poisson-winner-correctness}. It consists of two sequential stages of multi-input conjunctions. The first computes the
conjunction terms needed to obtain the activity, selector, and blocking
bits. After Clifford parity operations assemble these bits, the second
conjunction stage computes the masked output bits of the winning candidate. Each of these conjunction terms can be computed using the polynomial-size nonclean
$\AxCQ{4}$ implementation of AND from the proof of Lemma~6.5 in
\cite{parham2025mh}. By
\Cref{clm:poisson-threshold-precomputation}, the thresholds $T_i$ used in
computing the magnitude can be efficiently computed in advance and
hardwired into the circuit. Lastly, the phase $\Theta_{\mathrm{fin}}$ is already
contained in the input tape and so it is copied directly to the output using
CNOT gates.

\begin{enumerate}
    \item{} [$\Cliff$] Prepare the Boolean expressions in the
    comparison formulas of the activity, selector, and blocking bits appearing in the parallel
winner-selection formula~\eqref{eq:poisson-winner-bit} for computing the magnitude.

    The required comparisons are $Z\geq T_i$ for the activity bits $A_i$;
$U^{\mathrm{fin}}_{i,0}\geq U^{\mathrm{fin}}_{i,1}$ and
$U^{\mathrm{fin}}_{i,1}>U^{\mathrm{fin}}_{i,0}$ for the selector bits
$S_{i,0}$ and $S_{i,1}$; and, for the blocking bits $B_{j;i,a}$,
$U^{\mathrm{fin}}_{j,0}\leq U^{\mathrm{fin}}_{i,a}$ and
$U^{\mathrm{fin}}_{j,1}\leq U^{\mathrm{fin}}_{i,a}$ when $j<i$, or
$U^{\mathrm{fin}}_{j,0}<U^{\mathrm{fin}}_{i,a}$ and
$U^{\mathrm{fin}}_{j,1}<U^{\mathrm{fin}}_{i,a}$ when $j>i$.

    To express these comparisons in a form that can be evaluated in parallel,
consider two $q$-bit integers $a=a_1\cdots a_q$ and
$b=b_1\cdots b_q$, written most significant bit first. The relation $a<b$
is determined by the first bit at which $a$ and $b$ differ. Thus
$\mathbf{1}\{a<b\}$ is the XOR over $\ell\in[q]$ of
$(1\oplus a_\ell)b_\ell
\prod_{h<\ell}(1\oplus a_h\oplus b_h)$.
Each term checks that all more significant bits agree and that the first
difference has $a_\ell=0$ and $b_\ell=1$. Since there is at most one first
differing position, these terms are mutually exclusive. Hence $a<b$ is a parity of $O(q)$ conjunctions, each containing
$O(q)$ affine Boolean expressions in the input bits. The non-strict
comparisons used above are obtained by additionally including the equality
case, while reversing the inputs gives the corresponding $>$ and $\geq$
comparisons. Thus all of the required comparisons have the same
parity-of-conjunctions form.

 Compute these affine expressions using CNOT and X gates and fan out
    the required values so that the conjunction terms in the next step can
    be evaluated in parallel. For the activity bit comparisons $Z\geq T_i$, represent both $Z$ and $T_i$
    using $q+1$ bits and pad $Z$ with a leading zero, which also covers the
    case $T_i=2^q$.

    \item{} [$\AxCQ{4}$] Compute the conjunction terms needed for the
    activity, selector, and blocking bits.

    Using the polynomial-size nonclean $\AxCQ{4}$ implementation of AND from
the proof of \cite[Lemma~6.5]{parham2025mh}, each conjunction term from
the previous step can be computed in $\AxCQ{4}$. Since all of these terms
use separate workspace, they can be computed in parallel within the same
$\AxCQ{4}$ block.
    
    We now identify the conjunction terms required for each type of bit.
The activity and selector bits each require the terms from a single
comparison. Each blocking bit $B_{j;i,a}$ instead combines the activity
comparison defining $A_j$ with two coordinate comparisons against
$U^{\mathrm{fin}}_{i,a}$. Each of these three comparisons is a
parity of $O(q)$ conjunctions. Distributing their
    product over $\mathbb F_2$ therefore expresses each $B_{j;i,a}$ as a
    parity of $O(q^3)$ conjunctions, each involving $O(q)$ affine Boolean
    expressions. Hence all conjunction terms needed for the $A_i$,
    $S_{i,a}$, and $B_{j;i,a}$ can be computed simultaneously in this
    single $\AxCQ{4}$ block.

    \item{} [$\Cliff$] Assemble the activity, selector, and blocking bits.

    XOR the corresponding conjunction outputs to obtain all $A_i$,
    $S_{i,a}$, and $B_{j;i,a}$. Then compute the literals
    $1\oplus B_{j;i,a}$ and fan out the values needed in the winner
    computation.

    \item{} [$\AxCQ{4}$] Compute the masked output bits of the winning
    candidate in~\Cref{eq:poisson-winner-bit}.

    For every $i\in[K_{\max}]$, $a\in\{0,1\}$, and $b\in[q]$, compute
    $(U^{\mathrm{fin}}_{i,a})_bW_{i,a}$. By
    \eqref{eq:poisson-winner-bit}, each such term is a conjunction of
    $O(K_{\max})$ Boolean literals. Hence all masked output bits can be
    computed in parallel in $\AxCQ{4}$~\cite[Lemma~6.5]{parham2025mh}.

    \item{} [$\Cliff$] Assemble the output registers for the magnitude $R_{\mathrm{fin}}$ and phase $\Theta_{\mathrm{fin}}$.

    For each $b\in[q]$, XOR the masked output bits over $i$ and $a$. By
    \eqref{eq:poisson-minimum-selection}, this gives exactly
    $R_{\mathrm{fin}}$. Copy the $q$ input bits representing
    $\Theta_{\mathrm{fin}}$ into the phase-output register.
\end{enumerate}

By merging the adjacent Clifford
layers around the $\AxCQ{4}$ blocks from the conjunction computation, this overall circuit lies in
$\AxCQ{8}\subseteq\MH_8$.

We finally bound the circuit size and number of auxiliary qubits. There are
$O(K_{\max}^2)$ blocking bits. Expanding the three comparison predicates
defining each blocking bit gives
$O(K_{\max}^2q^3)$ conjunctions in the first conjunction stage, each
involving $O(q)$ affine expressions. 
This contributes
$O(K_{\max}^2q^4\log(q+1))$ gates and auxiliary qubits. The second conjunction stage contains $O(K_{\max}q)$ conjunctions, one for
each candidate and output bit, and each contains $O(K_{\max})$ literals.
Its cost is therefore
$O(K_{\max}^2q\log(K_{\max}+1))$.
The intervening Clifford computations are asymptotically smaller. Hence the
complete implementation has size and auxiliary-qubit count
$O(K_{\max}^2q^4\log(q+1)
+
K_{\max}^2q\log(K_{\max}+1))$.
\end{proof}

Combining the approximation guarantee of
\Cref{lem:finite-complex-gaussian-sampler}, the parameter choices of
\Cref{cor:finite-sampler-parameter-choice}, and the coherent
magic-hierarchy implementation of
\Cref{lem:finite-poisson-coherent-implementation} gives the following
summary guarantee for the finite Poisson sampler.

\begin{corollary}[Complex Gaussian sampler guarantees]
\label{lem:poisson-minimum-gaussian-sampler}
For any target approximation error
$\varepsilon_{\mathrm{approx}}^\star>0$ and target failure probability
$\eta_{\mathrm{fail}}^\star\in(0,1)$, the parameters
$\mu$, $K_{\max}$, and $q$ can be chosen so that, with
$s=\sqrt{\mu}$ and $R=(2K_{\max}+2)q$, the finite Poisson sampler
$\Samp_{\mathrm{fin}}:\{0,1\}^R\to\{0,1\}^{2q}$ is an
$(s,\varepsilon_{\mathrm{approx}}^\star,
\eta_{\mathrm{fail}}^\star)$-complex Gaussian sampler in the sense of
\Cref{def:complex-gaussian-sampler}.

Moreover, every sampler output has magnitude $\abs{V_{\mathrm{fin}}}<1$,
$R$ and $q$ are polynomial in
$\log(1/\varepsilon_{\mathrm{approx}}^\star)
+\log(1/\eta_{\mathrm{fail}}^\star)$,
the sampler can be evaluated classically in polynomial time, and it has an
exact nonclean coherent implementation in
$\AxCQ{8}\subseteq\MH_8$ of polynomial size using polynomially many
auxiliary qubits.
\end{corollary}

\begin{proof}
Choose $\mu$, $K_{\max}$, and $q$ according to
\Cref{cor:finite-sampler-parameter-choice}. Then
\Cref{lem:finite-complex-gaussian-sampler} gives a coupling of
$V_{\mathrm{fin}}$ with $G\sim\mc{N}^{\CC}(0,1)$ such that
$\abs{\sqrt{\mu}V_{\mathrm{fin}}-G}
\leq\varepsilon_{\mathrm{approx}}^\star$
except with probability at most $\eta_{\mathrm{fail}}^\star$.
Thus, with $s=\sqrt{\mu}$,
$\Samp_{\mathrm{fin}}$ satisfies the approximation requirement of
\Cref{def:complex-gaussian-sampler}. 

By \Cref{def:finite-poisson-sampler},
$R_{\mathrm{fin}}<1$, so every output
$V_{\mathrm{fin}}
=R_{\mathrm{fin}}e^{2\pi i\Theta_{\mathrm{fin}}}$
has magnitude strictly less than one. The parameter bounds in
\Cref{cor:finite-sampler-parameter-choice} give
$R=(2K_{\max}+2)q$ and $q$ polynomial in
$\log(1/\varepsilon_{\mathrm{approx}}^\star)
+\log(1/\eta_{\mathrm{fail}}^\star)$.
The sampler is classically evaluable in polynomial time, including the
precomputation of the thresholds $T_i$ from
\Cref{clm:poisson-threshold-precomputation}. Finally,
\Cref{lem:finite-poisson-coherent-implementation} gives the exact nonclean
$\AxCQ{8}$ implementation with polynomial size and polynomially many
auxiliary qubits.
\end{proof}

\paragraph{Compatibility with security parameters.}
We finally verify that the finite Poisson sampler satisfies the parameter
and classical-computation requirements used in
\Cref{cor:random-function-security-parameters} and
\Cref{thm:PRAL-security}. In the notation of the finite sampler above, set
$\varepsilon_{\mathrm{approx}}^\star=\xi$ and
$\eta_{\mathrm{fail}}^\star=\eta_{\mathrm{samp}}$, and choose
$\xi=\eps/(16CN^{3/2}2^T)$ and
$\eta_{\mathrm{samp}}=\eps/(4N)$. Since $N=2^n$, these choices give
$\mu,K_{\max}=O(n+\log(1/\eps))$ and
$q=O(n+T+\log(1/\eps))$.
Hence the tape length $R=(2K_{\max}+2)q$, the output length $2q$, and the
classical evaluation time are polynomial in
$n$, $T$, and $\log(1/\eps)$. The sampler outputs finite polar coordinates
$(R_{\mathrm{fin}},\Theta_{\mathrm{fin}})$. The corresponding real and
imaginary parts,
$R_{\mathrm{fin}}\cos(2\pi\Theta_{\mathrm{fin}})$ and
$R_{\mathrm{fin}}\sin(2\pi\Theta_{\mathrm{fin}})$, can be approximated to
any requested $p$ bits in time polynomial in $q$ and $p$.
Finally, with $\gamma=1/(2s)$ and $s=\sqrt{\mu}$, we have
$\gamma^2=1/(4\mu)$, which is rational and efficiently computable.
Thus the concrete finite Poisson sampler satisfies the sampler and
classical-computation requirements used in
\Cref{thm:PRAL-security}.

\subsection{Amplitude loading in \MH}\label{ssec:amplitude-loading-in-mh}
Now that we have computed the amplitudes of our target state, we need to prepare a state with those amplitudes -- to ``load" in those amplitude. We first give the definition of an amplitude loader.
\begin{definition}[Amplitude loader]
\label{def:amplitude-loader}
Let $v=(v_x)_{x\in\zo^n}$ be a complex vector, where $v_x=r_xe^{2\pi i\theta_x}$ and $r_x,\theta_x\in\zo^q$ encode
$q$-bit binary fractions. Let $U_v$ be an
amplitude oracle that, on input $\ket{x}$, coherently outputs
$\ket{r_x}$ and $\ket{\theta_x}$ for each $x\in\zo^n$, together with
an auxiliary workspace state $\ket{\mr{junk}_x}$.

Fix a balance parameter $\gamma\in(0,1]$ and an error parameter
$\delta_{\mr{load}}\in(0,1)$. 

An amplitude loader is a quantum
algorithm that, given forward and inverse access to $U_v$ for amplitude vector
$v$ satisfying $\norm{v}\geq\gamma\sqrt N$, outputs an
$n$-qubit mixed state $\sigma$ satisfying $\frac{1}{2}\tnorm{\sigma-\ket{\wh{v}}\bra{\wh{v}}}
    \leq \delta_{\mr{load}}$, where
$\ket{\wh v}=\frac{1}{\norm{v}}\sum_{x\in\zo^n}v_x\ket{x}$.
\end{definition}

To implement the amplitude loader, we write each magnitude $r_x$ in binary and prepare a weighted superposition over its digit positions. For each digit position, we check whether the corresponding bit of $r_x$ is one, and only those branches contribute when the weighted reference register is unprepared, producing an amplitude proportional to $r_x$. The encoded phase supplies the corresponding complex phase. This follows the linear-combination-of-unitaries state-preparation approach of~\cite{wang2021lcu}. Our implementation adapts this construction to minimize the number of magic-hierarchy levels required.

We first show how to implement the base amplitude loading algorithm in~\Cref{sssec:loader-single-trial}, which prepares $\ket{\wh v}$ with some success probability $p\geq \gamma^2$. To boost the success probability, we run multiple copies of this algorithm in parallel and select the first success one. We describe the repetition and selection procedure in~\Cref{sssec:loader-repetition}. Finally we show how to implement the full amplitude loader in~\Cref{sssec:full-amplitude-loader}.

\subsubsection{Amplitude loading for one trial}
\label{sssec:loader-single-trial}
\begin{lemma}[Amplitude loading for one trial]\label{lem:single-loading-trial}
Suppose we have access to an amplitude oracle $U_v$ and its inverse,
outputting $r_x$ and $\theta_x$ on $q$ qubits, with
$U_v,U_v^\dagger\in\AxCQ{\sigma}$ for even $\sigma\geq 2$. Assume
$\norm{v}\geq\gamma\sqrt N$. 

Then there exists a circuit in $\AxCQ{2\sigma + 2}$ that prepares $\ket{\wh{v}}=\frac{1}{\norm{v}}\sum_x v_x\ket{x}$ with probability $p=\frac{\norm{v}^2}{2^n(1-2^{-q})^2}\geq \gamma^2$. This circuit uses one call each to $U_v$ and $U_v^\dagger$, together with $O\pbra{n+q^2\log(q+1)}$ additional gates and $O(q^2)$ additional ancilla qubits.
\end{lemma}
\begin{proof}
    We use the following registers: $D$ is the $n$-qubit data register, $M$ and $P$ are the $q$-qubit
magnitude and phase registers, respectively. The address register $J$ introduced below
has $\log_2q$ qubits, and $B$ is a single flag qubit.

We begin by preparing an equal superposition over $x$ and applying the
amplitude oracle $U_v$, giving
$\frac{1}{\sqrt N}\sum_x
\ket{x}_D\ket{r_x}_M\ket{\theta_x}_P\ket{\mr{junk}_x}_G$.
    Here $r_x$ and $\theta_x$ are stored as $q$-bit binary fractions. Our goal is to convert these binary descriptions into the amplitude $v_x=r_xe^{2\pi i\theta_x}$ on each $\ket{x}$ on the data register $D$. Applying the phase $e^{2\pi i\theta_x}$ from the bits of $\theta_x$ is straightforward, so we focus on the magnitudes.
    
    Write each magnitude as a $q$-bit binary fraction $r_x=0.r_{x,0}r_{x,1}\cdots r_{x,q-1}=\sum_{j=0}^{q-1}2^{-(j+1)}r_{x,j}$. Thus the $j$th magnitude bit contributes the binary weight
$2^{-(j+1)}$ whenever $r_{x,j}=1$. Without loss of generality, take $q$ to be a power of two. Otherwise, pad
the magnitude encoding with trailing zeros to the next power-of-two length
and redefine $q$ to be this padded length.
    
    In the address register, we prepare a weighted superposition $\ket{w}=\sum_j w_j\ket{j}$ over the $q$ fractional-bit positions, where position $j$ has amplitude proportional to the square root of its binary weight $2^{-(j+1)}$, so $w_j
    =
    \frac{2^{-(j+1)/2}}{\sqrt{1-2^{-q}}}$. Note that $\ket{w}$ depends only on the precision $q$, not on $x$ or $r_x$.

    We next introduce a flag qubit $B$ initialized to $\ket{0}$ which indicates which positions of $r_x$ are 1. This will split up the state into positions which contribute to the magnitude and those that do not. For each fixed $x$, the $J,B$ registers therefore become $\sum_{j:r_{x,j}=0}w_j\ket{j}_J\ket{0}_B+ \sum_{j:r_{x,j}=1}w_j\ket{j}_J\ket{1}_B$.
    To compute $r_{x,j}$ into the flag $B$ for each position $j$, we need to perform a coherent lookup from the magnitude and address registers. Conceptually, this lookup consists of one generalized Toffoli for each
$j\in\{0,\ldots,q-1\}$. The gate associated with $j$ flips $B$ exactly
when $J=j$ and $r_{x,j}=1$. Since exactly one address value $j$ is present
on each computational-basis branch of $J$, the resulting flag contains
precisely the addressed magnitude bit. We give the low-depth
magic-hierarchy implementation of this lookup below.

Applying the phase $e^{2\pi i\theta_x}$ from the phase register, the full
state is now
\begin{equation}
    \frac{1}{\sqrt{N(1-2^{-q})}}
    \sum_x
    e^{2\pi i\theta_x}
    \ket{x}_D
    \left(
        \sum_{j:r_{x,j}=0}
        2^{-(j+1)/2}\ket{j}_J\ket0_B
        +
        \sum_{j:r_{x,j}=1}
        2^{-(j+1)/2}\ket{j}_J\ket1_B
    \right)
    \ket{r_x}_M
    \ket{\theta_x}_P
    \ket{\mathrm{junk}_x}_G.
\end{equation}

We next apply $U_v^\dagger$ to erase the magnitude, phase, and oracle
workspace registers. We then apply $W^\dagger$ to the address register
$J$, where $W\ket{0}=\ket{w}$. We call the branch
$J=0^{\log_2 q}$ and $B=1$ the successful branch.

Fix a basis string $x$. Only positions $j$ with $r_{x,j}=1$ have
$B=1$. Before applying $W^\dagger$, the contribution from such a position
is $w_j\ket{j}_J\ket{1}_B$. Since
$\bra{0}W^\dagger\ket{j}=\braket{w|j}=w_j$, its contribution to the
successful branch after $W^\dagger$ is weighted by
$w_j^2=2^{-(j+1)}/(1-2^{-q})$. Summing over the positions for which
$r_{x,j}=1$ therefore gives total magnitude
$\frac{1}{1-2^{-q}}\sum_j2^{-(j+1)}r_{x,j}
=\frac{r_x}{1-2^{-q}}$. Together with the phase $e^{2\pi i\theta_x}$, the contribution of basis
state $x$ to the successful branch is therefore
$\frac{v_x}{\sqrt{N}(1-2^{-q})}$. Hence the full successful component is
\begin{equation}
    \frac{1}{\sqrt{N}(1-2^{-q})}
    \sum_x
    v_x\ket{x}_D
    \ket{0^{\log_2q}}_J
    \ket1_B.
    \label{eq:loader-successful-component}
\end{equation}
All remaining terms are orthogonal to this component on the $J,B$
registers. Therefore, conditioned on the successful branch
$J=0^{\log_2q}$ and $B=1$, the normalized state of the data register is $\frac{1}{\norm{v}}
    \sum_x v_x\ket{x}
    =
    \ket{\wh v}$.

The probability of this successful outcome is the squared norm of
\eqref{eq:loader-successful-component}, which is $p
    =
    \frac{\norm{v}^2}
    {N(1-2^{-q})^2}$.
Since $\norm{v}\geq\gamma\sqrt N$ and
$1-2^{-q}\leq1$, the success probability is at least $p\geq\gamma^2$.

\paragraph{Magic-hierarchy implementation for one amplitude loader trial.}
Here we show that the above implementation is in $\AxCQ{2\sigma+2}$.

\begin{enumerate}
    \item{} [$\AxCQ{\sigma}$] Evaluate $U_v$ and prepare $\ket{w}$.

The Hadamards preparing the equal superposition over $x$ are Clifford and
can be absorbed into the initial Clifford block of $U_v$.

Recall that
$\ket{w}=\sum_{j=0}^{q-1}w_j\ket{j}$, where
$w_j=2^{-(j+1)/2}/\sqrt{1-2^{-q}}$.
To see that $\ket{w}$ is a product state, write the address as
$j=\sum_{h=0}^{\log_2 q-1}j_h2^h$, where each $j_h\in\{0,1\}$.
The unnormalized amplitude of $\ket{j}$ then factors as
$2^{-(j+1)/2}
=
2^{-1/2}\prod_h 2^{-j_h2^h/2}$.
Thus the contribution of each address bit factors independently. If
$j_h=0$, its factor is $1$, while if $j_h=1$, its factor is
$2^{-2^h/2}$. Therefore the $h$th address qubit can be prepared in a state
proportional to $\ket0+2^{-2^h/2}\ket1$. Since $q$ is a power of two,
$\prod_{h=0}^{\log_2q-1}(1+2^{-2^h})=2(1-2^{-q})$.
Therefore the resulting normalized product state is exactly
\begin{equation*}
    \ket{w}
    =
    \bigotimes_{h=0}^{\log_2 q-1}
    \frac{\ket0+2^{-2^h/2}\ket1}
    {\sqrt{1+2^{-2^h}}}.
\end{equation*}
Hence $\ket{w}$ can be prepared exactly by applying one single-qubit
rotation to each address qubit, with all $\log_2 q$ rotations performed in
a single $\QNC^0$ layer.

The address register on which $\ket{w}$ is prepared is disjoint from the registers on which $U_v$ acts,
so this layer can be performed at the same time as one of the $\QNC^0$
blocks of $U_v$. Thus evaluating $U_v$ and preparing $\ket{w}$ together
remain in $\AxCQ{\sigma}$.

\item{} [$\AxCQ{2}$] Perform the addressed magnitude-bit lookup and
    apply the phase.

    For each address $j\in\{0,\ldots,q-1\}$, the lookup flips $B$ exactly
    when $J=j$ and $r_{x,j}=1$. Conjugating $B$ by Hadamards turns these
    $q$ generalized Toffoli gates into commuting generalized controlled-$Z$
    gates. We implement these controlled phases using the parity-based exact
    $\mathrm{OR}$ construction of~\cite[Theorem~1]{TT16} and \cite[Lemma~6.1]{parham2025mh}. The parity terms for all
    $q$ controlled phases are computed into separate ancillas in a Clifford
    block. Then, all corresponding phase rotations are applied in one
    $\QNC^0$ layer, and the parity ancillas are then uncomputed. Hence the
    full addressed lookup lies in $\AxCQ{2}$.

    In the same $\QNC^0$ block, apply the phase
    $e^{2\pi i\theta_x}$. Writing
    $\theta_x=\sum_{k=0}^{q-1}2^{-(k+1)}\theta_{x,k}$, this is implemented
    by applying a one-qubit phase rotation of angle
    $2\pi2^{-(k+1)}$ to each phase bit $\theta_{x,k}$.

    \item{} [$\AxCQ{\sigma}$] Apply $U_v^\dagger$ to uncompute and apply $W^\dagger$ to load the magnitude.

    Apply $U_v^\dagger$ to return the magnitude, phase, and oracle
    workspace registers exactly to zero. The address register $J$ is
    disjoint from the oracle workspace, so the single $\QNC^0$ layer
    implementing $W^\dagger$ can be scheduled in parallel with one
    arbitrary-gate block of $U_v^\dagger$.

    Hence, applying $U_v^\dagger$ and $W^\dagger$ in parallel is in $\AxCQ{\sigma}$.
\end{enumerate}
Therefore the overall circuit for this single-trial amplitude loader is in $\AxCQ{\sigma}\circ\AxCQ{2}\circ\AxCQ{\sigma}
    \subseteq
    \AxCQ{2\sigma+2}$ 

Preparing the equal superposition uses $n$ Hadamard
gates. Preparing and unpreparing $\ket{w}$ uses $O(\log q)$ one-qubit
rotations, and applying the phase $e^{2\pi i\theta_x}$ uses $O(q)$
one-qubit rotations. For a fixed address $j$, the generalized controlled phase acts on
$d=\log_2q+2$ qubits, which uses $O(d2^d)
    =
    O(q\log(q+1))$
gates and $O(2^d)=O(q)$ ancilla qubits. The addressed magnitude-bit computation uses $O(q^2\log(q+1))$
additional gates and $O(q^2)$ auxiliary qubits. All remaining operations contribute at most $O(n+q)$ gates and
$O(\log q)$ ancillary qubits. Hence, besides the calls to $U_v$ and
$U_v^\dagger$, the complete single-trial circuit uses $O\pbra{n+q^2\log(q+1)}$
additional gates and $O(q^2)$ additional ancillary qubits.
\end{proof}

\subsubsection{Boosting success probability by parallel repetition and selection}
\label{sssec:loader-repetition}

Now that we have shown how to do one trial of the amplitude loader, we will boost its success probability by running multiple copies in parallel, and then coherently selecting one successful copy.

Fix a target failure probability $p_{\mr{fail}}\in(0,1)$. Suppose a
single trial succeeds with probability at least $p\geq\gamma^2$. The
probability that all $t$ trials fail is
$(1-p)^t\leq e^{-pt}\leq e^{-\gamma^2t}$. Therefore, choosing
$t=\ceil{\gamma^{-2}\ln(1/p_{\mr{fail}})}$ ensures that the probability
that all trials fail is at most $p_{\mr{fail}}$.

For trial $i$, denote its data, address, and flag registers by $D^{(i)}$, $J^{(i)}$, and $B^{(i)}$. Recall that a single trial succeeds when $J^{(i)}$ is in the all-zero state and $B^{(i)}$ is in $\ket{1}$. We coherently compute the success indicator $\succ^{(i)}
    =
    \mathbf{1}\!\sbra{J^{(i)}=0^{\log_2 q}
        \text{ and }
        B^{(i)}=1}$ into a fresh ancilla qubit. Equivalently, $\succ^{(i)}$ is the conjunction $B^{(i)}
    \prod_{k=1}^{\log q} (1-J^{(i)}_k)$, where $J^{(i)}_k$ denotes the $k$-th address qubit. Conceptually, this is a generalized Toffoli with
$\log_2 q+1$ controls and one target, controlled on $B^{(i)}=1$ and
$J^{(i)}_k=0$ for every address qubit $k$. By~\cite{TT16} and \cite[Lemma~6.1]{parham2025mh}, this can be implemented in $\AxCQ{2}$.

Since several trials may succeed simultaneously, we will compute a selector bit indicating the first successful trial, $\sel^{(i)}=\succ^{(i)}
    \prod_{h<i}
    \pbra{1-\succ^{(h)}}$. $\sel^{(i)}=1$ exactly when trial $i$ succeeds and every earlier trial fails. Therefore at most one selector bit is equal to one,
and all selector bits are zero when all trials fail. Each selector is a conjunction of at most $t$ literals, so by~\cite[Lemma~6.5]{parham2025mh}, all selector bits can therefore be
computed in parallel by a polynomial-size nonclean $\AxCQ{4}$ circuit.

Now we need to select and transfer the data register from the successful trial into an output register. To do this, we introduce a fresh $n$-qubit output register $D^{\mr{out}}$, initialized to the all-zero state. Since at most one of the $t$ select bits is 1, we coherently swap
$D_{\mr{out}}$ with the unique selected data register $D^{(j)}$ when
$\sel^{(j)}=1$, and act trivially if all selector bits are zero. We will use the following circuit to do this.

\begin{claim}[Selected register swap]\label{clm:selected-register-swap}
    Using the data, selector, and output registers defined above, there is a circuit in $\AxCQ{4}$ that swaps $D^{\mr{out}}$ with $D^{(j)}$ when
$\sel^{(j)}=1$, while leaving all other data registers unchanged. If all
selector qubits are $0$, the circuit acts trivially.
\end{claim}
\begin{proof}
For each trial $i\in[t]$, introduce an $n$-qubit temporary register
$T^{(i)}$ initialized to $\ket{0^n}$. The swap is implemented in two
$\AxCQ{2}$ stages: the first transfers the selected data register into
$D^{\mr{out}}$ via a temporary register, and the second clears the
selected source register and the temporary register.

Suppose trial $j$ is selected, so $\sel^{(j)}=1$, and write the state in
its data register as
$\ket{\psi_j}=\sum_z\alpha_z\ket{z}$.

\begin{enumerate}
    \item{} [$\AxCQ{2}$] Transfer the selected data register into
    $D^{\mr{out}}$.

    First fan out each selector qubit $\sel^{(i)}$ into $n$
    copies. Then, for every trial
    $i\in[t]$ and bit position $b\in[n]$, combine the $\AND$ of $\sel^{(i)}$ and $D^{(i)}_b$ into
    $T^{(i)}_b$. After fanout, all $nt$ Toffolis act on disjoint triples
    of qubits, so their constant-depth one- and two-qubit decompositions
    can be performed in parallel in a single $\QNC^0$ block. Reverse
    the selector fanout afterward.

    The temporary registers now contain
    $T^{(i)}_b=\sel^{(i)}D^{(i)}_b$. In particular, since $\sel^{(j)}=1$, the selected registers evolve as
    \begin{equation}
        \ket{\psi_j}_{D^{(j)}}
        \ket{0^n}_{T^{(j)}}
        \ket{0^n}_{D^{\mr{out}}}
        =
        \sum_z\alpha_z
        \ket{z}_{D^{(j)}}
        \ket{0^n}_{T^{(j)}}
        \ket{0^n}_{D^{\mr{out}}}
        \longmapsto
        \sum_z\alpha_z
        \ket{z}_{D^{(j)}}
        \ket{z}_{T^{(j)}}
        \ket{0^n}_{D^{\mr{out}}}.
        \label{eq:selected-swap-temp}
    \end{equation}
    For every $i\neq j$, $\sel^{(i)}=0$, so $T^{(i)}$ remains zero.

    Next XOR all $T^{(i)}$ registers into the initially all-zero
    $D^{\mr{out}}$ using CNOTs on corresponding qubits. Since only
    $T^{(j)}$ is nonzero, \eqref{eq:selected-swap-temp} becomes
    \begin{equation}
        \sum_z\alpha_z
        \ket{z}_{D^{(j)}}
        \ket{z}_{T^{(j)}}
        \ket{z}_{D^{\mr{out}}}.
        \label{eq:selected-swap-forward}
    \end{equation}

    The above steps have the form
    $\Cliff-\QNC^0-\Cliff$, and hence lies in $\AxCQ{2}$.

    \item{} [$\AxCQ{2}$] Clear the selected data and temporary register.

    Fan out the selector qubits and the qubits of $D^{\mr{out}}$ into
    enough copies that the following gates act on disjoint
    triples. For every $i\in[t]$ and $b\in[n]$, apply in parallel one
    Toffoli controlled on $\sel^{(i)}$ and $D^{\mr{out}}_b$, targeting
    $D^{(i)}_b$, and a second Toffoli with the same logical controls
    targeting $T^{(i)}_b$.

    Since $\sel^{(j)}=1$, these gates XOR the contents of
    $D^{\mr{out}}$ into both $D^{(j)}$ and $T^{(j)}$. Thus
    \eqref{eq:selected-swap-forward} transforms as
    \begin{equation}
        \sum_z\alpha_z
        \ket{z}_{D^{(j)}}
        \ket{z}_{T^{(j)}}
        \ket{z}_{D^{\mr{out}}}
        \longmapsto
        \sum_z\alpha_z
        \ket{0^n}_{D^{(j)}}
        \ket{0^n}_{T^{(j)}}
        \ket{z}_{D^{\mr{out}}} =
        \ket{0^n}_{D^{(j)}}
        \ket{0^n}_{T^{(j)}}
        \ket{\psi_j}_{D^{\mr{out}}}.
    \end{equation}
    
    All unselected data registers are unchanged because their selector qubits are zero.

    Finally, reverse the selector and output fanout circuits, returning
    all fanout workspace to zero. This stage is again in $\AxCQ{2}$.
\end{enumerate}

If all selector qubits are zero, no Toffoli acts nontrivially, so the
circuit acts trivially.

Both above stages have the form
$\Cliff-\QNC^0-\Cliff$. When they are composed, the adjacent Clifford
blocks merge into a single Clifford circuit. Hence the complete
selected-register swap circuit lies in $\AxCQ{4}$.
\end{proof}

\subsubsection{Full amplitude-loader implementation}
\label{sssec:full-amplitude-loader}

We now combine the single-trial loading circuit with the parallel
repetition and coherent-selection procedure above to obtain the full
amplitude loader of \Cref{def:amplitude-loader}. Given coherent access to a collection of amplitudes
$\{v_x\}_{x\in\zo^n}$, the loader prepares a state arbitrarily close to
the normalized state proportional to $\sum_x v_x\ket{x}$.

Let $N=2^n$, and let $U_v$ be an amplitude oracle so that on input $x$ it coherently outputs
the polar encoding of an amplitude
$v_x=r_xe^{2\pi i\theta_x}$, with $r_x$ and $\theta_x$ each represented
by a $q$-bit binary fraction. Let $v=(v_x)_{x\in\zo^n}$. 

\begin{lemma}[Amplitude loading in $\MH$]
\label{lem:amplitude-loading-in-mh}
Suppose we are given access to an amplitude oracle $U_v$ for
$v=(v_x)_{x\in\zo^n}$ and its inverse $U_v^\dagger$, where $U_v,U_v^\dagger\in\AxCQ{\sigma}$ for even $\sigma\geq10$.
Assume that $\max_x\abs{v_x}\leq1$ and
$\norm{v}\geq\gamma\sqrt N$ for some balance parameter
$\gamma\in(0,1]$, and define
$\ket{\wh v}=\frac{1}{\norm{v}}\sum_x v_x\ket{x}$. 

For any error parameter $\delta_{\mr{load}}\in(0,1)$, there exists an amplitude loader in $\AxCQ{2\sigma+4}$ whose reduced output state $\rho_v$ satisfies
\begin{equation}
    \frac{1}{2}
    \tnorm{
        \rho_v-\ket{\wh v}\bra{\wh v}
    }
    \leq
    \delta_{\mr{load}}.
\end{equation}

The construction uses
$t=\ceil{\gamma^{-2}\ln(1/\delta_{\mr{load}})}$ parallel repetitions
of the single trial from \Cref{lem:single-loading-trial}, and
hence $t$ parallel calls each to $U_v$ and $U_v^\dagger$, together with
$O\pbra{tq^2\log(q+1)+t^2\log(t+1)+nt}$ additional gates and
$O\pbra{tq^2+t^2\log(t+1)+nt}$ additional ancilla qubits.
\end{lemma}

\begin{proof}
Recall that $p_{\mr{fail}}$ denotes the target upper bound on the
probability that all parallel trials fail. By
\Cref{sssec:loader-repetition}, choosing
$t=\ceil{\gamma^{-2}\ln(1/p_{\mr{fail}})}$ ensures that the all-trails-fail probability
is at most $p_{\mr{fail}}$.

We run $t$ copies of the single-trial construction of
\Cref{lem:single-loading-trial} in parallel. On success, each trial
prepares the target state $\ket{\wh v}$ in its data register. 
Whenever at least one trial succeeds, the first successful trial is
selected and \Cref{clm:selected-register-swap} transfers its data
register to
$D^{\mr{out}}$. Hence, after
tracing out all other registers, the output state has the form
\begin{equation}
    \rho_v
    =
    \pbra{1-(1-p)^t}\ket{\wh v}\bra{\wh v}
    +
    (1-p)^t\rho_{\mr{bad}},
\end{equation}
where $\rho_{\mr{bad}}$ is the reduced output state conditioned on all
trials failing. Hence
\begin{equation}
    \frac{1}{2}
    \tnorm{\rho_v-\ket{\wh v}\bra{\wh v}} =\frac{1}{2}\tnorm{(1-p)^t
    \pbra{
        \rho_{\mr{bad}}
        -
        \ket{\wh v}\bra{\wh v}
    }} =
    \frac{1}{2}(1-p)^t
    \tnorm{\rho_{\mr{bad}}-\ket{\wh v}\bra{\wh v}} \leq
    (1-p)^t
    \leq
    p_{\mr{fail}}.\label{eq:loader-approx-guarantee}
\end{equation}
Setting $p_{\mr{fail}}=\delta_{\mr{load}}$ achieves the approximation
guarantee of \Cref{eq:loader-approx-guarantee}.

\paragraph{Magic-hierarchy level.}
Next we give the full low-depth implementation of the amplitude loader
in $\MH$ and show that it lies in $\AxCQ{2\sigma+4}$. 

Recall that each single trial lies in $\AxCQ{2\sigma+2}$. Running the
$t$ trials in parallel does not increase the $\MH$ depth since they act on disjoint registers. Computing the
success indicators is in $\AxCQ{2}$, computing the selector bits is in
$\AxCQ{4}$, and performing the selected-register swap is in
$\AxCQ{4}$. Naively composing these procedures sequentially would
therefore give an $\AxCQ{2\sigma+12}$ implementation.

We reduce this depth by scheduling part of the selection and transfer
procedure in parallel with the $U_v^\dagger$ step in each trial. After
copying the computational-basis contents of each data register into a
fresh register, the registers used for the success-indicator
computation, selector computation, and first half of the
selected-register swap are disjoint from those on which
$U_v^\dagger$ acts.

As in the single-trial construction, $W^\dagger$ occupies one
$\QNC^0$ layer of $U_v^\dagger$. The success-indicator computation,
selector computation, and first half of the selected-register swap
require one, two, and one additional $\QNC^0$ layers, respectively, together with their intervening
Clifford circuits.

We schedule these four $\QNC^0$ layers, together with the one layer for
$W^\dagger$, into the first five successive $\QNC^0$ layers of
$U_v^\dagger$. The Clifford parts of these constructions are scheduled
in parallel with the existing Clifford layers between the corresponding
$\QNC^0$ layers of $U_v^\dagger$. Hence, when $\sigma\geq10$, these
steps fit within the existing $\AxCQ{\sigma}$ implementation of
$U_v^\dagger$. Therefore, the parallelized trials, together with the
success-indicator computation, selector computation, and first half of
the selected-register swap, remain within $\AxCQ{2\sigma+2}$.

The only additional depth comes from the second half of the
selected-register swap of \Cref{clm:selected-register-swap}, which
clears the selected data and temporary registers in $\AxCQ{2}$. Together with the $\AxCQ{2\sigma+2}$ above, the overall depth of the amplitude loader is thus $\AxCQ{2\sigma+4}$.

For completeness, we list out the full $\MH$ implementation below:

\begin{enumerate}
    \item{} [$\AxCQ{\sigma}$] Run the forward part of all $t$ trials in
    parallel.

    This consists of preparing the equal superpositions and address states
    and applying all copies of $U_v$. Since the trials act on disjoint
    registers, this remains in $\AxCQ{\sigma}$.

    \item{} [$\AxCQ{2}$] Perform the addressed magnitude lookup for all
    trials in parallel.

    By \Cref{lem:single-loading-trial}, the lookup and phase application
    for one trial lie in $\AxCQ{2}$. Performing the corresponding gates
    for all $t$ trials in parallel does not increase the
    magic-hierarchy level.

    \item{} [$\Cliff$] Copy each data register $D^{(i)}$ into a fresh
    $n$-qubit register $C^{(i)}$ using bitwise CNOTs.

    The copy registers $C^{(i)}$ store only the computational-basis labels of
the data registers. The inverse oracle $U_v^\dagger$ leaves each label
$x$ unchanged while erasing its oracle workspace. Thus the operations on
$C^{(i)}$ can be carried out while $U_v^\dagger$ acts on the original
trial registers. After $U_v^\dagger$ finishes, $D^{(i)}$ and $C^{(i)}$
still contain the same basis label, so $C^{(i)}$ can later be erased by
bitwise CNOTs from $D^{(i)}$.

    \item{} [$\AxCQ{\sigma}$] Apply all copies of $U_v^\dagger$ in
    parallel while carrying out the coherent selection procedure.

    During the arbitrary-gate layers of $U_v^\dagger$, we also perform
    $W^\dagger$, compute the success indicators, compute the selectors,
    and perform the first $\AxCQ{2}$ stage of
    \Cref{clm:selected-register-swap} using the copied registers
    $C^{(i)}$. These operations require, respectively, one, one, two,
    and one arbitrary-gate layers, for a total of five. Since
    $\sigma\geq10$, these five layers can be absorbed into
    $U_v^\dagger$, while adjacent Clifford blocks are merged.

    \item{} [$\Cliff$] After $U_v^\dagger$ has erased the oracle
    workspace, erase each copy register $C^{(i)}$ by applying bitwise
    CNOTs from $D^{(i)}$ to $C^{(i)}$.

    \item{} [$\AxCQ{2}$] Perform the second stage of
    \Cref{clm:selected-register-swap}, clearing the selected source and
    temporary registers.
\end{enumerate}

Each of the $t$ parallel trials makes one call each to $U_v$ and
$U_v^\dagger$. By \Cref{lem:single-loading-trial}, the trials together
use $O\pbra{tq^2\log(q+1)+nt}$ additional gates. 
Computing all success
indicators costs $O\pbra{tq\log(q+1)}$ gates, 
and computing all selector bits
uses $O\pbra{t^2\log(t+1)}$ gates. Finally, copying and erasing the data
registers and performing the selected-register swap use $O(nt)$
additional gates. Thus the total number of additional gates is
$O\pbra{tq^2\log(q+1)+t^2\log(t+1)+nt}$. 

By \Cref{lem:single-loading-trial}, each trial uses $O(q^2)$
additional ancilla qubits. Since the $t$ trials are run in parallel on
disjoint registers, they use $O(tq^2)$ ancillas in total. For each
trial, computing the success indicator uses
$2^{\log q+2}-1=O(q)$ parity ancillas
since each indicator acts
on $\log q+2$ qubits. Thus all $t$ success indicators use $O(tq)$
ancillas. Each selector computation uses at most
$O\pbra{t\log(t+1)}$ ancillas, 
so the $t$ selector computations use
$O\pbra{t^2\log(t+1)}$ ancillas in total. The copied data registers, temporary
registers, and fanout workspace used by the selected-register swap
require $O(nt)$ additional qubits. Hence the total number of additional
ancilla qubits is
$O\pbra{tq^2+t^2\log(t+1)+nt}$.
\end{proof}

\subsection{Low-depth preparation of the \PRAL target states}
\label{ssec:end-to-end-preparation}

We now combine the finite Poisson sampler from~\Cref{ssec:complex-gaussian-sampler} with the amplitude loader from~\Cref{ssec:amplitude-loading-in-mh} to obtain low-depth preparation circuits for the target states $\ket{\psi_H}$. We first show this for an arbitrary function $H$ with a low-depth coherent evaluator, and then instantiate $H$ with the pseudorandom tape functions derived from the subexponentially secure PRF family underlying the \PRAL family.

Recall that for a function $H:\zo^n\to\zo^R$, we set
$v_x=\Samp_{\mr{fin}}(H(x))$ and, for a fixed balance threshold
$\gamma$, define the target state $\ket{\psi_H}$ to be the normalized
amplitude state proportional to $\sum_x v_x\ket{x}$ when
$\norm{v}\geq\gamma\sqrt{N}$, and $\ket{0^n}$ otherwise, as in
\Cref{def:target-state-for-H}.

\begin{corollary}[Low-depth preparation of the target states]
\label{cor:end-to-end-state-preparation}
Fix $H:\zo^n\to\zo^R$ and let $\ket{\psi_H}$ be the target state
defined above.
Suppose $H$ has a coherent evaluation circuit
$U_H,U_H^\dagger\in\AxCQ{k_H}$ for even $k_H\geq2$.

Then, for any error parameter
$\delta_{\mr{load}}\in(0,1)$, there exists a circuit in
$\AxCQ{2k_H+20}$ whose reduced output state $\rho_H$ satisfies
\begin{equation}
    \frac{1}{2}
    \tnorm{
        \rho_H-\ket{\psi_H}\bra{\psi_H}
    }
    \leq
    \delta_{\mr{load}}.
\end{equation}
If the coherent evaluator for $H$ has polynomial size and uses
polynomially many ancilla qubits, and the sampler and loader parameters
are polynomially bounded, then the preparation circuit has polynomial
size and uses polynomially many ancilla qubits.
\end{corollary}

\begin{proof}
Let $v_x=\Samp_{\mr{fin}}(H(x))$. First coherently evaluate $H(x)$ and
then apply the nonclean finite Poisson sampler of
\Cref{lem:poisson-minimum-gaussian-sampler} to the resulting tape.
This gives an amplitude oracle
\begin{equation}
    U_v\ket{x}\ket{0}
    =
    \ket{x}\ket{\mr{junk}_x}\ket{\enc(v_x)}.
\end{equation}
The evaluator for $H$ lies in $\AxCQ{k_H}$ and the sampler lies in
$\AxCQ{8}$. Their adjacent Clifford layers merge, so
$U_v\in\AxCQ{k_H+8}$. Reversing the construction gives
$U_v^\dagger$ at the same level.

If $\norm{v}\geq\gamma\sqrt N$, then
$\ket{\psi_H}=\ket{\wh v}$ by
\Cref{def:target-state-for-H}. Since $k_H+8\geq10$,
\Cref{lem:amplitude-loading-in-mh} applies with
$\sigma=k_H+8$, giving the stated approximation in
$\AxCQ{2(k_H+8)+4}=\AxCQ{2k_H+20}$. If $\norm{v}<\gamma\sqrt N$, then
$\ket{\psi_H}=\ket{0^n}$ and a trivial Clifford circuit prepares it
exactly. Since this is an existence statement for each fixed $H$, no
circuit for testing $\norm{v}$ is required.

Let $S_{\mr{samp}}$ and $A_{\mr{samp}}$ denote the gate and ancilla counts of one instance of the finite Poisson sampler, and let $S_{\mr{load}}$ and $A_{\mr{load}}$ denote the additional gate and ancilla counts of the amplitude loader, excluding its amplitude-oracle calls. If the coherent evaluator for $H$ has size $S_H$ and uses $A_H$ ancilla qubits, then the end-to-end construction uses
$2t(S_H+S_{\mr{samp}})+S_{\mr{load}}$ gates and
$t(A_H+A_{\mr{samp}})+A_{\mr{load}}$ ancilla qubits, where
$t=\ceil{\gamma^{-2}\ln(1/\delta_{\mr{load}})}$. By \Cref{lem:poisson-minimum-gaussian-sampler}, $S_{\mr{samp}}$ and
$A_{\mr{samp}}$ are polynomial in $q$ and $K_{\max}$, while
\Cref{lem:amplitude-loading-in-mh} gives polynomial bounds on
$S_{\mr{load}}$ and $A_{\mr{load}}$ in $n$, $q$, and $t$. Thus, if
$S_H$ and $A_H$ are polynomially bounded, $q$ and $K_{\max}$ are
polynomially bounded, and $\gamma^{-1}$ and
$\log(1/\delta_{\mr{load}})$ are polynomially bounded, then the entire
preparation circuit has polynomial size and uses polynomially many
ancilla qubits.
\end{proof}

We now specialize \Cref{cor:end-to-end-state-preparation} to the pseudorandom tape functions used in \Cref{thm:PRAL-security} to establish implementability of the \PRAL family.

\begin{corollary}[Low-depth implementation of the \PRAL family]
\label{cor:secure-ensemble-implementation}
Suppose the subexponentially secure PRF family $\mathcal G$ used in
\Cref{thm:PRAL-security} is computable by polynomial-size depth-$d_F$
$\TCZ$ circuits. Then the resulting \PRAL family is efficiently
$e^{-n}$-approximately $\MH_{8d_F+20}$-preparable.
\end{corollary}

\begin{proof}
For a fixed key $k$, recall that the tape function used in the
construction is
$F_k(x)=\operatorname{prefix}_R
\pbra{G_{m,k}(0^{m-n}x)}$.
Fixing the key, padding the input by zeros, and retaining only the first
$R$ output bits do not increase the $\TCZ$ depth. By the nonclean
simulation of depth-$d_F$ $\TCZ$ circuits in $\AxCQ{4d_F}$
from~\cite[Lemma~6.3]{parham2025mh}, each $F_k$ therefore has a
polynomial-size coherent evaluator in $\AxCQ{4d_F}$. For the sampler parameters of \Cref{thm:PRAL-security},
$\mu,K_{\max}=O(n+\lambda)$,
$q=O(n+\lambda+T)$, and
$\gamma=1/(2\sqrt{\mu})$. Set
$\delta_{\mr{load}}=e^{-n}$. Then
$t=\ceil{\gamma^{-2}\ln(1/\delta_{\mr{load}})}
=\ceil{4\mu n}=O(n(n+\lambda))$, so all resource parameters in
\Cref{cor:end-to-end-state-preparation} remain polynomially bounded. Applying \Cref{cor:end-to-end-state-preparation} with
$H=F_k$ and $k_H=4d_F$, every state in the support of the secure
ensemble has a polynomial-size preparation circuit using polynomially
many ancilla qubits in
$\AxCQ{2(4d_F)+20}=\AxCQ{8d_F+20}$, whose reduced output is within
trace distance $e^{-n}$ of the corresponding target state. Hence the
ensemble is efficiently $e^{-n}$-approximately
$\MH_{8d_F+20}$-preparable.
\end{proof}

Combining \Cref{cor:secure-ensemble-implementation} with
\Cref{thm:quantum-natural-proofs-barrier}, the quantum natural-proofs
barrier holds against polynomial-size $\MH_{8d_F+20}$ circuits.

\begin{theorem}[Main theorem: natural proofs barrier against MH]
\label{thm:main:natpfsMH:body}
Fix a constant $d_F\geq1$. Suppose there exists a classically subexponentially secure PRF family $\mathcal G$ in the sense of \Cref{def:subexp-prf}, whose functions are computable by polynomial-size depth-$d_F$ threshold circuits.

Then there exists a \PRAL family that is efficiently $e^{-n}$-approximately $\MH_{8d_F+20}$-preparable. Consequently, no natural promise property is $e^{-n}$-useful against $(\MH_{8d_F+20})_{\poly}$.
\end{theorem}

\begin{proof}
Use the finite Poisson sampler of \Cref{lem:poisson-minimum-gaussian-sampler}, with the security parameters verified immediately afterward. This sampler satisfies the assumptions of \Cref{thm:PRAL-security}, so the resulting target-state distributions form a \PRAL family.

For these same target distributions, \Cref{cor:secure-ensemble-implementation} gives efficient $e^{-n}$-approximate preparation in $\MH_{8d_F+20}$. The claimed obstruction then follows from \Cref{thm:quantum-natural-proofs-barrier}.
\end{proof}
 \section{Examples and nonexamples of quantum natural proofs}

In this section, we present several examples of natural properties satisfying Definition~\ref{def:natural-promise-property}.

\subsection{Pauli expectations}\label{ssec:example-pauli-expectations}
We begin with a simple example based on Pauli expectations, expanding upon the example from the introduction. A pure stabilizer state has many nonidentity Pauli expectations with absolute value one. On the other hand, for any fixed nonidentity Pauli, a Haar-random state has almost equal weight on its $+1$ and $-1$ eigenspaces, so the corresponding expectation is overwhelmingly likely to be small. This suggests using the largest nonidentity Pauli expectation to distinguish stabilizer states from typical states.

Let $\calP_n=\{I,X,Y,Z\}^{\otimes n}$ be the phase-free $n$-qubit Pauli operators, and define
\begin{equation}
    M(\psi)=\max_{P\in\calP_n\setminus\{I\}}\abs{\abra{\psi|P|\psi}}^2.
\end{equation}
Define the associated promise property by
\begin{align}
    P_n^{\mathrm{yes}}=\{\ket{\psi}:M(\psi)\leq 1/4\}, \qquad P_n^{\mathrm{no}}=\{\ket{\psi}:M(\psi)\geq 3/4\}
\end{align}
These sets are disjoint and invariant under global phase. We use squared expectations here because their values can be computed exactly from a rational \amplistwords. The constant gap makes the property robust to amplitude truncation and approximate state preparation.

\begin{claim}[Pauli-expectation property]
Let $\eps:\NN\to[0,1/16]$. The above promise property $(P_n^{\mathrm{yes}},P_n^{\mathrm{no}})$ is natural and is $\eps$-useful against Clifford state preparation, even when ancillary output qubits may be discarded.
\end{claim}

\begin{proof}
\textbf{Constructivity.}
To satisfy constructivity, we give a deterministic tester that decides the promise property from a sufficiently precise \amplistwords. Let $\tau(n)\coloneq n+6$ be the minimum number of fractional bits of precision required by the tester. 
Given $1^n$ and a truncated \amplistwords $\al_T(\psi)$ with $T\geq\tau(n)$, let $v\coloneq\al_T(\psi)$ denote the truncated amplitude vector. From $v$, compute $\widetilde M(v)\coloneq\max_{P\in\calP_n\setminus\{I\}}\abs{\bra{v}P\ket{v}}^2$, where the tilde distinguishes this truncated-input statistic from the true value $M(\psi)$. The tester accepts when $\widetilde M(v)\leq1/2$ and rejects otherwise.

We now bound the error between $\widetilde M(v)$ and $M(\psi)$. By Definition~\ref{def:state-amp-list}, truncating each amplitude to $T$ fractional bits gives $\norm{v-\ket{\psi}}_2\leq2^{(n+1)/2-T}$. Let $\delta_T\coloneq2^{(n+1)/2-T}$ denote this upper bound on the \amplistwords truncation error. Since $T\geq\tau(n)=n+6$, we have $\delta_T\leq1/64$. The triangle inequality then gives $\norm{v}_2\leq\norm{\psi}_2+\norm{v-\psi}_2\leq1+\delta_T$. To compare $\widetilde M(v)$ with $M(\psi)$, we first bound the change in the squared expectation of a fixed Pauli. Fix $P\in\calP_n\setminus\{I\}$. Adding and subtracting $\bra{v}P\ket{\psi}$ gives
$\bra{\psi}P\ket{\psi}-\bra{v}P\ket{v}
=
\bra{\psi-v}P\ket{\psi}
+
\bra{v}P\ket{\psi-v}$.
Applying the triangle inequality and Cauchy--Schwarz, and using that $P$ is unitary, gives
$\abs{\bra{\psi}P\ket{\psi}-\bra{v}P\ket{v}}
\leq
\norm{\psi-v}_2\norm{\psi}_2
+
\norm{v}_2\norm{\psi-v}_2$.
Since $\norm{\psi}_2=1$, $\norm{\psi-v}_2\leq\delta_T$, and $\norm{v}_2\leq1+\delta_T$, the difference is at most $\delta_T(2+\delta_T)<3\delta_T$. Since $\abs{\bra{\psi}P\ket{\psi}}\leq1$, the bound above gives
$\abs{\bra{v}P\ket{v}}\leq1+3\delta_T$. Therefore,
\begin{equation}
\abs{\abs{\bra{\psi}P\ket{\psi}}^2-\abs{\bra{v}P\ket{v}}^2}
\leq
\left(
\abs{\bra{\psi}P\ket{\psi}}
+
\abs{\bra{v}P\ket{v}}
\right)
\abs{
\abs{\bra{\psi}P\ket{\psi}}
-
\abs{\bra{v}P\ket{v}}
}
\leq
(2+3\delta_T)3\delta_T
<
\frac18.
\end{equation}
This bound holds for every nonidentity Pauli, the maximum remains bounded by $<1/8$. If $\ket{\psi}\in P_n^{\mathrm{yes}}$, then $M(\psi)\leq1/4$, so $\widetilde M(v)<3/8$. If $\ket{\psi}\in P_n^{\mathrm{no}}$, then $M(\psi)\geq3/4$, so $\widetilde M(v)>5/8$. The threshold $1/2$ therefore correctly separates the two promise sets.

The tester also runs in time polynomial in the length of the \amplistwords. The entries of $v$ are rational complex numbers, so each value $\abs{\bra{v}P\ket{v}}^2$ can be computed exactly. Each $n$-qubit Pauli is a $2^n\times2^n$ matrix with entries in $\{0,\pm1,\pm i\}$. More specifically, each Pauli permutes the computational-basis states up to phases in $\{\pm1,\pm i\}$, so computing one expectation $\bra{v}P\ket{v}$ requires $O(2^n)$ arithmetic operations. There are $4^n-1$ nonidentity Paulis, so computing $\widetilde M(v)$ requires $O(8^n)$ arithmetic operations, with polynomial overhead in $T+n$ for exact rational arithmetic. Since the \amplistwords has length $\Theta(2^nT)$, this running time is polynomial in the total input length.

\medskip
\noindent\textbf{Largeness.}
Fix a nonidentity Pauli $P$. Its $+1$ and $-1$ eigenspaces have equal dimension, so $\bra{\psi}P\ket{\psi}$ measures the imbalance of the state's weight between two equal-dimensional halves of the Hilbert space. For a Haar-random state, this imbalance is sharply concentrated around zero. In particular, Example~54 of \cite{mele2024haar} shows that for every $\varepsilon>0$,
$\Pr_{\psi\sim\Haar_n}[\abs{\bra{\psi}P\ket{\psi}}\geq\varepsilon]\leq2\exp(-2^n\varepsilon^2/(18\pi^3))$. Taking $\varepsilon=1/2$, and using that $\abs{\bra{\psi}P\ket{\psi}}^2\geq1/4$ if and only if $\abs{\bra{\psi}P\ket{\psi}}\geq1/2$, we obtain
$\Pr_{\psi\sim\Haar_n}[\abs{\bra{\psi}P\ket{\psi}}^2\geq1/4]\leq2\exp(-2^n/(72\pi^3))\leq\exp(-\Omega(2^n))$. There are $4^n-1$ nonidentity Paulis, so taking a union bound gives
$\Pr_{\psi\sim\Haar_n}[M(\psi)\geq1/4]\leq2(4^n-1)\exp(-2^n/(72\pi^3))\leq\exp(-\Omega(2^n))$.
Hence a Haar-random state lies in $P_n^{\mathrm{yes}}$ with overwhelming probability.

\medskip
\noindent\textbf{Usefulness.}
Fix $n$. Suppose a Clifford circuit acting on an all-zero input has reduced output $\sigma$ on the designated $n$-qubit output register, after discarding any ancillary output qubits, and let $\ket{\psi}$ be a pure target state satisfying
$D(\ket{\psi}\!\bra{\psi},\sigma)=\frac12\norm{\ket{\psi}\!\bra{\psi}-\sigma}_1\leq\eps(n)$,
where $\eps(n)\leq1/16$. We show that this forces $\ket{\psi}$ to lie in $P_n^{\mathrm{no}}$.

The joint output state before discarding the ancillary qubits is a pure stabilizer state. A reduced state of a pure stabilizer state is a normalized projector onto a stabilizer subspace. Thus $\sigma=\Pi/r$, where $\Pi$ is a projector onto a stabilizer subspace of rank $r$. We first show that the approximation guarantee forces $r=1$. Consider the two-outcome measurement with projectors $\ket{\psi}\!\bra{\psi}$ and $I-\ket{\psi}\!\bra{\psi}$. It accepts $\ket{\psi}$ with probability $1$ and $\sigma$ with probability $\bra{\psi}\sigma\ket{\psi}$. Since a measurement cannot increase trace distance, $D(\ket{\psi}\!\bra{\psi},\sigma)\geq1-\bra{\psi}\sigma\ket{\psi}$. Since $\sigma=\Pi/r$ and $\Pi$ is a projector,
$\bra{\psi}\sigma\ket{\psi}
=
\frac{1}{r}\bra{\psi}\Pi\ket{\psi}
\leq\frac1r$.
Therefore
$D(\ket{\psi}\!\bra{\psi},\sigma)\geq1-\frac1r$. If $r\geq2$, then $1-1/r\geq1/2$, so
$D(\ket{\psi}\!\bra{\psi},\sigma)\geq1/2$, contradicting $\eps(n)\leq1/16$. Therefore $r=1$, and $\sigma$ is itself a pure stabilizer state.

A pure stabilizer state is stabilized by a nonidentity Pauli, so there exists $P\in\calP_n\setminus\{I\}$ such that $\abs{\operatorname{Tr}(P\sigma)}=1$. Since $\norm{P}_{\mathrm{op}}=1$, trace-norm duality gives
$\frac{1}{2}\abs{\operatorname{Tr}(P\sigma)-\bra{\psi}P\ket{\psi}}
\leq D(\sigma,\ket{\psi}\!\bra{\psi})
\leq\eps(n)$.
Hence $\abs{\bra{\psi}P\ket{\psi}}\geq1-2\eps(n)\geq7/8$, and therefore
$M(\psi)\geq(7/8)^2=49/64>3/4$. Thus $\ket{\psi}\in P_n^{\mathrm{no}}$. Since $n$ was arbitrary, this holds for every positive $n$, and in particular proves the required eventual membership. Hence the property is $\eps$-useful against Clifford state preparation.
\end{proof}

\subsection{Approximate Pauli degree: CAT states and felinity}
\label{ssec:example-approximate-pauli-degree}

Anshu, Dong, Ou, and
Yao~\cite{anshu2025computational} show that states prepared by
constant-depth $\mathsf{QAC}^0$ circuits with limited ancilla admit low-degree approximations in operator norm. For any state that does not admit such an approximation this implies a circuit lower bound. 

We formalize this as a ``high approximate degree'' property, and explain why CAT states and, more generally,  pure states with high \emph{felinity} \cite{gretta2026paritynotinqac0iff} satisfy this property.
We then show that the
certificate is natural. Throughout, the approximation
is in operator norm; this is not a concentration statement about
squared Pauli coefficients.

\paragraph{Distance from low-degree operators}
Let $\mathcal V_k$ be the subspace of all $n$-qubit operators that are a sum of Hermitian operators that act nontrivially on at most $k$ qubits. Equivalently, we could define it as the span of all Paulis of weight at most $k$.

For a pure state $|\psi\rangle$, write
$\psi=|\psi\rangle\langle\psi|$ and define $\Delta_k(\psi)
  =\inf_{A\in\mathcal V_k}\|\psi-A\|_\mr{op}$,
where $\|\cdot\|_\mr{op}$ denotes operator norm. Thus a large value of
$\Delta_k(\psi)$ means that no low-degree operator approximates the
state's density matrix well. Note that the approximating operator $A$ need not be positive or have unit trace.

\paragraph{CAT states.}
Let
$\ket{\mathrm{CAT}_n}
=(\ket{0^n}+\ket{1^n})/\sqrt2$
and
$\rho_{\mathrm{CAT}_n}
=\ket{\mathrm{CAT}_n}\!\bra{\mathrm{CAT}_n}$.
Its projector has coherence between two strings that differ on every
qubit: $\langle0^n|\rho_{\mathrm{CAT}_n}|1^n\rangle=\frac12$.
No Pauli string of weight less than $n$ can contribute to this matrix
element. Such a string has an identity factor on some qubit, giving
a factor $\langle0|I|1\rangle=0$. Consequently, $\langle0^n|A|1^n\rangle=0$ for every $A\in\mathcal V_k, k<n$.

The link to approximate degree is that operator norm bounds every
matrix element between normalized vectors. Hence every such $A$
satisfies $\|\rho_{\mathrm{CAT}_n}-A\|_\mr{op}
  \geq
  \bigl|\langle0^n|(\rho_{\mathrm{CAT}_n}-A)|1^n\rangle\bigr|
  =\frac12$.
The degree-zero operator $I/2$ attains this error, so $\Delta_k(\mathrm{CAT}_n)=\frac12$ for every $0\leq k<n$.
Thus approximation error strictly below $1/2$ requires Pauli degree $n$.

\paragraph{Felinity.}
The CAT-state argument, which uses coherence between the specific complementary pair $\ket{0^n}$ and $\ket{1^n}$, extends more generally to arbitrary complementary computational-basis pairs through the notion of felinity.
For an $n$-qubit state $\rho$, write
$p_x=\langle x|\rho|x\rangle$ and let $\bar x$ be the bitwise
complement of $x$. The \emph{felinity} introduced by Gretta, Gupta,
and Joshi~\cite{gretta2026paritynotinqac0iff} is
\[
  \mathcal F_n(\rho)
  =2\sum_{x\in\{0,1\}^n}p_xp_{\bar x}.
\]
For a pure state, a large felinity guarantees a large off-diagonal
matrix element between some complementary pair $(x, \bar{x})$. Just as the CAT state,
no operator of degree less than $n$ can reproduce that matrix element.

\begin{lemma}[Felinity implies high approximate Pauli degree]
\label{lem:felinity-approximate-degree}
Let $\psi=|\psi\rangle\langle\psi|$ be an $n$-qubit pure state.
For every $0\leq k<n$,
\[
  \Delta_k(\psi)
  \geq\max_x\sqrt{p_xp_{\bar x}}
  \geq\frac{\mathcal F_n(\psi)}2.
\]
\end{lemma}

\begin{proof}
For every $A\in\mathcal V_k$ and every $x$, the argument above gives
$\langle x|A|\bar x\rangle=0$. Since $\psi$ is pure, $\|\psi-A\|_\mr{op}
  \geq|\langle x|\psi|\bar x\rangle|
  =\sqrt{p_xp_{\bar x}}$.
Taking the maximum over $x$ and the infimum over $A$ proves the first
inequality. For the second inequality, set $q_x=\sqrt{p_xp_{\bar x}}$. Cauchy--Schwarz gives
$\sum_xq_x\leq1$, so
\[
  \frac{\mathcal F_n(\psi)}2
  =\sum_xq_x^2
  \leq\left(\max_xq_x\right)\sum_xq_x
  \leq\max_xq_x.
\]
\end{proof}

\paragraph{A natural property from approximate Pauli degree.} We now describe our natural property that applies to these states. 
Let $\ell_n=\lfloor n/10\rfloor$ and define
\begin{equation}
  P_n^{\mathrm{yes}}
  =\{|\psi\rangle:\Delta_{\ell_n}(\psi)\geq1/4\},
  \qquad
  P_n^{\mathrm{no}}
  =\{|\psi\rangle:\Delta_{\ell_n}(\psi)\leq1/8\}.
  \label{eq:example-approximate-degree-promise}
\end{equation}
CAT states belong to the YES-set, as do pure states
with felinity at least $1/2$. 

\medskip
\noindent\textbf{Constructivity.}
To decide the promise property from an \amplistwords, we need to efficiently estimate $\Delta_{\ell_n}(\psi)$, the distance from the input state to the space of low-degree operators. Since this distance is defined by minimizing an operator-norm error over $\mathcal V_{\ell_n}$, the minimization can be formulated as a semidefinite program.

Given an \amplistwords with rational amplitude vector $v$, if $v=0$, reject. Otherwise, form
$\widehat\rho=vv^\dagger/\langle v,v\rangle$. The density-matrix representation keeps all input entries
rational. Compute the optimum of
\begin{equation}
\begin{aligned}
  \text{minimize}\quad &u\\
  \text{subject to}\quad
  &-uI\preceq\widehat\rho-
       \sum_{|P|\leq\ell_n}c_PP\preceq uI,\\
  &-2\leq c_P\leq2,\qquad 0\leq u\leq2,
\end{aligned}
\label{eq:example-approximate-degree-sdp}
\end{equation}
where all coefficients $c_P$ are real. The  constraints are included only to make the feasible region
explicitly bounded; they do not change the optimum. Indeed,
$I/2\in\mathcal V_{\ell_n}$ and $\left\|\widehat\rho-\frac I2\right\|_{\mathrm{op}}=\frac12$, 
so the optimum is at most $1/2$. Consequently, an optimal approximant $A$ may be chosen with
\[
    \|A\|_{\mathrm{op}}
    \leq
    \|\widehat\rho-A\|_{\mathrm{op}}
    +\|\widehat\rho\|_{\mathrm{op}}
    \leq\frac32.
\]
Its Pauli coefficients therefore satisfy
\[
    |c_P|
    =\frac{|\operatorname{Tr}(PA)|}{N}
    \leq\|A\|_{\mathrm{op}}
    \leq\frac32.
\]
Thus the looser bounds $|c_P|\leq2$ and $u\leq2$ preserve the optimum.

The program has at most $N^2$ variables and two $N\times N$
positive-semidefinite constraints, so its description size is polynomial
in $N$ and in the input bit length. The constraints make the feasible
region bounded. Since $0\preceq\widehat\rho\preceq I$, the choice
$c_P=0$ and $u=3/2$ is strictly feasible, with spectral slack at
least $1/2$.
Standard semidefinite optimization therefore approximates the optimum
to additive error $1/64$ in time polynomial in $N$ and the input bit
length~\cite{vandenberghe1996semidefinite}.

We next bound the error introduced by truncating the \amplistwords. If each real and
imaginary part is truncated to $T$ fractional bits, let
$\eta=\sqrt{2N}\,2^{-T}$, so
$\norm{v-\ket{\psi}}_2\leq\eta$.
The reverse triangle inequality gives
$\abs{\norm{v}_2-1}\leq\eta$.
For $\eta<1$, define the normalized truncated state
$\widehat{\ket{\psi}}\coloneq v/\norm{v}_2$.
Then
$\norm{\widehat{\ket{\psi}}-\ket{\psi}}_2
\leq
\abs{1-\norm{v}_2}
+\norm{v-\ket{\psi}}_2
\leq2\eta$. Since
$\widehat{\rho}=\widehat{\ket{\psi}}\widehat{\bra{\psi}}$,
\[
\|\widehat\rho-\psi\|_{\mr{op}}
\leq
\|\widehat\rho-\psi\|_1
\leq
2\norm{\widehat{\ket{\psi}}-\ket{\psi}}_2
\leq4\eta.
\]
The triangle inequality further implies
\[
\left|
\inf_{A\in\mathcal V_{\ell_n}}\|\widehat{\rho}-A\|_{\mathrm{op}}
-
\inf_{A\in\mathcal V_{\ell_n}}\|\psi-A\|_{\mathrm{op}}
\right|
\leq
\|\widehat{\rho}-\psi\|_{\mr{op}}.
\]
Therefore, replacing $\psi$ by $\widehat\rho$ changes the
optimal approximation error by at most $4\eta$.

Thus $T\geq n+10$, together with the numerical error $1/64$, gives
total error less than $1/32$ in our estimation of $\Delta_{\ell_n}(\psi)$.
Accepting when the computed value is at least $3/16$ separates the
two promise sets. The running time is polynomial in the full
\amplistwords length for every such $T$.

\medskip
\noindent\textbf{Largeness.}
For a Haar-random state $\ket{\psi}$, let
\[
L(\psi)=\sum_{1\leq|P|\leq\ell_n}
|\langle\psi|P|\psi\rangle|^2.
\]
For every nonidentity Pauli $P$,
$\mathbb E|\langle\psi|P|\psi\rangle|^2=1/(N+1)$.
Let
$M\coloneq\sum_{j\leq n/10}\binom nj3^j$.
Since
$h_2(1/10)+(1/10)\log_2 3<0.628$,
we have
$M\leq2^{0.628n}$, so there are at most $2^{0.628n}$ nonidentity Paulis of weight at most $\ell_n$. Therefore,
\begin{align}
    \mathbb E \sbra{L(\psi)}=\frac{M}{N+1}=2^{-\Omega(n)}. \label{eq:degree-example:haar-concentration}
\end{align}
Now we will show that if $\psi$ were approximately low degree: $\Delta_{\ell_n}(\psi) < 1/4$ then it must also have some concentration of low degree Pauli mass: $L(\psi) \geq 1/10$.

Suppose that $\Delta_{\ell_n}(\psi)<1/4$, so some
$A=\sum_Pc_PP\in\mathcal V_{\ell_n}$ satisfies
$\|\psi-A\|_{\mathrm{op}}<1/4$.
Then
\begin{align*}
    1 - \bra{ \psi} A \ket{\psi} \leq \norm{\psi - A }_\mr{op} < 1/4.
\end{align*}
so $\langle\psi|A|\psi\rangle\geq3/4$. Also $\frac{1}{N} \Tr( A - \psi) \leq \norm{A - \psi}_\mr{op} \leq \frac{1}{4}$, so $\Tr(A)/N \leq 1/4 + 1/N$.
And by triangle inequality
$\|A\|_{\mathrm{op}}\leq5/4$.
Hence, by Pauli orthogonality and Cauchy--Schwarz,
\begin{align*}
    \frac{3}{4} \leq \bra{\psi} A \ket{\psi} = c_I + \sum_{1\leq|P|\leq\ell_n} c_P \bra{\psi} P \ket{\psi} 
    &\leq \frac{\Tr(A)}{N} + \pbra{\sum_{1\leq|P|\leq\ell_n} |c_P|^2}^{1/2} \cdot \pbra{\sum_{1\leq|P|\leq\ell_n} |\bra{\psi} P \ket{\psi}|^2}^{1/2} \\
    &\leq \frac{1}{4} + \frac{1}{N} + \frac{5}{4} \cdot \sqrt{L(\psi)}.
\end{align*}
In the last inequality we used Pauli orthogonality, which gives
$\sum_P |c_P|^2=\operatorname{Tr}(A^2)/N\leq\norm{A}_{\mathrm{op}}^2\leq(5/4)^2$.
Therefore, $L(\psi) \geq \pbra{\frac{2N-4}{5N}}^2$. 
For $N\geq16$, this implies $L(\psi)\geq1/10$. Therefore, by
Markov's inequality and \Cref{eq:degree-example:haar-concentration},
\[
\Pr_{\psi\sim\mathrm{Haar}}
[\Delta_{\ell_n}(\psi)<1/4]
\leq
\Pr[L(\psi)\geq1/10]
\leq
10\,\mathbb E[ L(\psi)]
=
2^{-\Omega(n)}.
\]
Thus the YES-set has Haar measure $1-2^{-\Omega(n)}$.

\medskip
\noindent\textbf{Usefulness}  against $\QAC^0$ with limited ancilla.
We apply the state-synthesis theorem of Anshu, Dong, Ou, and
Yao~\cite[Theorem~6.4 and
Corollary~6.5]{anshu2025computational}.
In the regime of fixed depth and $O(n)$ ancillary qubits, the cited
result gives an $o(n)$-degree Pauli approximation with operator-norm
error $o(1)$ for any pure target prepared with fidelity at least
$1-e^{-n}$, even when auxiliary output qubits are discarded.

Let $\psi=\ket{\psi}\!\bra{\psi}$, and let $\sigma$ be the
circuit's reduced output. If $\frac12\|\sigma-\psi\|_1\leq e^{-n}$, then 
$\langle\psi|\sigma|\psi\rangle
    =\operatorname{Tr}(\psi\sigma)
    \geq 1-e^{-n}$.
Thus the fidelity hypothesis of the cited approximation theorem is
satisfied. The theorem therefore gives an operator
$A_n\in\mathcal V_{d_n}$ such that $d_n=o(n)$ and $\|\psi-A_n\|_{\mathrm{op}}=o(1)$. Since $\ell_n=\lfloor n/10\rfloor$, for all sufficiently large $n$ we
have $d_n\leq\ell_n$ and
$\|\psi-A_n\|_{\mathrm{op}}\leq1/8$. Therefore $\Delta_{\ell_n}(\psi)
    \leq\|\psi-A_n\|_{\mathrm{op}}
    \leq\frac18$
and hence $\ket{\psi}\in P_n^{\mathrm{no}}$.

\paragraph{Constant felinity.}
The thresholds $1/4$ and $1/8$ here were chosen only for convenience. More generally, for any fixed $f_0>0$, the same largeness, constructivity, and usefulness arguments apply after replacing these thresholds by $f_0/4$ and $f_0/8$. By Lemma~\ref{lem:felinity-approximate-degree}, every pure state with $\mathcal F_n(\psi)\geq f_0$ satisfies $\Delta_{\ell_n}(\psi)\geq f_0/2$ and therefore lies in the new YES-set. 

Thus, although high felinity does not itself satisfy the largeness condition, it implies the broader property of high approximate Pauli degree, which is natural.

\subsection{Unique ground states of local Hamiltonians}\label{ssec:example-unique-ground-state}

Another example of a natural proof comes from local Hamiltonians.
A shallow circuit state is automatically the unique ground state of
a local Hamiltonian: if
$\ket{\psi}=U\ket{0^n}$ and $U$ is a depth-$d$ circuit of two-qubit
gates, then
\begin{equation}
    H_U
    =
    \frac1n\sum_{i=1}^n
    U\bigl(I-\ket0\!\bra0_i\bigr)U^\dagger
\end{equation}
has $\ket{\psi}$ as its unique ground state.  Each term acts on at
most $2^d$ qubits, and $H_U$ has operator norm at most one and
spectral gap $1/n$. Thus, a robust obstruction to being the unique ground state of such a local, gapped Hamiltonian gives a shallow-circuit lower bound.

Let $\mathcal P_{\leq k}$ denote the nonidentity Pauli strings of
weight at most $k$, and let
$m_k=|\mathcal P_{\leq k}|=n^{O(k)}$ for constant $k$.
For $\rho_\psi=\ket\psi\!\bra\psi$, define
\begin{equation}
    d_k(\psi)
 =
 \min_{\substack{\sigma\succeq0,\ \Tr\sigma=1\\
                  \Tr(\rho_\psi\sigma)\leq1/4}}
 \max_{P\in\mathcal P_{\leq k}}
 \left|\Tr[P(\sigma-\rho_\psi)]\right|.
\end{equation}
Thus $d_k(\psi)$ measures how well a state that is globally far from
$\ket\psi$ can reproduce all of its $k$-local Pauli expectation
values. In particular, $d_k(\psi)=0$ whenever there is a state with
overlap at most $1/4$ with $\ket\psi$ and exactly the same $k$-local
expectation values.

For fixed $k$, consider the promise property
\begin{equation}
    P_n^{\mathrm{yes}}
 =
 \{\ket\psi:d_k(\psi)\leq2^{-n/5}\},
 \qquad
 P_n^{\mathrm{no}}
 =
 \left\{\ket\psi:
 d_k(\psi)\geq\frac{1}{2nm_k}\right\}.
\end{equation}
For sufficiently large $n$ these sets are disjoint.

\paragraph{Constructivity.}
The quantity $d_k(\psi)$ is the optimum of a semidefinite program.
Given $\rho_\psi$, introduce variables $\sigma$ and $t$ and minimize
$t$ subject to
\begin{equation}
    \sigma\succeq0,\qquad
    \Tr\sigma=1,\qquad
    \Tr(\rho_\psi\sigma)\leq\frac14,\qquad
    -t\leq\Tr[P(\sigma-\rho_\psi)]\leq t
\end{equation}
for every $P\in\mathcal P_{\leq k}$. For constant $k$ there are only $n^{O(k)}$ Pauli constraints and the
matrix dimension is $2^n$, so the SDP can be solved in time polynomial
in the length of the complete amplitude list.  The inverse-polynomial
gap between the YES and NO thresholds also leaves ample room for the
finite-precision error in the input and the numerical optimization.
Hence the property is constructive.

\paragraph{Largeness.}
For a Haar-random state, take
\[
    \sigma=\frac{I-\rho_\psi}{2^n-1}.
\]
This state is orthogonal to $\ket\psi$, and for every nonidentity
Pauli $P$,
\[
 \left|\Tr[P(\sigma-\rho_\psi)]\right|
 =
 \frac{2^n}{2^n-1}
 |\langle\psi|P|\psi\rangle|.
\]
For Haar-random $\ket\psi$,
$\mathbb E|\langle\psi|P|\psi\rangle|^2=1/(2^n+1)$.
Since there are only $m_k=\poly(n)$ Paulis of weight at most $k$,
a union bound and Markov's inequality give
\[
 \Pr_{\psi\sim\mathrm{Haar}}
 \left[
   \max_{P\in\mathcal P_{\leq k}}
   |\langle\psi|P|\psi\rangle|
   >2^{-n/5}
 \right]
 =2^{-\Omega(n)}.
\]
Thus $d_k(\psi)\leq2^{-n/5}$ with probability
$1-2^{-\Omega(n)}$, so the YES-set is Haar-large.

\paragraph{Usefulness.}
Suppose $\ket\psi$ is the unique ground state of a $k$-local
Hamiltonian $H$ with $\|H\|_{\rm op}\leq1$ and spectral gap
$\gamma$.  For every state $\sigma$ satisfying
$\Tr(\rho_\psi\sigma)\leq1/4$, $\Tr(H\sigma)-\langle\psi|H|\psi\rangle
 \geq\frac34\gamma$.
Writing $H=\sum_Pc_PP$, we have $|c_P|\leq1$, and therefore
\[
 \frac34\gamma
 \leq
 \sum_{P\in\mathcal P_{\leq k}}
 |c_P|\,
 |\Tr[P(\sigma-\rho_\psi)]|
 \leq
 m_k d_k(\psi).
\]
Hence $d_k(\psi)\geq\frac{3\gamma}{4m_k}$.

Now let $\ket\psi=U\ket{0^n}$ for a depth-$d$ circuit of two-qubit
gates and set $k=2^d$.  The Hamiltonian $H_U$ above is $k$-local,
has norm at most one, and has gap $1/n$.  Consequently
\[
    d_k(\psi)\geq\frac{3}{4nm_k}
    >\frac{1}{2nm_k},
\]
so every such shallow-circuit state lies in $P_n^{\mathrm{no}}$.
The promise property is therefore useful against depth-$d$
two-qubit state preparation.

\paragraph{Example: CAT states.}
For
\[
 \ket{\mathrm{CAT}_\pm}
 =
 \frac{\ket{0^n}\pm\ket{1^n}}{\sqrt2},
\]
the two orthogonal states have identical reduced density matrices on
every proper subset of the qubits.
Thus, taking $\sigma=\ket{\mathrm{CAT}_-}\!\bra{\mathrm{CAT}_-}$ in the definition of $d_k$, we obtain $d_k(\mathrm{CAT}_+)=0$ for every $k<n$.
Hence, $\ket{\mathrm{CAT}_+}$ lies in the YES-set for every fixed $k<n$, recovering the familiar constant-depth lower bound for CAT-state preparation through a natural property.

\paragraph{Approximate preparation with discarded ancillas.}
The usefulness statement extends to the reduced-output model used
elsewhere in this paper. Suppose a depth-$d$ circuit on $n+a$ qubits
has reduced output $\omega$ satisfying
$D(\omega,\rho_\psi)\leq\epsilon$. Its global pure output
$\ket{\Phi}$ is the unique ground state of a
$2^d$-local Hamiltonian of norm at most one and gap $1/(n+a)$.
By Uhlmann's theorem, there is an ancilla state $\ket{\eta}$ such that
$\ket{\Phi}$ is within $O(\sqrt{\epsilon})$ of
$\ket{\psi}\ket{\eta}$. 
Fixing the ancillas to \(\ket{\eta}\), define the effective output Hamiltonian \(H_{\mathrm{eff}}=(I\otimes\langle\eta|)H_U(I\otimes|\eta\rangle)\). Each local term of \(H_U\) remains supported only on the output qubits in its original support, so \(H_{\mathrm{eff}}\) is still \(2^d\)-local. The preceding argument then gives
\[
 d_{2^d}(\psi)
 \geq
 \frac{1}{m_{2^d}}
 \left(
   \Omega\!\left(\frac1{n+a}\right)
   -O(\sqrt{\epsilon})
 \right).
\]
Hence, for polynomially many ancillas and
$\epsilon=e^{-\Omega(n)}$, we have
$d_{2^d}(\psi)\geq1/\poly(n)$. Choosing, for example, YES and NO
thresholds $2^{-n/5}$ and $2^{-n/6}$ therefore makes the property
useful even when arbitrary polynomially many ancillary outputs may
be discarded.

\subsection{Mutual information lower bounds}

Our next example is based on the mutual-information lower bound of Parham~\cite{parham2025mh} against state preparation in the first level $\AxCQ{1}$ of the magic hierarchy, which consists of an arbitrary-size Clifford circuit followed by a $\QNC^0$ circuit. The lower bound applies to states for which the mutual information between every pair of qubits is bounded below, while the mutual information between sufficiently small disjoint subsets of qubits is bounded above.

For disjoint subsets $A,B\subseteq[n]$ of qubit indices, recall that their mutual information is defined by
$I_\psi(A:B)=S(\rho_A)+S(\rho_B)-S(\rho_{AB})$,
where $\rho_A$, $\rho_B$, and $\rho_{AB}$ are the corresponding reduced states and
$S(\rho)=-\operatorname{Tr}(\rho\log_2\rho)$
is the von Neumann entropy in bits.

Theorem~3.6 of~\cite{parham2025mh} considers an $n$-qubit state $\ket{\psi}$ for which there exist integers $k\geq0$ and $s\geq1$ and parameters $k<\alpha\leq\beta<k+1$ satisfying
\begin{enumerate}
    \item $I_\psi(i:j)\geq\alpha$ for every pair of distinct qubits $i,j$.
    \item $I_\psi(A:B)\leq\beta$ for every pair of disjoint subsets $A,B\subseteq[n]$ with $|A|,|B|<s$.
\end{enumerate}
The theorem gives lower bounds on any approximate $\AxCQ{1}$ preparation of a state satisfying these conditions. We define a promise property by choosing explicit parameters in the case $k=0$ and introducing gaps around the two mutual-information conditions.

\begin{definition}[Mutual information promise property]
\label{def:mutual-info-promise-property}
For each $n$, set
$\alpha_n^{yes}\coloneq 2^{-2n}$,
$\alpha_n^{no}\coloneq 2^{-2n-1}$,
$\beta_n^{yes}\coloneq 1/4$,
$\beta_n^{no}\coloneq 1/2$, and
$s_n\coloneq \max\{2,\lfloor n/4\rfloor\}$.

The mutual information promise property
$P_{\mathrm{MI}}=\{(P_n^{yes},P_n^{no})\}_{n\in\NN}$
is defined by
\begin{align}
P_n^{yes}
\coloneq
\Bigl\{
\ket{\psi}\in S(\CC^{2^n}) :
&\ I_\psi(i:j)\geq \alpha_n^{yes}
\text{ for every } i\neq j,
\nonumber\\
&\ I_\psi(A:B)\leq \beta_n^{yes}
\text{ for all disjoint } A,B\subseteq[n]
\text{ with } |A|,|B|<s_n
\Bigr\},
\label{eq:mut-info-property-yes}
\end{align}
and
\begin{align}
P_n^{no}
\coloneq
\Bigl\{
\ket{\psi}\in S(\CC^{2^n}) :
&\ I_\psi(i:j)\leq \alpha_n^{no}
\text{ for some } i\neq j
\nonumber\\
&\ \text{or }\
I_\psi(A:B)\geq \beta_n^{no}
\text{ for some disjoint } A,B\subseteq[n]
\text{ with } |A|,|B|<s_n
\Bigr\}.
\label{eq:mut-info-property-no}
\end{align}
\end{definition}
Since $\alpha_n^{no}<\alpha_n^{yes}$ and $\beta_n^{yes}<\beta_n^{no}$, the sets $P_n^{yes}$ and $P_n^{no}$ are disjoint. Both sets are invariant under global phase.

\begin{claim}[Mutual information property]\label{clm:mutual-info-main-claim}
The mutual information promise property of Definition~\ref{def:mutual-info-promise-property} is natural in the sense of Definition~\ref{def:natural-promise-property} and is $\varepsilon_{\mathrm{MI}}$-useful against clean-output $\AxCQ{1}$ state preparation, where
$\varepsilon_{\mathrm{MI}}(n)\coloneq2^{-4n-10}$.
\end{claim}
In the following sections, we first prove that the mutual information promise property satisfies largeness and constructivity, and then prove its usefulness against clean-output $\AxCQ{1}$ state preparation.

\paragraph{Largeness.}

Let $N\coloneq2^n$. The definition of $P_n^{yes}$ imposes two conditions, which we analyze separately. Let $E_L$ be the event that $I_\phi(i:j)\geq N^{-2}=2^{-2n}$ for every pair of distinct qubits $i,j$, and let $E_U$ be the event that $I_\phi(A:B)\leq1/4$ for every pair of disjoint subsets $A,B\subseteq[n]$ with $|A|,|B|<s_n$. By Definition~\ref{def:mutual-info-promise-property}, $\ket{\phi}\in P_n^{yes}$ exactly when both $E_L$ and $E_U$ occur.

We first show that a Haar-random state satisfies the pairwise mutual-information condition $E_L$ with high probability.

\begin{claim}[Pairwise mutual information lower bound]
\label{clm:mut-info-lower-bound}
For all sufficiently large $n$,
$\Pr_{\phi\sim\Haar_n}[E_L^c]
=
O(n^22^{-n/4})
=
o(1)$.
\end{claim}

\begin{proof}
We first establish the two facts about Haar-random Pauli expectations used below. Fix $P\in\calP_n\setminus\{I\}$ and let
$X_P\coloneq\bra{\phi}P\ket{\phi}$ for $\ket{\phi}\sim\Haar_n$.

First, we show an anti-concentration bound for Haar-random Pauli expectations. Every nonidentity Pauli has eigenvalues $+1$ and $-1$, each with multiplicity $N/2$. In an eigenbasis of $P$, the squared amplitudes of a Haar-random state are Dirichlet$(1,\ldots,1)$ distributed~\cite{zyczkowski2001induced}, so
$(1+X_P)/2\sim\operatorname{Beta}(N/2,N/2)$.
Consequently, $X_P$ has density
$f_N(x)=\frac{\Gamma(N)}{2^{N-1}\Gamma(N/2)^2}(1-x^2)^{N/2-1}$
for $-1<x<1$. This density is maximized at zero. To bound its maximum, write $N=2m$ and define
$a_m\coloneq\binom{2m}{m}/4^m$.
Then
$f_N(0)=ma_m$.
Moreover,
$a_1\leq1/\sqrt2$ and
$a_{m+1}=a_m(2m+1)/(2m+2)
\leq a_m\sqrt{(m+1)/(m+2)}$,
so induction gives
$a_m\leq1/\sqrt{m+1}<1/\sqrt m$.
Hence
$f_N(0)\leq\sqrt N$, and therefore, for every $0<\delta\leq1$,
$\Pr[\abs{X_P}\leq\delta]\leq2\delta\sqrt N$.

We also use the Haar second-moment identity $\mbb{E}[X_P^2]=1/(N+1)$, from Example~51 of~\cite{mele2024haar}.

Now fix distinct qubits $i$ and $j$. Define their connected $Z$-correlation by
$C_{ij}\coloneq
\bra{\phi}Z_iZ_j\ket{\phi}
-
\bra{\phi}Z_i\ket{\phi}\bra{\phi}Z_j\ket{\phi}$.
Equation~(5) of~\cite{wolf2008arealaws}, after converting the mutual information to bits, gives
$I_\phi(i:j)\geq C_{ij}^2/(2\ln2)$.
Thus, if $I_\phi(i:j)<N^{-2}$, then
$\abs{C_{ij}}<\delta_N$, where
$\delta_N\coloneq\sqrt{2\ln2}/N$. Set
$X\coloneq\bra{\phi}Z_iZ_j\ket{\phi}$,
$U\coloneq\bra{\phi}Z_i\ket{\phi}$, and
$V\coloneq\bra{\phi}Z_j\ket{\phi}$,
so that $C_{ij}=X-UV$.
Define the good event
$G\coloneq\{\abs{U}\leq N^{-3/8},\abs{V}\leq N^{-3/8}\}$.
On $G$, $\abs{UV}\leq N^{-3/4}$. Thus, if $\abs{C_{ij}}<\delta_N$ and $G$ occurs, then
$\abs{X}\leq\abs{C_{ij}}+\abs{UV}
<\delta_N+N^{-3/4}$.
Equivalently,
\begin{equation}
\{\abs{C_{ij}}<\delta_N\}
\subseteq
\{\abs{X}\leq\delta_N+N^{-3/4}\}
\cup
G^c.
\label{eq:connected-correlation-event-split}
\end{equation}
For sufficiently large $n$, $\delta_N+N^{-3/4}\leq1$. Applying the anti-concentration bound above to the nonidentity Pauli $Z_iZ_j$ gives
$\Pr[\abs{X}\leq\delta_N+N^{-3/4}]
\leq
2(\delta_N+N^{-3/4})\sqrt N
=
O(N^{-1/4})$. For the second event,
$\mbb{E}[U^2]=\mbb{E}[V^2]=1/(N+1)$.
Markov's inequality gives
$\Pr[\abs{U}>N^{-3/8}]
=
\Pr[U^2>N^{-3/4}]
\leq
N^{3/4}\mbb{E}[U^2]
\leq
N^{-1/4}$,
and the same holds for $V$. Thus
$\Pr[G^c]\leq2N^{-1/4}$. Combining these bounds with Eq.~\ref{eq:connected-correlation-event-split} gives
$\Pr[\abs{C_{ij}}<\delta_N]=O(N^{-1/4})$,
and therefore
$\Pr[I_\phi(i:j)<N^{-2}]=O(N^{-1/4})$.
A union bound over the $\binom n2$ pairs gives
$\Pr[E_L^c]
=
O(n^2N^{-1/4})
=
O(n^22^{-n/4})
=
o(1)$.
\end{proof}
We next show that a Haar-random state satisfies the small-region mutual-information condition $E_U$ with high probability.

\begin{claim}[Small-region mutual information upper bound]
\label{clm:mut-info-upper-bound}
For all sufficiently large $n$,
$\Pr_{\phi\sim\Haar_n}[E_U^c]
\leq
\exp\pbra{-\Omega(2^n/n^2)}
=
o(1)$.
\end{claim}

\begin{proof}
Fix disjoint subsets $A,B\subseteq[n]$ with $|A|,|B|<s_n$, and let
$\rho_A$ and $\rho_B$ denote the corresponding reduced states of $\ket{\phi}$. Assume both subsets are nonempty. Let
$R\coloneq A\cup B$,
$r\coloneq|R|=|A|+|B|$, and let $\rho_R$ denote the reduced state on $R$.
Since
$S(\rho_A)\leq|A|$
and
$S(\rho_B)\leq|B|$,
we have
$I_\phi(A:B)\leq r-S(\rho_R)$.
It is therefore enough to show that every relevant $R$ satisfies
$S(\rho_R)>r-1/4$. 

For all sufficiently large $n$,
$s_n=\lfloor n/4\rfloor$, and hence
$r\leq2s_n-2\leq n/2-2$.
Corollary~VI.1 of~\cite{hayden2006aspects} implies that, for every fixed $\alpha>0$, except with probability
$2^{n-1}2^{-\Omega(2^n\alpha^2/n^2)}$,
every subsystem $R\subseteq[n]$ simultaneously satisfies
$S(\rho_R)\geq x_R-\alpha-\beta_R$,
where
$x_R\coloneq\min\{|R|,n-|R|\}$
and
$\beta_R\coloneq(\ln2)^{-1}2^{2x_R-n}$. For every subsystem relevant to $E_U$, we have
$x_R=r$
and
$\beta_R\leq1/(16\ln2)<1/8$.
Taking $\alpha=1/8$, except with probability
$\exp\pbra{-\Omega(2^n/n^2)}$,
every relevant $R$ simultaneously satisfies
$S(\rho_R)>r-1/4$.
Hence
$I_\phi(A:B)<1/4$
for every tested pair $A,B$, which proves
$\Pr[E_U^c]\leq\exp\pbra{-\Omega(2^n/n^2)}=o(1)$.
\end{proof}

\begin{lemma}[Largeness of the mutual information promise property]
\label{lem:mut-info-largeness}
The mutual information promise property of
Definition~\ref{def:mutual-info-promise-property}
satisfies the largeness condition of
Definition~\ref{def:natural-promise-property}.
\end{lemma}

\begin{proof}
Since $P_n^{yes}=E_L\cap E_U$, a union bound and
Claims~\ref{clm:mut-info-lower-bound} and~\ref{clm:mut-info-upper-bound} give
$\Pr_{\phi\sim\Haar_n}[\ket{\phi}\notin P_n^{yes}]=o(1)$.
Therefore
$\Pr_{\phi\sim\Haar_n}[\ket{\phi}\in P_n^{yes}]=1-o(1)$.
In particular, this probability is at least $1/2\geq2^{-n}$ for all sufficiently large $n$.
\end{proof}

\paragraph{Constructivity.}

To prove constructivity, we give a deterministic classical tester that decides the promise property from a sufficiently precise \amplistwords. We'll have to bound two sources of error, from truncating and renormalizing the \amplistwords, and from numerically evaluating the mutual informations of the resulting state. We will make use of the two following facts for bounding these errors.

\begin{claim}[Continuity of mutual information]
\label{clm:mutual-info-continuity}
Let $\rho$ and $\sigma$ be $n$-qubit states satisfying
$D(\rho,\sigma)\leq\eps_{\mathrm{tr}}$, where
$\eps_{\mathrm{tr}}\in[0,1/2]$ is the trace-distance error bound. Then, for every pair of nonempty disjoint subsets $A,B\subseteq[n]$,
$\abs{I_\rho(A:B)-I_\sigma(A:B)}
\leq
2\eps_{\mathrm{tr}}\pbra{|A|+|B|}
+3h_2(\eps_{\mathrm{tr}})$,
where
$h_2(x)\coloneq-x\log_2x-(1-x)\log_2(1-x)$.
\end{claim}

\begin{proof}
Partial trace cannot increase trace distance, so the reduced states on $A$, $B$, and $AB$ are each at trace distance at most $\eps_{\mathrm{tr}}$. By the Fannes--Audenaert inequality~\cite{audenaert2007sharp},
$\abs{S(\rho_R)-S(\sigma_R)}
\leq
\eps_{\mathrm{tr}}|R|+h_2(\eps_{\mathrm{tr}})$
for $R\in\{A,B,AB\}$.
Applying the triangle inequality to
$I(A:B)=S(\rho_A)+S(\rho_B)-S(\rho_{AB})$
and using that $A$ and $B$ are disjoint gives the claimed bound.
\end{proof}
\begin{claim}[Efficient numerical evaluation of mutual information]
\label{clm:efficient-mutual-info-evaluation}
Let $\omega$ be an $n$-qubit density matrix whose entries have rational real and imaginary parts of bit length at most $\ell$. Given disjoint subsets $A,B\subseteq[n]$ and an integer $r\geq1$, there is a deterministic classical algorithm that computes
$\widetilde I_\omega(A:B)$ satisfying
$\abs{\widetilde I_\omega(A:B)-I_\omega(A:B)}\leq2^{-r}$
in time $\poly(2^n,\ell,r)$.
\end{claim}

\begin{proof}
If either $A$ or $B$ is empty, the algorithm outputs zero. Otherwise, it suffices to approximate
$S(\omega_A)$, $S(\omega_B)$, and $S(\omega_{AB})$. Fix $R\in\{A,B,AB\}$. The reduced density matrix $\omega_R$ can be formed exactly by partial trace. It has dimension
$d_R=2^{|R|}\leq2^n$
and rational entries of polynomial bit length. Deterministic finite-precision algorithms approximate the eigenvalues of a rational Hermitian matrix in time polynomial in the matrix dimension, input bit length, and requested precision~\cite{sobczyk2025hermitian}. Let
$\lambda^{(R)}=(\lambda_1,\ldots,\lambda_{d_R})$
be the eigenvalues of $\omega_R$ and set
$\delta_{\mathrm{spec}}\coloneq2^{-3(n+r+2)}$.
Approximate each $\lambda_i$ to additive error at most
$\delta_{\mathrm{spec}}/d_R$.
The resulting eigenvalue vector is within
$\delta_{\mathrm{spec}}$
of $\lambda^{(R)}$ in $\ell_1$ distance. After clipping negative entries to zero and renormalizing, we obtain a probability vector
$\widehat{\lambda}^{(R)}$
within
$2\delta_{\mathrm{spec}}$
in $\ell_1$ distance, hence within
$\delta_{\mathrm{spec}}$
in total variation distance. Let
$H(p)\coloneq-\sum_i p_i\log_2p_i$
denote Shannon entropy. Since von Neumann entropy is the Shannon entropy of the eigenvalues, Fannes--Audenaert~\cite{audenaert2007sharp} gives
$\abs{H(\widehat{\lambda}^{(R)})-S(\omega_R)}
\leq
n\delta_{\mathrm{spec}}+h_2(\delta_{\mathrm{spec}})
\leq
2^{-r-3}$.
The finite sum defining
$H(\widehat{\lambda}^{(R)})$
can also be evaluated to error at most
$2^{-r-3}$
using standard multiple-precision algorithms for the logarithm~\cite{brent1976fast}. Thus each entropy is approximated to error at most
$2^{-r-2}$,
and combining the three entropy estimates gives mutual-information error less than
$2^{-r}$.
\end{proof}

\begin{lemma}[Constructivity of the mutual information promise property]
\label{lem:mut-info-constructivity}
The mutual information promise property of
Definition~\ref{def:mutual-info-promise-property}
satisfies the constructivity condition of
Definition~\ref{def:natural-promise-property}.
\end{lemma}

\begin{proof}
The two promise gaps are
$\Delta_{\alpha,n}
=
\alpha_n^{yes}-\alpha_n^{no}
=
2^{-2n-1}$
and
$\Delta_\beta
=
\beta_n^{no}-\beta_n^{yes}
=
1/4$.
Let
$\Delta_n\coloneq2^{-2n-1}$
denote the smaller gap, and define the midpoint thresholds
$m_{\alpha,n}\coloneq3\cdot2^{-2n-2}$
and
$m_\beta\coloneq3/8$. If either $A$ or $B$ is empty, then $I_\psi(A:B)=0$, so the pair automatically satisfy the regional condition and do not need to be tested. 

Given $1^n$ and a truncated \amplistwords $\al_T(\psi)$, the tester proceeds as follows.
\begin{enumerate}
    \item Interpret $\al_T(\psi)$ as a vector $v\in\CC^{2^n}$ and let
    $q\coloneq\norm{v}_2^2$.
    If $q=0$, reject. Otherwise define
    $\widehat{\rho}\coloneq vv^\dagger/q$.

    \item For every pair of distinct qubits $i,j$ and every pair of nonempty disjoint subsets $A,B\subseteq[n]$ with $|A|,|B|<s_n$, compute approximations
    $\widetilde I_{\widehat{\rho}}(i:j)$ and
    $\widetilde I_{\widehat{\rho}}(A:B)$.

    \item Accept if
    $\widetilde I_{\widehat{\rho}}(i:j)\geq m_{\alpha,n}$
    for every $i\neq j$ and
    $\widetilde I_{\widehat{\rho}}(A:B)\leq m_\beta$
    for every tested pair $A,B$. Reject otherwise.
\end{enumerate}

We first bound the error caused by truncating the input state. By Definition~\ref{def:state-amp-list},
$\norm{v-\ket{\psi}}_2
\leq
2^{(n+1)/2-T}
\eqqcolon
\delta_T$.
The reverse triangle inequality gives
$\abs{1-\norm{v}_2}\leq\delta_T$.
Hence, for
$\widehat{\ket{\psi}}\coloneq v/\norm{v}_2$,
$\norm{\widehat{\ket{\psi}}-\ket{\psi}}_2
\leq2\delta_T$.
Let
$\eps_T\coloneq2\delta_T=2^{(n+3)/2-T}$.
Then $\widehat{\rho}$ is the density matrix corresponding to
$\widehat{\ket{\psi}}$, and
$D(\widehat{\rho},\ket{\psi}\!\bra{\psi})\leq\eps_T$. For $\eps_T\leq1/2$, Claim~\ref{clm:mutual-info-continuity} gives, for every tested pair,
$\abs{I_{\widehat{\rho}}(A:B)-I_\psi(A:B)}
\leq
4n\eps_T+3h_2(\eps_T)$.
Set
$\tau(n)\coloneq4n+11$.
For every $T\geq\tau(n)$, using
$h_2(x)\leq2x\log_2(1/x)$
and
$\Delta_n=2^{-2n-1}$ gives $\frac{4n\eps_T+3h_2(\eps_T)}{\Delta_n}
\leq
2^{-3n/2-17/2}(25n+57)
\leq
\frac18$. Thus the error due to truncating the input state satisfies
$\abs{I_{\widehat{\rho}}(A:B)-I_\psi(A:B)}
\leq\Delta_n/8$. For the numerical evaluation error, choose
$r_n\coloneq2n+4$,
so that
$2^{-r_n}=\Delta_n/8$.
The entries of
$\widehat{\rho}=vv^\dagger/q$
have rational real and imaginary parts of bit length polynomial in $T$ and $n$, so Claim~\ref{clm:efficient-mutual-info-evaluation} gives
$\abs{\widetilde I_{\widehat{\rho}}(A:B)-I_{\widehat{\rho}}(A:B)}
\leq\Delta_n/8$.
Combining the truncation and numerical evaluation errors by the triangle inequality, every estimate used by the tester satisfies
$\abs{\widetilde I_{\widehat{\rho}}(A:B)-I_\psi(A:B)}
\leq\Delta_n/4$.

Since $\Delta_n$ is the smaller of the two promise gaps, this error is strictly less than half the distance between the yes- and no-thresholds in either condition. Thus, if $\ket{\psi}\in P_n^{yes}$, every pairwise estimate remains above $m_{\alpha,n}$ and every regional estimate remains below $m_\beta$. If $\ket{\psi}\in P_n^{no}$, a pair or pair of subsets witnessing this remains on the rejecting side of the corresponding midpoint. Therefore the tester accepts every state in $P_n^{yes}$ and rejects every state in $P_n^{no}$.

Finally, let $L$ be the total input length. Since the \amplistwords contains $2^n$ complex amplitudes represented to $T$ bits of precision,
$L=\Theta(n+2^nT)$.
There are $O(n^2)$ qubit pairs and at most $3^n$ ordered disjoint subset pairs. Each mutual information can be evaluated in
$\poly(2^n,T,n)$
time by Claim~\ref{clm:efficient-mutual-info-evaluation}. Since
$3^n\leq(2^n)^2$
and
$2^n,T,n=O(L)$,
the total running time is $\poly(L)$.
\end{proof}

\paragraph{Usefulness.} We now apply the lower bound from \cite[Theorem~3.6]{parham2025mh} to show that the mutual information promise property is useful against clean-output $\AxCQ{1}$ state preparation.
\begin{lemma}[Usefulness of the mutual information promise property]
\label{lem:mutual-info-usefulness}
The mutual information promise property of
Definition~\ref{def:mutual-info-promise-property}
is $\varepsilon_{\mathrm{MI}}$-useful against clean-output $\AxCQ{1}$ state preparation, where
$\varepsilon_{\mathrm{MI}}(n)\coloneq2^{-4n-10}$.
\end{lemma}

\begin{proof}
Let $\{\ket{\psi_n}\}$ be a pure-state family that is
$\varepsilon_{\mathrm{MI}}$-approximately cleanly preparable by $\AxCQ{1}$.
For each $n$, let $C_n=U_nK_n$ be a clean approximating circuit, where $K_n$ is a Clifford circuit and $U_n$ is a $\QNC^0$ circuit of depth $d_n\leq d$ for some constant $d$ independent of $n$. Thus, for some number $a_n$ of work qubits,
$C_n\ket{0^{n+a_n}}
=
\ket{\phi_n}\ket{0^{a_n}}$,
and
$D(\ket{\psi_n}\!\bra{\psi_n},\ket{\phi_n}\!\bra{\phi_n})
\leq\varepsilon_{\mathrm{MI}}(n)$.

Suppose, for contradiction, that
$\ket{\psi_n}\notin P_n^{no}$
for infinitely many $n$. For every such $n$, Definition~\ref{def:mutual-info-promise-property} gives
$I_{\psi_n}(i:j)>2^{-2n-1}$
for every $i\neq j$, and
$I_{\psi_n}(A:B)<1/2$
for every pair of disjoint subsets $A,B\subseteq[n]$ with
$|A|,|B|<s_n$. 
We apply the $a=0$ case of \cite[Theorem~3.6]{parham2025mh} with
$k=0$, $\alpha=2^{-2n-1}$, $\beta=1/2$, and $s=s_n$.
The original $a_n$ work qubits are the qubits required to be returned to $\ket{0}$.

Let
$\delta_n\coloneq
\norm{\ket{\psi_n}\!\bra{\psi_n}-\ket{\phi_n}\!\bra{\phi_n}}_1$.
Since our approximation error is measured in trace distance,
$\delta_n\leq2\varepsilon_{\mathrm{MI}}(n)=2^{-4n-9}$.
The resulting lower bound implies that at least one of
$\delta_n>\left(\frac{2^{-2n-1}}{1+3e}\right)^{\ln4}$,
$d_n\geq\frac14\log s_n$,
or
$d_n=\Omega\pbra{\log(1/\delta_n)}$
must hold. For the first alternative,
$1+3e<16$ and $\ln4<3/2$ give
$\left(2^{-2n-1}/(1+3e)\right)^{\ln4}
>
2^{-3n-8}
>
2^{-4n-9}
\geq\delta_n$.
The second alternative fails for sufficiently large $n$ because
$s_n=\Theta(n)$ while $d_n\leq d$ is bounded.
If $\delta_n>0$, then
$\delta_n\leq2^{-4n-9}$
implies
$\log(1/\delta_n)\geq4n+9$,
so the third alternative also fails for sufficiently large $n$.
The case $\delta_n=0$ is ruled out by the exact-preparation argument preceding Theorem~3.6 in~\cite{parham2025mh}. Therefore
$\ket{\psi_n}\in P_n^{no}$
for all sufficiently large $n$, proving the required eventual membership and hence
$\varepsilon_{\mathrm{MI}}$-usefulness.
\end{proof}

Thus combining \Cref{lem:mut-info-largeness}, \Cref{lem:mut-info-constructivity}, and \Cref{lem:mutual-info-usefulness}, we have proved \Cref{clm:mutual-info-main-claim}.
This example is formulated in the clean-output preparation model, which is required for the usefulness argument above. A corresponding natural-proofs barrier therefore uses a \PRAL family that is also preparable in the clean-output model.

\subsection{Discussion and non-examples}

\paragraph{Non-example: Exact state lower bounds via transcendental numbers} 
Jia and Wolf construct explicit states and unitaries that require exponential-size circuits to prepare/implement \emph{exactly} \cite{jia2023hay}. Here we explain their proof technique and explain why it is not natural. 
They consider a state that has a large set of algebraically independent numbers: a collection of numbers is algebraically
independent if it satisfies no nonzero polynomial relation with
rational coefficients. 
Since amplitudes of a state are some polynomial expression in the entries of the individual gates, in order to have exponentially many independent algebraically independent numbers, we must have exponentially many gates. 

Their proof does not fit into our framework as a natural proof because the algebraic property of containing many algebraically-independent entries is \emph{not constructible}. Exact algebraic relations cannot, in general, be determined from numerical approximations. With only finite precision, we cannot even tell if an entry is irrational, regardless of runtime. 

\paragraph{Finding a natural property containing the nonnatural property}
In the example in \Cref{ssec:example-approximate-pauli-degree}, high felinity itself is not a natural property because it does not satisfy the largeness condition. Nevertheless by \Cref{lem:felinity-approximate-degree}, high felinity implies another property, high approximate Pauli degree, which does. Thus the lower bound for high-felinity states can still be viewed as a natural proof, using high approximate Pauli degree as the certificate.

 \providecommand{\unittodo}[1]{\par\noindent\textbf{[TODO: #1]}\par}

\section{Natural proofs for unitary implementation}
\label{sec:unitary-natural-proofs}
In this section we extend our natural-proofs framework to lower bounds
for unitary implementation. The abstract barrier carries over directly:
pseudorandom finite-precision matrix descriptions supported on
efficiently implementable unitaries would rule out natural unitary
lower-bound properties. Unlike in the state-preparation setting,
however, we do not instantiate this pseudorandom object from a standard
cryptographic assumption.

We show below that attempting to reproduce the statistical random-function stage of our state
construction for full unitary matrices would already imply a positive solution to the unitary synthesis problem, via the relaxed-net argument of Brakerski and Yuen~\cite{brakerski2026scalablePRU}. We therefore treat the unitary barrier conditionally, give an example of a known unitary lower-bound technique that is natural, and discuss its relation to scalable pseudorandom unitaries and unitary synthesis.

\subsection{Natural unitary properties and a conditional barrier}
\label{ssec:unitary-natural-definition}

Write $N=2^n$, and let $\matdesc_T(U)$ denote the full $N\times N$
matrix obtained by truncating the real and imaginary parts of every
entry of $U$ to $T$ fractional bits. When supplied to a classical
algorithm, this matrix is represented using a fixed standard binary
encoding of length $O(N^2(T+1))$.

\begin{definition}[Natural unitary promise property]
\label{def:unitary-natural-property}
A unitary promise property consists of disjoint sets
$P_n^{\mathrm{yes}},P_n^{\mathrm{no}}\subseteq\mathrm U(2^n)$,
with both sets invariant under global phase. It is \emph{natural}
if it satisfies:
\begin{enumerate}
    \item \emph{Constructivity.} There is a deterministic classical
    algorithm and a polynomially bounded precision threshold
    $\tau(n)$ such that, for every $T\geq\tau(n)$, the algorithm
    accepts $\matdesc_T(U)$ for every YES-unitary and rejects it
    for every NO-unitary, in time polynomial in the total input length.

    \item \emph{Largeness.} The YES-set is Haar-measurable and has
    probability at least $2^{-cn}$ under Haar measure on
    $\mathrm U(2^n)$, for some constant $c>0$ and all sufficiently
    large $n$.
\end{enumerate}
Outside the two promise sets the algorithm's behavior is unrestricted.
\end{definition}

The barrier we state below is agnostic to the particular notion of approximation for implementing the unitary. Fix any notion of $\epsilon$-approximation that is invariant under global phase, and fix the notion of implementation of an $n$-qubit unitary by a circuit in $\mathcal C$, e.g. whether ancillas are allowed. A natural unitary property is $\epsilon$-useful relative to this implementation notion if every unitary family approximately implementable by polynomial-size circuits in $\mathcal C$ belongs to the NO-set for all sufficiently large $n$.

The unitary analogue of a \PRAL family is a family with pseudorandom
finite-precision matrix descriptions in the following sense.

\begin{definition}[Pseudorandom finite-precision matrix descriptions]
\label{def:pseudorandom-unitary-matrix-descriptions}
A family
$\mathcal D^{\mathrm U}
=\{D^{\mathrm U}_{n,\lambda,T}\}_{n,\lambda,T\in\NN}$
of distributions on $\mathrm U(2^n)$ has
\emph{pseudorandom finite-precision matrix descriptions} if, for every
probabilistic polynomial-time classical distinguisher $A$ and every
pair of polynomially bounded functions
$\lambda,T:\NN\to\NN$,
\begin{align*}
\abs{
    \Pr_{U\gets D^{\mathrm U}_{n,\lambda(n),T(n)}}
    \sbra{
        A\pbra{
            1^n,1^{\lambda(n)},\matdesc_{T(n)}(U)
        }=1
    }
    -
    \Pr_{V\sim\operatorname{Haar}(\mathrm U(2^n))}
    \sbra{
        A\pbra{
            1^n,1^{\lambda(n)},\matdesc_{T(n)}(V)
        }=1
    }
}
\leq 2^{-\lambda(n)}
\end{align*}
for all sufficiently large $n$.
The running time of $A$ is measured with respect to its total input
length. Efficient sampling from $D^{\mathrm U}_{n,\lambda,T}$ is not
required.
\end{definition}

We say that such a family is $\eps$-approximately \emph{efficiently implementable} if there there exists some polynomial $p$ such that every unitary in the support of $D^{\mathrm U}_{n,\lambda,T}$ can be $\eps$-approximately implemented by a circuit of size at most $p(n, \lambda, T)$.

\begin{proposition}[Conditional unitary natural-proofs barrier]
\label{prop:unitary-natural-barrier}
Fix a notion of $\epsilon$-approximate unitary implementation. If a
family with pseudorandom finite-precision matrix descriptions is efficiently
$\epsilon$-approximately implementable by polynomial-size circuits in
$\mathcal C$, then no natural unitary property is $\epsilon$-useful against
$\mathcal C$ relative to this implementation notion.
\end{proposition}

\begin{proof}
The recognition algorithm for a useful property rejects every
supported unitary, but accepts a Haar-random unitary with the
probability guaranteed by largeness. Choosing the security guarantee
for these descriptions stronger than this acceptance gap contradicts
pseudorandomness, exactly as in the state-preparation argument.
\end{proof}

The underlying possibility was raised by Aaronson, who asked whether
polynomial-size quantum circuits can give rise to matrices that look
Haar-random even to exponential-time classical algorithms with access
to their entries~\cite[Section~3.2]{aaronson2016barbados}.
This proposition isolates the pseudorandom object needed for such a
barrier. Unlike our main state-preparation theorem, it does not by
itself construct that object from a cryptographic assumption.

\subsection{Example: Pauli spectrum concentration}
\label{ssec:unitary-natural-examples}
Nadimpalli, Vasconcelos, Parham and Yuen \cite{nadimpalli2024pauli} showed that for $\QAC^0$ circuits with sublinear ancillas, the Pauli spectrum of their single output observables is concentrated on low-degree Pauli strings. For an $n$-qubit  unitary $U$, let $Z_1$ denote the Pauli $Z$ operator acting on the first qubit, and consider the conjugated observable $A_U=U^\dagger Z_1U$. The Pauli expansion of $A_U$ is
\[
 A_U=U^\dagger Z_1U=\sum_{P\in\mathcal P_n}a_P(U)P,
 \qquad
 a_P(U)=\frac{1}{N}\operatorname{Tr}(P U^\dagger Z_1U).
\]
The \emph{Pauli weight above degree $k$} is defined as
\[
 W_{>k}(U)=\sum_{\operatorname{wt}(P)>k}|a_P(U)|^2.
\]
Nadimpalli, Vasconcelos, Parham and Yuen show that for $\QAC^0$ circuits with sublinear ancillas, $W_{>k}(U)$ is exponentially small in $n$ for $k=\lfloor n/2\rfloor$. This gives a natural certificate for unitaries that cannot be implemented by such circuits.

\emph{Property:} Let $P^{\mathrm{yes}}_n$ be the set of unitaries with $W_{>n/2}(U)\geq 1/2$, and $P^{\mathrm{no}}_n$ be the set of unitaries with $W_{>n/2}(U)\leq 1/4$. 

\emph{Constructivity:} $W_{>n/2}(U)$ can be computed from the full matrix description of $U$, up to the precision of the description. This is because we can for each $P \in \calP_n$ compute $a_P(U)$ from the matrix description of $U$ and then sum over all $P$ with weight greater than $n/2$.

\emph{Largeness.}
When $U$ is Haar-random, the expected value of the mass on each of the non-identity Paulis is the same $\Expect_U\abs{a_P(U)}^2 = \frac{1}{4^n -1 }$.  Therefore the expected Pauli mass on weights at most $k$ is $\sum_{j=1}^k\binom{n}{j} 3^j/(4^n-1) \leq e^{-\Omega(n)}$. Since the total Pauli
mass is one, Markov's inequality implies
\[
    \Pr_{U\sim\mathrm{Haar}}
    \left[W_{>k}(U)\geq\frac12\right]
    \geq 1-e^{-\Omega(n)}.
\]

\subsection{Implications for unitary synthesis}
\label{ssec:unitary-synthesis}

Our state construction first produces statistically Haar-like
amplitude lists using a random function, and then replaces that
function by a PRF. A corresponding construction for unitary
matrices would have an additional consequence: a resolution to the Aaronson-Kuperberg unitary synthesis problem \cite{aaronson2007quantum}. 

Brakerski and Yuen~\cite{brakerski2026scalablePRU} show that
scalable ROM-PRUs imply unitary synthesis. The following
observation derives the same conclusion directly from statistical
closeness of finite-precision matrix descriptions.

\begin{proposition}[Statistical matrix descriptions imply
oracle unitary synthesis]
\label{prop:statistical-matrices-synthesis}
Suppose that, for each $n$, a uniform quantum oracle algorithm
$\mathcal A_n$ uses polynomially many queries to a Boolean oracle
$f:\zo^{m(n)}\to\zo$, where $m(n)=\poly(n)$.
Assume that, for every $f$, its induced channel approximates
an $n$-qubit unitary $U_f$ to diamond-norm error at most $1/100$.

Let $D_n$ be the distribution of $U_f$ for uniformly random $f$.
If, at precision $T=n+10$,
\[
 \TVD\bigl(\matdesc_T(D_n),
           \matdesc_T(\Haar_n^{\mathrm U})\bigr)
 \leq\frac14,
\]
where $\Haar_n^{\mathrm U}$ is Haar measure on $\mathrm U(2^n)$,
then unitary synthesis is possible with polynomial query
complexity, polynomial Boolean-oracle input length, and
diamond-norm error at most $1/10$.
\end{proposition}

\begin{proof}
Let $\mathcal S=\{U_f:f:\zo^{m(n)}\to\zo\}$, and let
$\mathcal B$ be the set of its truncated matrix descriptions.
Since $\matdesc_T(D_n)$ is supported on $\mathcal B$,
statistical closeness gives
\[
 \Pr_{U\sim\Haar_n^{\mathrm U}}
 [\matdesc_T(U)\in\mathcal B]\geq\frac34.
\]
Two unitaries with the same truncated description satisfy
\[
 \|U-V\|_{\mathrm{op}}
 \leq\|U-V\|_{\mathrm F}
 \leq2\sqrt2\,2^{n-T}<\frac1{100}.
\]
Thus a Haar-random unitary is within operator-norm distance
$1/100$ of $\mathcal S$ with probability at least $3/4$.

Following the relaxed-net composition argument
of~\cite[Corollary~3.5]{brakerski2026scalablePRU}, fix any
target unitary $W$ and sample Haar-random $U$.
Both $U^\dagger$ and $UW$ are Haar distributed, so with positive
probability both are within $1/100$ of $\mathcal S$.
There are therefore $V_1,V_2\in\mathcal S$ satisfying
\[
 \|V_1V_2-W\|_{\mathrm{op}}
 \leq\|V_1-U^\dagger\|_{\mathrm{op}}
     +\|V_2-UW\|_{\mathrm{op}}
 \leq\frac2{100}.
\]

Combine the Boolean oracles corresponding to $V_1,V_2$
using a selector bit, and run their implementations in the
order that implements $V_1V_2$.
Operator-norm error $\delta$ between unitaries gives
diamond-norm error at most $2\delta$ between their channels.
Including the two implementation errors, the resulting
channel is within $4/100+2/100<1/10$ of the channel of $W$.
The construction uses polynomially many queries to a Boolean
oracle of input length $m(n)+1$.
\end{proof}

This argument relies essentially on statistical closeness, since membership in $\mathcal B$ need not be efficiently decidable. It therefore shows that the statistical random-function strategy would imply unitary synthesis, not that unitary synthesis follows from our computational matrix-description criterion.

\printbibliography
\end{document}